%% file: LowrankNonsymmetric-Final-Correction.tex
\documentclass[journal]{IEEEtran}

\usepackage[english]{babel}
\input{macro}

\title{The Global Optimization Geometry \\of Low-Rank Matrix Optimization}
\author{Zhihui Zhu, Qiuwei Li, Gongguo Tang, and Michael B. Wakin
\thanks{The first two authors contribute equally. This work was supported by NSF grant CCF-1409261, NSF grant CCF-1464205, NSF grant 2008460, NSF CAREER grant CCF-1149225, and Award N660011824020 from the DARPA Lagrange Program. ZZ is with the Department of Electrical and  Computer Engineering, University of Denver, QL is with the Department of Mathematics, University of California, Los Angeles, and GT and MW  are with the  Department of Electrical Engineering, Colorado School of Mines. Email: zhihui.zhu@du.edu,  qiuweili@math.ucla.edu, \{gtang, mwakin\}@mines.edu.}
}

\begin{document}

\maketitle

\begin{abstract}
This paper considers general rank-constrained optimization problems that minimize a general objective function $f(\mX)$ over the set of rectangular $n\times m$ matrices that have rank at most $r$. To tackle the rank constraint and also to reduce the computational burden, we factorize $\mX$ into $\mU\mV^\T$ where $\mU$ and $\mV$ are $n\times r$ and $m\times r$ matrices, respectively, and then optimize over the small matrices $\mU$ and $\mV$. We characterize the global optimization geometry of the nonconvex factored problem and show that the corresponding objective function satisfies the robust strict saddle property as long as the original objective function $f$ satisfies restricted strong convexity and smoothness properties, ensuring global convergence of many local search algorithms (such as noisy gradient descent) in polynomial time for solving the factored problem. We also provide a comprehensive analysis for the optimization geometry of a matrix factorization problem where we aim to find $n\times r$ and $m\times r$ matrices $\mU$ and $\mV$ such that $\mU\mV^\T$ approximates a given matrix $\mX^\star$. Aside from the robust strict saddle property, we show that the objective function of the matrix factorization problem has no spurious local minima and obeys the strict saddle property not only for the exact-parameterization case where $\rank(\mX^\star) = r$, but also for the over-parameterization case where $\rank(\mX^\star) < r$ and the under-parameterization case where $\rank(\mX^\star) > r$. These geometric properties imply that a number of iterative optimization algorithms (such as gradient descent) converge to a global solution with random initialization.
\end{abstract}


\section{Introduction}
Low-rank matrices arise in a wide variety of applications throughout science and engineering, ranging from quantum tomography~\cite{aaronson2007learnability}, signal processing~\cite{liu2009interior},
machine learning~\cite{srebro2004maximum,xu2016dynamic}, and so on; see~\cite{davenport2016overview} for a comprehensive review. In all of these settings, we often encounter the following rank-constrained optimization problem:
\begin{equation}\begin{split}
& \minimize_{\mX\in\R^{n\times m}} f(\mX),\\
& \st \rank(\mX)\leq r,
\label{eq:general problem}\end{split}\end{equation}
where the objective function $f:\R^{n\times m}\rightarrow \R$ is smooth.

Whether the objective function $f$ is convex or nonconvex, the rank constraint renders  low-rank matrix optimizations of the form~\eqref{eq:general problem} highly nonconvex and computationally NP-hard in general~\cite{fazel2004rank}. Significant efforts have been devoted to transforming~\eqref{eq:general problem} into a convex problem by replacing the rank constraint with one involving the nuclear norm. This strategy has been widely utilized in matrix inverse problems~\cite{recht2010guaranteed} arising in signal processing \cite{davenport2016overview}, machine learning \cite{harchaoui2012large}, and control \cite{fazel2004rank}. With convex analysis techniques, nuclear norm  minimization has been proved to provide optimal performance in recovering low-rank matrices~\cite{candes2009exact}. However, in spite of the optimal performance,
solving nuclear norm minimization is very computationally expensive even with specialized first-order algorithms. For example, the singular value thresholding  algorithm \cite{cai2010singular} requires performing an expensive singular value decomposition (SVD) in each iteration, making it computationally prohibitive in large-scale settings. This prevents nuclear norm minimization from scaling to practical problems.

To relieve the computational bottleneck, recent studies propose to factorize the variable into $\mX =\mU\mV^\T$, and optimize over the $n\times r$ and $m\times r$ matrices $\mU$ and $\mV$ rather than the $n\times m$ matrix $\mX$. The rank constraint in~\eqref{eq:general problem} then is automatically satisfied through the factorization. This strategy is usually referred to as the Burer-Monteiro type decomposition after the authors in~\cite{burer2003nonlinear,burer2005local}. Plugging this parameterization of $\mX$ in~\eqref{eq:general problem}, we can recast the program into the following one:
\begin{align}
\minimize_{\mU\in\R^{n\times r},\mV\in\R^{m\times r}} h(\mU,\mV):=f(\mU\mV^\T).
\label{eq:factored problem no regularizer}\end{align}
The bilinear nature of the parameterization renders the objective function of \eqref{eq:factored problem no regularizer} nonconvex. Hence, it can  potentially have spurious local minima (i.e., local minimizers that are not global minimizers) or even saddle points. With technical innovations in analyzing the landscape of nonconvex functions, however, several recent works have shown that the factored objective function $h(\mU,\mV)$ in certain matrix inverse problems has no spurious local minima~\cite{bhojanapalli2016lowrankrecoveryl,ge2016matrix,park2016non}.

\subsection{Summary of results and outline}
In this paper,  we provide a comprehensive geometric analysis for solving general low-rank optimizations of the form~\eqref{eq:general problem} using the factorization approach \eqref{eq:factored problem no regularizer}. Our work actually rests on the recent works~\cite{ge2015escaping,sun2015nonconvex,lee2016gradient,panageas2016gradient,jin2017escape} ensuring a number of iterative optimization methods (such as gradient descent) converge to a local minimum with random initialization provided the problem satisfies the so-called strict saddle property (see Definition~\ref{def:strict saddle property} in Section~\ref{sec:Preliminaries}). If the objective function further obeys the robust strict saddle property~\cite{ge2015escaping} (see Definition~\ref{def:robust strict saddle} in Section~\ref{sec:Preliminaries}) or belongs to the class of so-called $\calX$ functions~\cite{sun2015nonconvex}, the recent works~\cite{ge2015escaping,sun2015nonconvex} show that many local search algorithms can converge to a local minimum in polynomial time. The implications of this line of work have had a tremendous impact on a number of nonconvex problems in applied mathematics, signal processing, and machine learning.

We begin this paper in Section~\ref{sec:Preliminaries} with the notions of strict saddle, strict saddle property, and robust strict saddle property. Considering that many invariant functions are not strongly convex (or even convex) in any neighborhood around a local minimum point, we then provide a revised robust strict saddle property\footnote{A similar notion of a revised robust strict saddle property has also been utilized in \cite{jin2017escape}, which shows that noisy gradient descent converges to a local minimum in a number iterations that depends only poly-logarithmically on the dimension. In a nutshell, \cite{jin2017escape} has a different focus than this work: the focus in \cite{jin2017escape} is on providing convergence analysis of a noisy gradient descent algorithm with a robust strict saddle property, while in the present paper, we establish a robust strict saddle property for the nonsymmetric matrix factorization and more general low-rank optimization (including matrix sensing) problems with the factorization approach.} requiring a regularity condition (see Definition~\ref{def:regularity condition} in Section~\ref{sec:Preliminaries}) rather than strong convexity  near the local minimum points (which is one of the requirements for the strict saddle property). The stochastic gradient descent algorithm is guaranteed to converge to a local minimum point in polynomial time for problems satisfying the revised robust strict saddle property~\cite{ge2015escaping,jin2017escape}.

In Section~\ref{sec:low-rank optimization}, we consider the geometric analysis for solving general low-rank optimizations of the form~\eqref{eq:general problem} using the factorization approach \eqref{eq:factored problem no regularizer}. Provided the objective function $f$ satisfies certain restricted strong convexity and smoothness conditions, we show that the low-rank optimization problem with the factorization~\eqref{eq:factored problem no regularizer} (with an additional regularizer---see Section~\ref{sec:low-rank optimization} for the details) obeys the revised robust strict saddle property.  In Section~\ref{sec:matrix sensing}, we consider a stylized application in matrix sensing where the measurement operator satisfies the restricted isometry property (RIP)~\cite{recht2010guaranteed}. In the case of Gaussian measurements, as guaranteed by this robust strict saddle property, a number of iterative optimizations can find the unknown matrix $\mX^\star$ of rank $r$ in polynomial time with high probability when the number of measurements exceeds a constant times $(n+m)r^2$. 

Our main approach for analyzing the optimization geometry of \eqref{eq:factored problem no regularizer} is based on the
geometric analysis for  the following non-square low-rank matrix factorization problem: given $\mX^\star \in \R^{n \times m}$,
\begin{equation}
\minimize_{\mU\in\R^{n\times r},\mV^{m\times r}} \left\|\mU\mV^\T - \mX^\star\right\|_F^2.
\label{eq:low rank fact no regu}\end{equation}
In particular, we show the optimization geometry for the low-rank matrix factorization problem~\eqref{eq:low rank fact no regu} is preserved for the general low-rank optimization \eqref{eq:factored problem no regularizer} under certain restricted strong convexity and smoothness conditions on $f$. Thus, in Appendix~\ref{sec:low rank factorization}, we provide a comprehensive geometric analysis for \eqref{eq:low rank fact no regu}, which can be viewed as an important foundation of many popular matrix factorization problems such as the matrix sensing problem and matrix completion. We show that the low-rank matrix factorization problem~\eqref{eq:low rank fact no regu} (with an additional regularizer) has no spurious local minima and obeys the strict saddle property---that is  the objective function in~\eqref{eq:low rank fact no regu} has a directional negative curvature at all critical points but local minima---not only for the exact-parameterization case where $\rank(\mX^\star) =r$, but also for the over-parameterization case where $\rank(\mX^\star) <r$ and the under-parameterization case where $\rank(\mX^\star) >r$. The strict saddle property and lack of spurious local minima ensure that a number of local search algorithms applied to the matrix factorization problem~\eqref{eq:low rank fact no regu} converge to global optima which correspond to the best rank-$r$ approximation to $\mX^\star$. Further, we completely analyze the low-rank matrix factorization problem~\eqref{eq:low rank fact no regu} for the exact-parameterization case and show that it obeys the revised robust strict saddle property.

\subsection{Relation to existing work}
Unlike the objective functions of convex optimizations that have simple landscapes,
such as where all local minimizers are global ones, the objective functions of general nonconvex programs have much more complicated landscapes. In recent years, by exploiting the underlying optimization geometry, a surge of progress has been made in providing theoretical justifications for matrix factorization problems such as~\eqref{eq:factored problem no regularizer} using a number of previously heuristic algorithms (such as alternating minimization \cite{li2019alternating}, 
gradient descent, and the trust region method). Typical examples include phase retrieval \cite{sun2016geometric,candes2015Wirtinger,chen2015solving}, blind deconvolution \cite{lee2016fast,li2016rapid}, dictionary learning \cite{agarwal2014learning,sun2016complete1,sun2015complete}, phase synchronization~\cite{liu2016estimation} and
matrix sensing and completion \cite{tu2015low,wang2016unified,sun2016guaranteed,ge2016matrix,jin2016provable,jain2013low,ge2017no}.

These iterative algorithms can be sorted into two categories based on whether a good initialization is required. One set of algorithms consist of two steps: initialization and local refinement. Provided the function satisfies a regularity condition or similar properties, a good guess lying in the attraction basin of the global optimum can lead to global convergence of the following iterative step. We can obtain such initializations by spectral methods for phase retrieval~\cite{candes2015Wirtinger}, phase synchronization~\cite{liu2016estimation} and low-rank matrix recovery problems~\cite{tu2015low,bhojanapalli2015dropping,zhao2015nonconvex,wang2016unified}. As we have mentioned, a regularity condition is also adopted in the revised robust strict saddle property.

Another category of works attempt to analyze the landscape of the objective functions in a larger space rather than the regions near the global optima. We can further separate these approaches into two types based on whether they involve the strict saddle property or the robust strict saddle property. The strict saddle property and lack of spurious local minima are proved for low-rank, positive semidefinite (PSD) matrix recovery~\cite{bhojanapalli2016lowrankrecoveryl} and completion~\cite{ge2016matrix},  PSD matrix optimization problems with generic objective functions~\cite{li2016}, low-rank non-square matrix estimation from linear observations~\cite{park2016non}, low-rank nonsquare optimization problems with generic objective functions~\cite{zhu2017GlobalOptimality} and generic nuclear norm regularized problems \cite{li2016}.
The strict saddle property along with  the lack of spurious local minima ensures a number of iterative algorithms such as gradient descent~\cite{ge2015escaping} and the trust region method~\cite{conn2000trust}  converge to the global minimum with random initialization~\cite{lee2016gradient,ge2015escaping,sun2015complete}.

A few other works which are closely related to our work attempt to study the {\em global geometry} by characterizing  the landscapes of the objective functions in the whole space rather than the regions near the global optima or all the critical points. As we discussed before, a number of local search algorithms are guaranteed to find a local optimum (which is also the global optimum if there are no spurious local minima) because of this robust strict saddle property. In~\cite{ge2015escaping}, the authors proved that tensor decomposition problems satisfy this robust strict saddle property. Sun et al.~\cite{sun2016geometric} studied the global geometry of the phase retrieval problem. The very recent work in~\cite{li2019symmetry} analyzed the global geometry for PSD low-rank matrix factorization of the form~\eqref{eq:low rank fact no regu} and the related matrix sensing problem when the rank is exactly parameterized (i.e., $r = \rank(\mX^\star)$). The factorization approach for matrix inverse problems with quadratic loss functions is considered in~\cite{ge2017no}.  We extend this line by considering general rank-constrained
optimization problems including a set of matrix inverse problems.

Finally, we remark that our work is also closely related to the recent works in low-rank matrix factorization of the form~\eqref{eq:low rank fact no regu} and its variants~\cite{tu2015low,sun2016guaranteed,bhojanapalli2016lowrankrecoveryl,ge2016matrix,li2019symmetry,wang2016unified,park2016non,ge2017no,zhu2017GlobalOptimality}. As we discussed before, most of these works except~\cite{li2019symmetry,ge2017no} (but including~\cite{park2016non} which also focuses on nonsymmetric matrix sensing) only characterize the geometry either near the global optima or all the critical points. Instead, we characterize the {\em global}geometry for general (rather than PSD) low-rank matrix factorization and sensing. Because the analysis is different, the proof strategy in the present paper is also very different than that of~\cite{park2016non,zhu2017GlobalOptimality}. The results for PSD matrix sensing in \cite{li2019symmetry} build heavily on the concentration properties of Gaussian measurements, while our results for matrix sensing depend on the RIP of the measurement operator and thus can be applied to other matrix sensing problems whose measurement operator is not necessarily from a Gaussian measurement ensemble. Also, \cite{ge2017no} considers matrix inverse problems with quadratic loss functions and its proof strategy is very different than that in the present paper: the proof in \cite{ge2017no} is specified to quadratic loss functions, while we consider  the rank-constrained optimization problem with general objective functions in \eqref{eq:general problem} and our proof utilizes the fact that the gradient and Hessian of the low-rank matrix sensing are respectively very close to those in low-rank matrix factorization. Furthermore, in terms of the matrix factorization, we show that the objective function in~\eqref{eq:low rank fact no regu} obeys the strict saddle property and has no spurious local minima not only for exact-parameterization ($r = \rank(\mX^\star)$), but also for over-parameterization ($r > \rank(\mX^\star)$) and under-parameterization ($r < \rank(\mX^\star)$). Local (rather than global) geometry results for  exact-parameterization and under-parameterization are also covered in \cite{zhu2017GlobalOptimality}. As noted above, the work in \cite{li2019symmetry,ge2017no} for low-rank matrix factorization only focuses on exact-parameterization ($r = \rank(\mX^\star)$).
The under-parameterization implies that we can find the best rank-$r$ approximation to $\mX^\star$ by many efficient iterative optimization algorithms such as gradient descent.

\subsection{Notation}
\label{sec:notation}

Before proceeding, we first briefly introduce some notation used throughout the paper. The symbols $\mId$ and $\mzero$ respectively represent the identity and zero matrices with appropriate sizes. Also $\mId_n$ is used to denote the $n\times n$ identity matrix.  For any natural number $n$, we let $[n]$ or $1:n$ denote the set $\{1,2,...,n\}$. We use $|\Omega|$ denote the cardinality (i.e., the number of elements) of a set $\Omega$. MATLAB notations  are adopted for matrix indexing; that is, for the $n\times m$ matrix $\mA$,
its $(i,j)$-th element is denoted by $\mA[i,j]$, its $i$-th row (or column) is denoted by $\mA[i,:]$ (or $\mA[:,i]$), and $\mA[\Omega_1,\Omega_2]$ refers to a $|\Omega_1|\times |\Omega_2|$ submatrix obtained by taking the elements in rows $\Omega_1$ of columns $\Omega_2$ of matrix $\mA$. Here $\Omega_1\subset[n]$ and $\Omega_2\subset[n]$. We use $a\gtrsim b$ (or $a\lesssim b$) to represent that there is a constant so that $a\geq \text{Const}\cdot b$ (or $a\leq \text{Const}\cdot b$).

We say that a (not necessarily square) matrix $\mA \in \R^{n\times r}$ is orthonormal if the columns of $\mA$ are normalized and orthogonal to each other, i.e., $\mA^\T \mA = \mId$. The set of $r\times r$ orthonormal matrices is denoted by $\calO_r:=\{\mR\in\R^{r\times r}:\mR^\T\mR = \mId\}$. We say that a (not necessarily square) matrix $\mA \in \R^{n\times r}$ is orthogonal if $\langle\mA[:,i], \mA[:,j]\rangle = 0$ for all $i\neq j$; that is the columns of $\mA$ are orthogonal to each other, but are not necessarily normalized and could even be zero.

If a function $h(\mU,\mV)$ has two arguments, $\mU\in\R^{n\times r}$ and $\mV\in\R^{m\times r}$, we occasionally use the notation $h(\mW)$ when we put these two arguments into a new one as $\mW=\begin{bmatrix}\mU \\ \mV \end{bmatrix}$.  For a scalar function $f(\mZ)$ with a matrix variable $\mZ\in\R^{n\times m}$, its gradient is an $n\times m$ matrix whose $(i,j)$-th entry is $[\nabla f(\mZ)][i,j] = \frac{\partial f(\mZ)}{\partial\mZ[i,j]}$ for all $i\in\{1,2,\ldots,n\}, j\in\{1,2,\ldots,m\} $. The Hessian of $f(\mZ)$ can be viewed as an $nm\times nm$ matrix $[\nabla^2 f(\mZ)][i,j] = \frac{\partial^2 f(\mZ)}{\partial\vz[i]\partial \vz[j]}$ for all $i,j\in\{1,\ldots,nm\}$, where $\vz[i]$ is the $i$-th entry of the vectorization of $\mZ$. An alternative way to represent the Hessian is by a bilinear form defined via
$[\nabla^2f(\mZ)](\mA,\mB) = \sum_{i,j,k,l}\frac{\partial^2 f(\mZ)}{\partial \mZ[i,j]\partial \mZ[k,\ell]}\mA[i,j]\mB[k,\ell]$ for any $\mA,\mB\in\R^{n\times m}$. These two notations will be used interchangeably whenever the specific form can be inferred from context.

\section{Preliminaries}\label{sec:Preliminaries}
In this section, we provide a number of important definitions in optimization and group theory. To begin, suppose $h(\vx):\R^n\to \R$ is  twice differentiable.
\begin{defi}[Critical points]  A point $\vx$ is a critical point of $h(\vx)$ if $\nabla h(\vx) = \vzero$.
\label{def:critical point}\end{defi}

\begin{defi}[Strict saddles; or ridable saddles in \cite{sun2015complete}]
A critical point $\vx$ is a strict saddle if the Hessian matrix evaluated at this point has a strictly negative eigenvalue, i.e., $\lambda_{\min}(\nabla^2 h(\vx))<0$.
\end{defi}

\begin{defi}[Strict saddle property~\cite{ge2015escaping}]\label{def:strict saddle property}
A twice differentiable function satisfies the strict saddle property if each  critical point either corresponds to a local minimum or is a strict saddle.
\end{defi}
Intuitively, the strict saddle property requires a function to have a directional negative curvature at all of the critical points but local minima. This property allows a number of iterative algorithms such as noisy gradient descent~\cite{ge2015escaping} and the trust region method~\cite{conn2000trust} to further decrease the function value at all the strict saddles and thus converge to a local minimum. 

In~\cite{ge2015escaping}, the authors proposed a noisy gradient descent algorithm for the optimization of functions satisfying the robust strict saddle property.
\begin{defi}[Robust strict saddle property~\cite{ge2015escaping}] Given $\alpha,\gamma,\epsilon,\delta$, a twice differentiable $h(\vx)$ satisfies the  $(\alpha,\gamma,\epsilon,\delta)$-robust strict saddle property if for every point $\vx$ at least one of the following applies:
\begin{enumerate}
\item There exists a local minimum point $\vx^\star$ such that $\|\vx^\star - \vx\|\leq \delta$, and the function $h(\vx')$ restricted to a $2\delta$ neighborhood of $\vx^\star$ (i.e., $\|\vx^\star - \vx'\|\leq 2\delta$) is $\alpha$-strongly convex;
\item $\lambda_{\min}\left(\nabla^2 h(\vx)\right)\leq -\gamma$;
\item $\|\nabla h(\vx)\|\geq \epsilon$.
\end{enumerate}
\label{def:robust strict saddle}\end{defi}
In words, the above robust strict saddle property says that for any point whose gradient is small, then either the Hessian matrix evaluated at this point has a strictly negative eigenvalue, or it is close to a local minimum point. Thus the robust strict saddle property not only requires that the function obeys the strict saddle property, but also that it is well-behaved (i.e., strongly convex) near the local minima and has large gradient at the points far way to the critical points.

Intuitively, when the gradient is large, the function value will decrease in one step by gradient descent; when the point is close to a saddle point, the noise introduced in the noisy gradient descent could help the algorithm escape the saddle point and the function value will also decrease; when the point is close to a local minimum point, the algorithm then converges to a local minimum.  Ge et al.~\cite{ge2015escaping} rigorously showed that the noisy gradient descent algorithm (see \cite[Algorithm 1]{ge2015escaping}) outputs a local minimum in a polynomial number of steps if the function $h(\vx)$ satisfies the robust strict saddle property.

It is proved in~\cite{ge2015escaping} that tensor decomposition problems satisfy this robust strict saddle property. However, requiring the local strong convexity prohibits the potential extension of the analysis in~\cite{ge2015escaping} for the noisy gradient descent algorithm to many other problems, for which it is not possible to be strongly convex in any neighborhood around the local minimum points. Typical examples include the matrix factorization problems due to the rotational degrees of freedom for any critical point. This motivates us to weaken the local strong convexity assumption relying on the approach used by \cite{candes2015Wirtinger,tu2015low}
and to provide the following revised robust strict saddle property for such problems. To that end, we list some necessary definitions related to groups and invariance of a function under the group action.

\begin{defi}[Definition 7.1 \cite{chirikjian2016harmonic})] A (closed) binary operation, $\circ$, is a law of composition that produces an element of a set from two elements of the same set. More precisely, let $\calG$ be a set and $a_1,a_2\in \calG$ be arbitrary elements. Then $(a_1,a_2)\rightarrow a_1\circ a_2\in \calG$.
\end{defi}

\begin{defi}
[Definition 7.2 \cite{chirikjian2016harmonic})] A {\bf group} is a set $\calG$ together with a (closed) binary operation $\circ$ such that for any elements $a,a_1,a_2,a_3\in \calG$ the following properties hold:
\begin{itemize}
\item Associative property: $a_1\circ (a_2\circ a_3) = (a_1 \circ a_2)\circ a_3$.
\item There exists an identity element $e\in \calG$ such that $e\circ a = a\circ e = a$.
\item There is an element $a^{-1}\in \calG$ such that $a^{-1}\circ a = a\circ a^{-1} = e$.
\end{itemize}
\end{defi}
With this definition, it is common to denote a group just by $\calG$ without saying the binary operation $\circ$ when it is clear from the context.

\begin{defi} Given a function $h(\vx):\R^n\to \R$ and a group $\calG$ of operators on $\R^n$, we say $h$ is invariant under the group action (or under an element $a$ of the group) if
\[
h(a(\vx)) = h(\vx)
\]
for all $\vx\in\R^n$ and $a\in\calG$.
\end{defi}
Suppose the group action also preserves the energy of $\vx$, i.e., $\|a(\vx)\| = \|\vx\|$ for all $a\in\calG$.
Since for any $\vx\in\R^n$, $h(a(\vx)) = h(\vx)$ for all $a\in\calG$, it is straightforward to stratify the domain of $h(\vx)$ into equivalent classes. The vectors in each of these equivalent classes differ by a group action. One implication is that when considering the distance of two points $\vx_1$ and $\vx_2$, it would be helpful to use the distance between their corresponding classes:
\begin{equation}\begin{split}
\dist(\vx_1,\vx_2) :&= \min_{a_1\in\calG,a_2\in\calG}\|a_1(\vx_1) - a_2(\vx_2)\|\\
 &= \min_{a\in\calG}\|\vx_1 - a(\vx_2)\|,
\end{split}\label{eq:dist with group}\end{equation}
where the second equality follows because $\|a_1(\vx_1) - a_2(\vx_2)\| = \|a_1(\vx_1 - a_1^{-1}\circ a_2(\vx_2))\| = \|\vx_1 - a_1^{-1}\circ a_2(\vx_2)\| $ and $a_1^{-1}\circ a_2\in\calG$.
Another implication is that the function $h(\vx)$ cannot possibly be strongly convex (or even convex) in any neighborhood around its local minimum points because of the existence of the equivalent classes. Before presenting the revised robust strict saddle property for invariant functions, we list two examples to illuminate these concepts.

\vspace{.1in}
{\noindent \em Example 1:} As one example, consider the phase retrieval problem of recovering an $n$-dimensional complex vector $\vx^\star$ from $\left\{y_i = \left|\vb_i^\H \vx^\star\right|,i = 1,\ldots,p\right\}$, the magnitude of its projection onto a collection of known complex vectors $\vb_1, \vb_2,\ldots,\vb_p$~\cite{candes2015Wirtinger,sun2016geometric}. The unknown $\vx^\star$ can be estimated by solving the following natural least-squares formulation~\cite{candes2015Wirtinger,sun2016geometric}
\[
\minimize_{\vx\in\C^n}h(\vx) = \frac{1}{2p}\sum_{i=1}^p \left(y_i^2 - \left|\vb_i^\H \vx\right|^2\right)^2,
\]
where we note that here the domain of $\vx$ is $\C^n$. For this case, we denote the corresponding
\[
\calG = \{e^{j \theta}:\theta\in[0,1)\}
\]
 and the group action as $a(\vx) = e^{j \theta}\vx$, where $a = e^{j\theta}$ is an element in $\calG$. It is clear that $h(a(\vx)) = h(\vx)$ for all $a\in\calG$. Due to this invariance of $h(\vx)$, it is impossible to recover the global phase factor of the unknown $\vx^\star$ and the function $h(\vx)$ is not strongly convex in any neighborhood of $\vx^\star$.

\vspace{.1in}
{\noindent \em Example 2:} As another example, we revisit the general factored low-rank optimization problem~\eqref{eq:factored problem no regularizer}:
\begin{align*}
\minimize_{\mU\in\R^{n\times r},\mV\in\R^{m\times r}} h(\mU,\mV) = f(\mU\mV^\T).
\end{align*}
We recast the two variables $\mU,\mV$ into $\mW$ as $\mW = \begin{bmatrix}\mU\\ \mV \end{bmatrix}$.
For this example, we denote the corresponding $\calG = \calO_r$ and the group action on $\mW$ as
$a(\mW) = \begin{bmatrix}\mU\mR\\ \mV\mR \end{bmatrix}$ where $a=\mR\in\calG$. We have that $h(a(\mW)) = h(\mW)$ for all $a\in\calG$ since $\mU\mR(\mV\mR)^\T = \mU\mV^\T$ for any $\mR\in\calO_r$. Because of this invariance, in general $h(\mW)$ is not strongly convex in any neighborhood around its local minimum points even though $f(\mX)$ is a strongly convex function; see~\cite{li2019symmetry} for the symmetric low-rank factorization problem and Theorem~\ref{thm:strict saddle property} in Appendix~\ref{sec:low rank factorization} for the nonsymmetric low-rank factorization problem.

\vspace{.1in}
In the examples illustrated above, due to the invariance, the function is not strongly convex (or even convex) in any neighborhood around its local minimum  point and thus it is prohibitive to apply the standard approach in optimization to show the convergence in a small neighborhood around the local minimum point.  To overcome this issue, Cand\`{e}s et al.~\cite{candes2015Wirtinger} utilized the so-called regularity condition as a sufficient condition for local convergence of gradient descent applied for the phase retrieval problem. This approach has also been applied for the matrix sensing problem~\cite{tu2015low} and semi-definite optimization~\cite{bhojanapalli2015dropping}.
\begin{defi}[Regularity condition~\cite{candes2015Wirtinger,tu2015low}] Suppose $h(\vx):\R^n\rightarrow \R$ is invariant under the group action of the given group $\calG$. Let $\vx^\star\in\R^n$ be a local minimum point of $h(\vx)$. Define the set $B(\delta,\vx^\star)$ as
\[
B(\delta,\vx^\star): = \left\{\vx\in\R^n:\dist(\vx,\vx^\star)\leq \delta\right\},
\]
where the distance $\dist(\vx,\vx^\star)$ is defined in~\eqref{eq:dist with group}. Then we say the function $h(\vx)$ satisfies the $(\alpha,\beta,\delta)$-regularity condition if for all $\vx\in B(\delta,\vx^\star)$, we have
\begin{align}
\left\langle \nabla h(\vx), \vx - a(\vx^\star)\right\rangle \geq \alpha \dist(\vx,\vx^\star)^2 + \beta\|\nabla h(\vx)\|^2,
\label{eq:regularity condition}\end{align}
where $a = \argmin_{a'\in\calG}\|\vx - a'(\vx^\star)\|$.
\label{def:regularity condition}\end{defi}
We remark that $(\alpha,\beta)$ in the regularity condition~\eqref{def:regularity condition} must satisfy $\alpha\beta\leq \frac{1}{4}$ since by applying Cauchy-Schwarz
\[
\left\langle \nabla h(\vx), \vx - a(\vx^\star)\right\rangle \leq \|\nabla h(\vx)\|\dist(\vx,\vx^\star)
\]
and the inequality of arithmetic and geometric means
\[
\alpha \dist^2(\vx,\vx^\star) + \beta \|\nabla h(\vx)\|^2 \geq 2\sqrt{\alpha\beta}\dist(\vx,\vx^\star)\|\nabla h(\vx)\|^2.
\]

\begin{lem}\label{lem:local descent}\cite{candes2015Wirtinger,tu2015low}
If the function $h(\vx)$ restricted to a $\delta$ neighborhood of $\vx^\star$ satisfies the $(\alpha,\beta,\delta)$-regularity condition, then as long as gradient descent starts from a point $\vx_0\in B(\delta,\vx^\star)$, the gradient descent update
\begin{align*}
\vx_{t+1} = \vx_{t} - \nu\nabla h(\vx_t)
\end{align*}
with step size $0<\nu\leq 2\beta$ obeys $\vx_t\in B(\delta,\vx^\star)$ and
\[
\dist^2(\vx_t,\vx^\star)\leq \left(1 - 2\nu \alpha\right)^t\dist^2(\vx_0,\vx^\star)
\]
for all $t\geq 0$.
\end{lem}
The proof is given in~\cite{candes2015Wirtinger}. To keep the paper self-contained, we also provide the proof of Lemma~\ref{lem:local descent} in Appendix~\ref{sec:prf local descent}. We remark that  the decreasing rate $1-2\nu\alpha\in [0,1)$ since we choose $\nu\leq 2\beta$ and $\alpha\beta\leq \frac{1}{4}$.

Now we establish the following revised robust strict saddle property for invariant functions by replacing the strong convexity condition in Definition~\ref{def:robust strict saddle} with the regularity condition.

\begin{defi}[Revised robust strict saddle property for invariant functions]\label{def:revised robust strict saddle} Given a twice differentiable $h(\vx):\R^{n}\to \R$ and a group $\calG$, suppose $h(\vx)$ is invariant under the group action and the energy of $\vx$ is also preserved under the group action, i.e., $h(a(\vx)) = h(\vx)$ and $\|a(\vx)\|_2 = \|\vx\|_2$ for all $a\in\calG$. Given $\alpha,\beta,\gamma,\epsilon,\delta$, $h(\vx)$ satisfies the $(\alpha,\beta,\gamma,\epsilon,\delta)$-robust strict saddle property if for any point $\vx$ at least one of the following applies:
\begin{enumerate}
\item There exists a local minimum point $\vx^\star$ such that $\dist(\vx,\vx^\star)\leq \delta$, and the function $h(\vx')$ restricted to $2\delta$ a neighborhood of $\vx^\star$ (i.e., $\dist(\vx',\vx^\star)\leq 2\delta$) satisfies the $(\alpha,\beta,2\delta)$-regularity condition defined in Definition~\ref{def:regularity condition};
\item $\lambda_{\min}\left(\nabla^2 h(\vx)\right)\leq -\gamma$;
\item $\|\nabla h(\vx)\|\geq \epsilon$.
\end{enumerate}
\end{defi}
Compared with Definition~\ref{def:robust strict saddle}, the revised robust strict saddle property requires the local descent condition instead of strict convexity in a small neighborhood around any local minimum point. With the convergence guarantee in Lemma~\ref{lem:local descent}, the convergence analysis of the stochastic gradient descent algorithm in~\cite{ge2015escaping} for the robust strict saddle functions can also be applied for the revised robust strict saddle functions defined in Definition~\ref{def:revised robust strict saddle} with the same convergence rate.\footnote{As mentioned previously, a similar notion of a revised robust strict saddle property has also recently been utilized in \cite{jin2017escape}.} We omit the details here and refer the reader to \cite{jin2017escape} for more details on this. In the rest of the paper, the  robust strict saddle property refers to the one in Definition~\ref{def:revised robust strict saddle}.

\section{Low-rank Matrix Optimization with the factorization approach}\label{sec:low-rank optimization}
In this section, we consider the minimization of general rank-constrained optimization problems of the form~\eqref{eq:general problem} using the factorization approach \eqref{eq:factored problem no regularizer} (which we repeat as follows):
\begin{align*}
\minimize_{\mU\in\R^{n\times r},\mV\in\R^{m\times r}} h(\mU,\mV)=f(\mU\mV^\T),
\end{align*}
where the rank constraint in \eqref{eq:general problem} is automatically satisfied by the factorization approach. With necessary assumptions on $f$ in Section \ref{sec:assumption}, we provide geometric analysis of the factored problem in Section~\ref{sec:robust strict saddle for general}.  We then present a stylized application in matrix sensing in Section~\ref{sec:matrix sensing}.

\subsection{Assumptions and regularizer}
\label{sec:assumption}
Before presenting our main results, we lay out the necessary assumptions on the objective function $f(\mX)$. As is known, without any assumptions on the problem, even minimizing traditional quadratic objective functions is challenging. For this reason, we focus on problems satisfying the following two assumptions.
\begin{assump}\label{assump:1}
$f(\mX)$ has a critical point $\mX^\star\in\R^{n\times m}$ which has rank $r$.
\end{assump}
\begin{assump}\label{assump:2}
$f(\mX)$ is $(2r,4r)$-restricted strongly convex and smooth, i.e., for any $n\times m$ matrices $\mX, \mD$ with $\rank(\mX)\leq 2r$ and $\rank(\mD)\leq 4r$, the Hessian of $f(\mX)$ satisfies
\begin{align}
a\left\|\mD\right\|_F^2 \leq [\nabla^2 f(\mX)](\mD,\mD) \leq b \left\|\mD\right\|_F^2
\label{eq:RIP like}\end{align}
for some positive $a$ and $b$.
\end{assump}

\Cref{assump:1} is equivalent to the existence of a rank $r$ $\mX^\star$ such that  $\nabla f(\mX^\star) = \mzero$, which is very mild and holds in many matrix inverse problems including matrix sensing~\cite{recht2010guaranteed}, matrix completion~\cite{candes2009exact} and 1-bit matrix completion~\cite{davenport20141}, where the unknown matrix to be recovered is a critical point of $f$.

\Cref{assump:2} is also utilized in~\cite[Conditions 5.3 and 5.4]{wang2016unified} and \cite{zhu2017GlobalOptimality}, where weighted low-rank matrix factorization and a set of matrix inverse problems are proved to satisfy the $(2r,4r)$-restricted  strong convexity and smoothness condition \eqref{eq:RIP like}. We discuss matrix sensing as a typical example satisfying this assumption in Section~\ref{sec:matrix sensing}.

Combining \Cref{assump:1} and \Cref{assump:2}, we have that $\mX^\star$ is the unique global minimum of \eqref{eq:general problem}.

\begin{prop}\label{prop:RIP to unique} Suppose $f(\mX)$ satisfies the $(2r,4r)$-restricted  strong convexity and smoothness condition \eqref{eq:RIP like} with positive  $a$ and $b$. Assume $\mX^\star$ is a critical point of $f(\mX)$ with $\rank(\mX^\star) = r$. Then $\mX^\star$ is the global minimum of \eqref{eq:general problem}, i.e.,
\[
f(\mX^\star)\leq f(\mX), \ \forall \ \mX\in\R^{n\times m}, \rank(\mX)\leq r
\]
and the equality holds only at $\mX = \mX^\star$.
\end{prop}
The proof of \Cref{prop:RIP to unique} is given in Appendix~\ref{sec:prf prop RIP to unique}.
We note that \Cref{prop:RIP to unique} guarantees that $\mX^\star$ is the unique global minimum of \eqref{eq:general problem} and it is expected that solving the factorized problem \eqref{eq:general low rank} also gives $\mX^\star$. \Cref{prop:RIP to unique} differs from \cite{zhu2017GlobalOptimality} in that it only requires $\mX^\star$ as a critical point, while \cite{zhu2017GlobalOptimality} needs $\mX^\star$ as a global minimum of $f$.

Before presenting the main result, we note that if $f$ satisfies \eqref{eq:RIP like} with positive $a$ and $b$ and we rescale $f$ as $f' = \frac{2}{a + b} f$, then $f'$ satisfies
\begin{align*}
\frac{2a}{a + b}\left\|\mD\right\|_F^2 \leq [\nabla^2 f'(\mX)](\mD,\mD) \leq \frac{2b}{a + b} \left\|\mD\right\|_F^2.
\end{align*}
It is clear that $f$ and $f'$ have the same optimization geometry (despite the scaling difference). Let $a' = \frac{2a}{a+ b} = 1 - c$ and $b' = \frac{2a}{a + b} = 1+c$ with $c = \frac{b-a}{a+b}$. We have $0<a'\leq 1 \leq b'$ and $a' + b'= 2$. Thus, throughout the paper and without the generality, we assume
\begin{align}
a = 1-c, \ b = 1 + c, \ c\in[0,1).
\label{eq: a b c}\end{align}

Now let $\mX^\star = \mPhi\mSigma\mPsi^\T = \sum_{i=1}^{r}\sigma_i\vphi_i\vpsi_i^\T$ be a reduced SVD of $\mX^\star$, where $\mSigma$ is a diagonal matrix with $\sigma_1\geq \cdots \geq \sigma_r$ along its diagonal. Denote \begin{align}
\mU^\star = \mPhi\mSigma^{1/2}\mR, \mV^\star = \mPsi\mSigma^{1/2}\mR
\label{eq:define U Vstar}\end{align}
for any $\mR\in\calO_r$. We first introduce the following ways to stack $\mU$ and $\mV$ together that are widely used through the paper:
\[
\mW = \begin{bmatrix} \mU \\ \mV \end{bmatrix}, \quad \widehat\mW = \begin{bmatrix} \mU \\ -\mV \end{bmatrix}, \mW^\star = \begin{bmatrix} \mU^\star \\ \mV^\star \end{bmatrix}, \quad \widehat\mW^\star = \begin{bmatrix} \mU^\star \\ -\mV^\star \end{bmatrix}.
\]

Before moving on, we note that for any solution $(\mU,\mV)$ to \eqref{eq:factored problem no regularizer}, $(\mU\mR_1,\mV\mR_2)$ is also a solution to \eqref{eq:factored problem no regularizer} for any $\mR_1,\mR_2\in\R^{r\times r}$ such that $\mU\mR_1\mR_2^\T\mV^\T = \mU\mV^\T$. As an extreme example, $\mR_1 = c \mId$ and $\mR_2 = \frac{1}{c}\mId$ where $c$ can be arbitrarily large. In order to address this ambiguity (i.e., to reduce the search space of $\mW$ for \eqref{eq:low rank fact no regu}), we utilize the trick in \cite{tu2015low,park2016non,wang2016unified,zhu2017GlobalOptimality} by introducing a regularizer $\rho$ and turn to solve the following problem
\begin{equation}
\minimize_{\mU\in\R^{n\times r},\mV\in\R^{m\times r}} G(\mW) := h(\mW) + \rho(\mW),
\label{eq:general low rank}\end{equation}
where
\[
\rho(\mW): = \frac{\mu}{4} \left\|\mU^\T\mU - \mV^\T \mV\right\|_F^2.
\]
We remark that $\mW^\star$ is still a global minimizer of the factored problem \eqref{eq:low rank approx fact} since both the first term and $\rho(\mW)$ achieve their global minimum at $\mW^\star$. The regularizer $\rho(\mW)$ is applied to force the difference between the Gram matrices of $\mU$ and $\mV$ as small as possible. The global minimum of $\rho(\mW)$ is $0$, which is achieved when $\mU$ and $\mV$ have the same Gram matrices, i.e., when $\mW$ belongs to
\begin{align}\label{eq:set of balanced factors}
\calE: = \left\{\mW = \begin{bmatrix} \mU \\ \mV \end{bmatrix}: \mU^\T\mU - \mV^\T \mV =  \mzero\right\}.
\end{align}
Informally, we can view \eqref{eq:general low rank} as finding a point from $\calE$ that also minimizes the first term in \eqref{eq:general low rank}. This is rigorously established in the following result which reveals that any critical point $\mW$ of $g(\mW)$ belongs to $\calE$ (that is $\mU$ and $\mV$ are balanced factors of their product $\mU\mV^\T$) for any $\mu>0$.
\begin{lem}\label{lem:critical point balanced}\cite[Theorem 3]{zhu2017GlobalOptimality}
Suppose $G(\mW)$ is defined as in \eqref{eq:general low rank} with $\mu>0$. Then any critical point $\mW$ of $G(\mW)$ belongs to $\calE$, i.e.,
\begin{align}
\nabla G(\mW) = \mzero \quad \Rightarrow \quad \mU^\T\mU=\mV^\T\mV.
\label{eq:critical point balanced}\end{align}
\end{lem}
For completeness, we include the proof of Lemma~\ref{lem:critical point balanced} in Appendix~\ref{sec:proof critical point balanced}.

\subsection{Global geometry for general low-rank optimization}
\label{sec:robust strict saddle for general}
We now characterize the global optimization geometry of the factored problem \eqref{eq:general low rank}. As explained in Section~\ref{sec:Preliminaries} that $G(\mW)$ is invariant under the matrices $\mR\in\setO_r$,  we first recall the discussions in Section~\ref{sec:Preliminaries} about the revised robust strict saddle property for the invariant functions. To that end, we follow the notion of the distance between equivalent classes for invariant functions  defined in \eqref{eq:dist with group} and define the distance between $\mW_1$ and $\mW_2$ as follows
\begin{equation}\begin{split}
\dist(\mW_1,\mW_2):&= \min_{\mR_1\in\calO_r,\mR_2\in\calO_r}\left\|\mW_1\mR_1 - \mW_2\mR_2\right\|_F\\
&= \min_{\mR\in\calO_r}\left\|\mW_1- \mW_2\mR\right\|_F.
\end{split}\label{eq:distance W1 and W2}\end{equation}
For convenience, we also denote the best rotation matrix $\mR$ so that $\left\|\mW_1- \mW_2\mR\right\|_F$ achieves its minimum by $\mR(\mW_1,\mW_2)$, i.e.,
\begin{align}
\mR(\mW_1,\mW_2) : = \arg\min_{\mR'\in\calO_r}\left\|\mW_1- \mW_2\mR'\right\|_F,
\label{eq:procrustes}\end{align}
which is also known as the orthogonal Procrustes problem~\cite{higham1995matrix}. The solution to the above minimization problem is characterized by the following lemma.
\begin{lem}\cite{higham1995matrix}
Let $\mW_2^\T\mW_1 = \mL\mS\mP^\T$ be an SVD of $\mW_2^\T\mW_1$. An optimal solution for the orthogonal Procrustes problem~\eqref{eq:procrustes} is given by
\begin{align*}
\mR(\mW_1,\mW_2)  = \mL\mP^\T.
\end{align*}
Moreover, we have
\begin{align*}
\mW_1^\T\mW_2\mR(\mW_1,\mW_2) & = \left(\mW_2\mR(\mW_1,\mW_2)\right)^\T\mW_1\\
 &= \mP\mS\mP^\T\succeq \mzero.
\end{align*}
\label{lem:Procrustes problem}\end{lem}
To ease the notation, we drop $\mW_1$ and $\mW_2$ in $\mR(\mW_1,\mW_2)$ and rewrite $\mR$ instead of $\mR(\mW_1,\mW_2)$ when they ($\mW_1$ and $\mW_2$) are clear from the context.  Now we are well equipped to present the robust strict saddle property for $G(\mW)$ in the following result.
\begin{thm}\label{thm:robust strict saddle property general}
Define the following regions
\begin{align*}
\calR_1:&=\left\{\mW: \dist(\mW,\mW^\star)\leq \sigma_r^{1/2}(\mX^\star)\right\},
\end{align*}
\begin{align*}
\calR_2:&=\bigg\{\mW: \sigma_r(\mW)\leq \sqrt{\frac{1}{2}}\sigma_r^{1/2}(\mX^\star), \\ &\quad\quad\|\mW\mW^\T\|_F\leq \frac{20}{19} \|\mW^\star\mW^{\star\T}\|_F\bigg\},
\end{align*}
\begin{align*}
\calR_3'&:=\bigg\{\mW: \dist(\mW,\mW^\star)> \sigma_r^{1/2}(\mX^\star),  \|\mW\|\leq \frac{20}{19} \|\mW^\star\| ,\\
 &\sigma_r(\mW)>\sqrt{\frac{1}{2}}\sigma_r^{1/2}(\mX^\star) , \|\mW\mW^\T\|_F\leq \frac{20}{19} \|\mW^\star\mW^{\star\T}\|_F\bigg\},
 \end{align*}
 \begin{align*}
\calR_3'':&=\bigg\{\mW:\|\mW\|> \frac{20}{19} \|\mW^\star\| = \frac{20}{19}\sqrt{2}\|\mX^\star\|^{1/2},\\
 &\quad\quad\|\mW\mW^\T\|_F\leq \frac{10}{9} \|\mW^\star\mW^{\star\T}\|_F\bigg\},
 \end{align*}
 \begin{align*}
\calR_3''':&=\left\{\mW:\|\mW\mW^\T\|_F> \frac{10}{9} \|\mW^\star\mW^{\star\T}\|_F = \frac{20}{9}\|\mX^\star\|_F\right\}.
\end{align*}
Let $G(\mW)$ be defined as in \eqref{eq:general low rank} with $\mu = \frac{1}{2}$. Suppose $f(\mX)$ has a critical point $\mX^\star\in\R^{n\times m}$ of rank $r$ and satisfies the $(2r,4r)$-restricted  strong convexity and smoothness condition \eqref{eq:RIP like} with positive constants $a = 1-c,b = 1+c$  and
\begin{align}\label{eq:thm RIP condition}
c\le {\frac{1}{100}}\frac{\sigma_r^{3/2}(\mX^\star)}{\|\mX^\star\|_F\|\mX^\star\|^{1/2}}.
\end{align}
 Then $G(\mW)$ has the following robust strict saddle property:
\begin{enumerate}
\item For any $\mW\in\calR_1$, $G(\mW)$ satisfies the local regularity condition:
\begin{equation}
\begin{split}
&\left\langle \nabla G(\mW), \mW - \mW^\star\right\rangle\\
& \geq {\frac{1}{16}}\sigma_r(\mX^\star)\dist^2(\mW,\mW^\star) +  {\frac{1}{260}} \frac{1}{\|\mX^\star\|}\|\nabla G(\mW)\|_F^2.
\end{split}
\label{eq:thm sensing regularity condition}
\end{equation}
where $\dist(\mW,\mW^\star)$ and $\mR$ are defined in \eqref{eq:distance W1 and W2} and \eqref{eq:procrustes}, respectively.
\item For any $\mW\in\calR_2$, $G(\mW)$ has a directional negative curvature, i.e.,
    \begin{align}
\lambda_{\min}\left(\nabla^2G(\mW)\right) \le {-\frac{1}{6}}\sigma_r(\mX^\star).
\label{eq:thm sensing negative curvature}\end{align}
\item For any $\mW\in\calR_3=\calR_3'\cup \calR_3''\cup \calR_3'''$, $G(\mW)$ has large gradient:
    \begin{align}
    &\|\nabla G(\mW)\|_F \ge \revise{ \frac{1}{50}}\sigma_r^{3/2}(\mX^\star), \quad \forall~ \mW\in\calR_3'; \label{eq:thm sensing large gradient 1}\\
        &\|\nabla G(\mW)\|_F\ge {\frac{1}{50}} \|\mW\|^3, \quad \forall~ \mW\in\calR_3''; \label{eq:thm sensing large gradient 2}\\
    &\|\nabla G(\mW)\|_F \ge {\frac{1}{45}} \left\| \mW\mW^\T\right\|_F^{3/2}, \quad \forall~ \mW\in\calR_3'''. \label{eq:thm sensing large gradient 3}
    \end{align}
\end{enumerate}
\end{thm}
The proof of this result is given in Appendix~\ref{sec:prf robust strict saddle property sensing}. The main proof strategy is to utilize \Cref{assump:1} and \Cref{assump:2} about the function $f$ to control the deviation between the gradient (and the Hessian) of the general low-rank optimization \eqref{eq:general low rank} and the counterpart of the matrix factorization problem so that the landscape of the general low-rank optimization \eqref{eq:general low rank} has a similar geometry property. To that end, in Appendix~\ref{sec:low rank factorization}, we provide a comprehensive geometric analysis for the matrix factorization problem \eqref{eq:low rank fact no regu}. The reason for choosing $\mu=\frac{1}{2}$ is also discussed in Appendix~\ref{sec:robust strict saddle property PCA}.
We note that the results in  Appendix~\ref{sec:low rank factorization} are also of independent interest, as we show  that the objective function in~\eqref{eq:low rank fact no regu} obeys the strict saddle property and has no spurious local minima not only for exact-parameterization ($r = \rank(\mX^\star)$), but also for over-parameterization ($r > \rank(\mX^\star)$) and under-parameterization ($r < \rank(\mX^\star)$). Several remarks follow.

\begin{remark}
Note that
\begin{align*}
\calR_1\cup \calR_2 \cup \calR_3' \supseteq\bigg\{ &\mW: \|\mW\|\leq \frac{20}{19} \|\mW^\star\|_F,\\
 &\|\mW\mW^\T\|_F\leq \frac{10}{9} \|\mW^\star\mW^{\star\T}\|_F\bigg\},
\end{align*}
which further implies
\begin{align*}
\calR_1\cup \calR_2 \cup \calR_3'\cup\calR_3'' \supseteq\{\mW:  \|\mW\mW^\T\|_F\leq \frac{10}{9} \|\mW^\star\mW^{\star\T}\|_F\}.
\end{align*}
Thus, we conclude that $\calR_1\cup\calR_2\cup\calR_3'\cup\calR_3''\cup\calR_3''' = \R^{(n+m)\times r}$. Now the convergence analysis of the stochastic gradient descent algorithm in \cite{ge2015escaping,jin2017escape} for the robust strict saddle functions also holds for $G(\mW)$.
\end{remark}

\begin{remark}
\Cref{thm:robust strict saddle property general} states that the objective function for the general low-rank optimization~\eqref{eq:general low rank} also satisfies the robust strict saddle property when~\eqref{eq:thm RIP condition} holds. The requirement for $c$ in \eqref{eq:thm RIP condition} can be weakened to ensure the properties of $g(\mW)$ are preserved for $G(\mW)$ in some regions. For example, the local regularity condition \eqref{eq:thm sensing regularity condition} holds when
\[
c\leq \frac{1}{50}
\]
which is independent of $\mX^\star$. With the analysis of the global geometric structure in $G(\mW)$, Theorem~\ref{thm:robust strict saddle property sensing} ensures that many local search algorithms can converge to $\mX^\star$ (which is the the global minimum of \eqref{eq:general problem} as guaranteed by \Cref{prop:RIP to unique}) with random initialization. In particular, stochastic gradient descent when applied to the matrix sensing problem~\eqref{eq:sensing fact} is guaranteed to find the global minimum $\mX^\star$ in polynomial time.
\end{remark}
\begin{remark}
Local (rather than global) geometry results for the general low-rank optimization \eqref{eq:general low rank} are also covered in \cite{zhu2017GlobalOptimality}, which only characterizes the geometry at all the critical points. Instead, \Cref{thm:robust strict saddle property general} characterizes the global geometry for general low-rank optimization \eqref{eq:general low rank}. Because the analysis is different, the proof strategy for \Cref{thm:robust strict saddle property general} is also very different than that of~\cite{zhu2017GlobalOptimality}. Since \cite{zhu2017GlobalOptimality} only considers local geometry, the result in \cite{zhu2017GlobalOptimality} requires $c\leq 0.2$, which is slightly less restrictive than the one in \eqref{eq:thm RIP condition}.
\end{remark}

\begin{remark}
To explain the necessity of the requirement on the constants $a$ and $b$ in \eqref{eq:thm RIP condition}, we utilize the symmetric weighted PCA problem (so that we can visualize the landscape of the factored problem in Figure~\ref{fig:spurious local minima}) as an example where the objective function is
\begin{align}
f(\mX)=\frac{1}{2}\|\mOmega\odot(\mX-\mX^\star)\|_F^2,
\label{eq:weighted PCA}\end{align}
where $\mOmega\in\R^{n\times n}$ contains positive entries. The Hessian quadratic form for $f(\mX)$ is given by $[\nabla^2 f(\mX)](\mD,\mD)=\|\mOmega\odot \mD\|_F^2$ for any $\mD\in\R^{n\times n}$. Thus, we have
\begin{align*}
\min_{ij}|\mOmega[i,j]|^2 \leq \frac{[\nabla^2 f(\mX)](\mD,\mD)}{\|\mD\|_F^2}\leq \max_{ij}|\mOmega[i,j]|^2.
\end{align*}
Comparing with~\eqref{eq:RIP like}, we see that $f$ satisfies the restricted strong convexity and smoothness conditions with the constants $a = \min_{ij}|\mOmega[i,j]|^2$ and $b = \max_{ij}|\mOmega[i,j]|^2$. In this case, we also note that if each entry $W_{ij}$ is nonzero (i.e., $\min_{ij}|\mOmega[i,j]|^2>0$), the function $f(\mX)$ is strongly convex, rather than only restrictively strongly convex,  implying that \eqref{eq:weighted PCA}  has  a unique optimal solution $\mX^\star$.
By applying the factorization approach, we get the factored objective function
\begin{align}
h(\mU) = \frac{1}{2}\|\mOmega\odot(\mU\mU^\T-\mX^\star)\|_F^2.
\label{eq:weighted PCA factored}\end{align}
To illustrate the necessity of the requirement on the constants $a$ and $b$ as in \eqref{eq:thm RIP condition} so that the factored problem~\eqref{eq:weighted PCA factored} has no spurious local minima and obeys the robust strict saddle property, we set $\mX^\star = \begin{bmatrix} 1 & 1\\1 &1 \end{bmatrix}$ which is a rank-$1$ matrix and can be factorized as $\mX^\star = \mU^\star\mU^{\star\T}$ with $\mU^\star = \begin{bmatrix} 1 \\1\end{bmatrix}$. We then plot the landscapes of the factored objective function $h(\mU)$  with $\mOmega = \begin{bmatrix} 1 & 1\\1 &1 \end{bmatrix}$ and $\begin{bmatrix} 8 & 1\\1 &8 \end{bmatrix}$ in Figure \ref{fig:spurious local minima}. We observe from Figure \ref{fig:spurious local minima} that as long as the elements in $\mOmega$ have a small dynamic range (which corresponds to a small $b/a$), $h(\mU)$ has no spurious local minima, but if the elements in $\mOmega$ have a large dynamic range (which corresponds to a large $b/a$),  spurious local minima can appear in $h(\mU)$.

\begin{figure}[htb!]
\begin{minipage}{0.48\linewidth}
\centerline{
\includegraphics[width=1.5in]{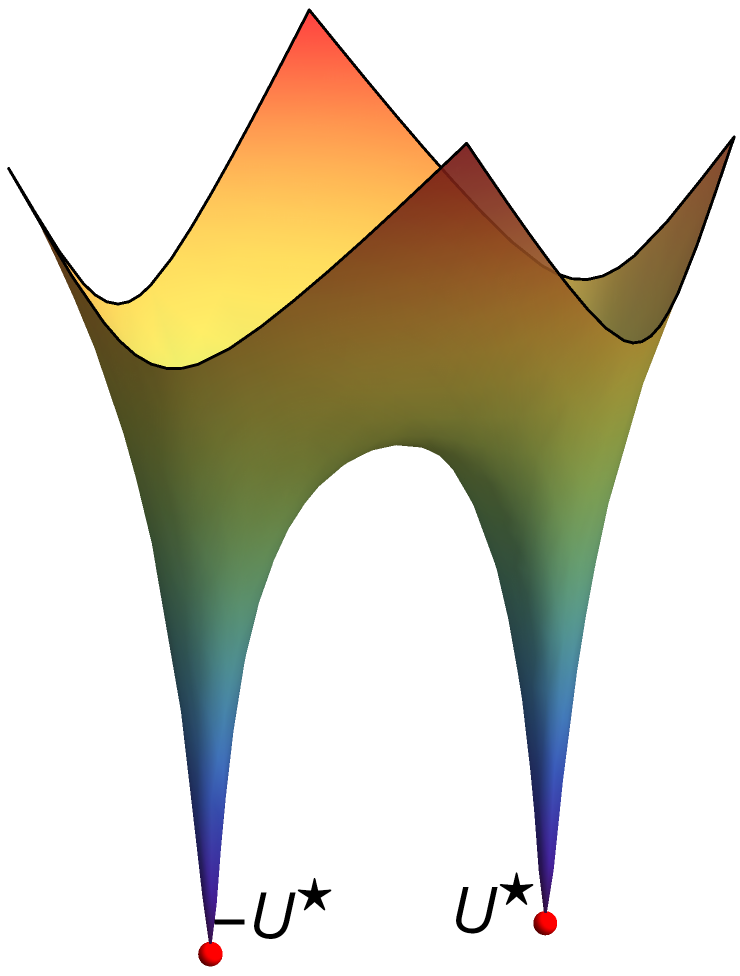}}
\centering{(a)}
\end{minipage}
\hfill
\begin{minipage}{0.48\linewidth}
\centerline{
\includegraphics[width=1.5in]{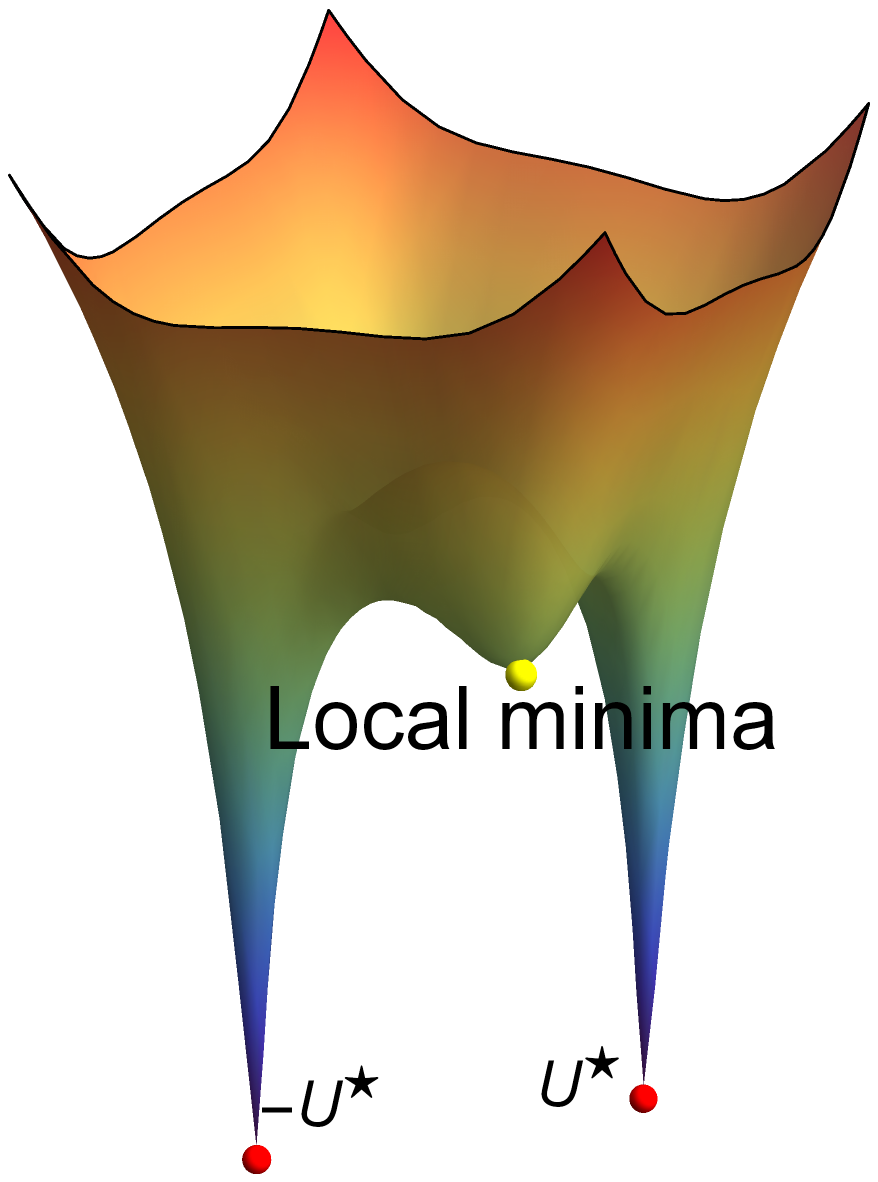}}
\centering{(b)}
\end{minipage}
\caption[Foo content]{Landscapes of $h(\mU)$ in \eqref{eq:weighted PCA factored} with $\mX^\star = \begin{bmatrix} 1 & 1\\1 &1 \end{bmatrix}$ and (a) $\mOmega = \begin{bmatrix} 1 & 1\\1 &1 \end{bmatrix}$; (b) $\mOmega = \begin{bmatrix} 8 & 1\\1 &8 \end{bmatrix}$.
}
\label{fig:spurious local minima}
\end{figure}
\end{remark}

\begin{remark}
The global geometry of low-rank matrix recovery but with analysis customized to linear measurements and quadratic loss functions is also covered in \cite{li2019symmetry,ge2017no}. Since  \Cref{thm:robust strict saddle property general} only requires the $(2r,4r)$-restricted strong convexity and smoothness property~\eqref{eq:RIP like}, aside from low-rank matrix recovery~\cite{candes2011tight}, it can also be applied to many other low-rank matrix optimization problems~\cite{udell2014generalized} which do not necessarily involve quadratic loss functions. Typical examples include 1-bit matrix completion~\cite{davenport20141,cai2013max} and Poisson principal component analysis (PCA)~\cite{salmon2014poisson}. We refer to~\cite{zhu2017GlobalOptimality} for more discussion on this issue. In next section, we consider a stylized application of \Cref{thm:robust strict saddle property general} in matrix sensing and compare it with the result in \cite{li2019symmetry}.
\end{remark}

\subsection{Stylized application: Matrix sensing}\label{sec:matrix sensing}
In this section, we extend the previous geometric analysis to the matrix sensing problem
\begin{equation}\begin{split}
\minimize_{\mU\in\R^{n\times r},\mV\in\R^{m\times r}} G(\mW) &:= \frac{1}{2}\left\|\calA\left(\mU\mV^\T - \mX^\star\right)\right\|_2^2 + \rho(\mW),
\end{split}\label{eq:sensing fact}\end{equation}
where $\calA:\R^{n\times m}\rightarrow \R^p$ is a known linear measurement operator and $\mX^\star$ is the unknown rank $r$ matrix to be recovered. In this case, we have
\[
f(\mX) = \frac{1}{2}\left\|\calA\left(\mX - \mX^\star\right)\right\|_2^2.
\]

The derivative of $f(\mX)$ at $\mX^\star$ is
\[
\nabla f(\mX^\star) = \calA^*\calA(\mX^\star - \mX^\star) = \mzero,
\]
which implies that $f(\mX)$ satisfies \Cref{assump:1}. The Hessian quadratic form $\nabla^2 f(\mX)[\mD,\mD]$ for any $n\times m$ matrices $\mX$ and $\mD$ is given by
\[
 \nabla^2 f(\mX)[\mD,\mD] = \left\|\calA(\mD)\right\|^2.
\]
The following matrix Restricted Isometry Property (RIP) serves as a way to link the low-rank matrix factorization problem~\eqref{eq:low rank approx fact} with the matrix sensing problem~\eqref{eq:sensing fact} and certifies $f(\mX)$ satisfying \Cref{assump:2}.
\begin{defi}[Restricted Isometry Property (RIP)~\cite{candes2005decoding,recht2010guaranteed}] The map $\calA:\R^{n\times m}\rightarrow \R^p$ satisfies the $r$-RIP with constant $\delta_r$ if~\footnote{By abuse of notation, we adopt the conventional notation $\delta_r$ for the RIP constant. The subscript $r$ can be used to distinguish the RIP constant $\delta_r$  from $\delta$ which is used as a small constant in Section~\ref{sec:Preliminaries}.}
\begin{align}\label{eq:RIP}
\left(1 - \delta_r\right) \left\|\mD\right\|_F^2 \leq \left\|\calA(\mD)\right\|^2\leq \left(1+\delta_r\right) \left\|\mD\right\|_F^2
\end{align}
holds for any $n\times m$ matrix $\mD$ with $\rank(\mD)\leq r$.
\end{defi}

If $\calA$ satisfies the $4r$-restricted isometry property with constant $\delta_{4r}$, then $f(\mX)$ satisfies the $(2r,4r)$-restricted strong convexity and smoothness condition~\eqref{eq:RIP like} with constants $a = 1-\delta_{4r}$ and $b = 1-\delta_{4r}$ since
\begin{equation}\begin{split}
(1-\delta_{4r})\left\|\mD\right\|_F^2 &\leq \nabla^2 f(\mX)[\mD,\mD] = \left\|\calA(\mD)\right\|^2\\ &\leq (1+\delta_{4r})\left\|\mD\right\|_F^2
\end{split}\label{eq:RIP consequence}\end{equation}
for any rank-$4r$ matrix $\mD$. Comparing \eqref{eq:RIP consequence} with \eqref{eq:RIP like}, we note that the RIP is stronger than the restricted strong convexity and smoothness property \eqref{eq:RIP like} as the RIP gives that \eqref{eq:RIP consequence} holds for all $n\times m$ matrices $\mX$, while \Cref{assump:2} only requires that \eqref{eq:RIP like} holds for all rank-$2r$ matrices.

Now, applying \Cref{thm:robust strict saddle property general}, we obtain a similar geometric guarantee to \Cref{thm:robust strict saddle property general} for the matrix sensing problem~\eqref{eq:sensing fact} when $\calA$  satisfies the RIP.
\begin{cor}\label{thm:robust strict saddle property sensing}
Let $\calR_1,\calR_2,\calR_3',\calR_3'',\calR_3'''$ be the regions as defined in Theorem~\ref{thm:robust strict saddle property}.
Let $G(\mW)$ be defined as in \eqref{eq:sensing fact} with $\mu = \frac{1}{2}$ and $\calA$ satisfying the $4r$-RIP with
\begin{align}\label{eq:cor RIP condition}
\delta_{4r} \le {\frac{1}{100}} \frac{\sigma_r^{3/2}(\mX^\star)}{\|\mX^\star\|_F\|\mX^\star\|^{1/2}}.
\end{align}
 Then $G(\mW)$ has the following robust strict saddle property:
\begin{enumerate}
\item For any $\mW\in\calR_1$, $G(\mW)$ satisfies the local regularity condition:
\begin{equation}\begin{split}
&\left\langle \nabla G(\mW), \mW - \mW^\star\right\rangle\\
& \geq {\frac{1}{16}}\sigma_r(\mX^\star)\dist^2(\mW,\mW^\star) + {\frac{1}{260}} \frac{1}{\|\mX^\star\|}\|\nabla G(\mW)\|_F^2.
\end{split}\label{eq:cor sensing regularity condition}\end{equation}
where $\dist(\mW,\mW^\star)$ and $\mR$ are defined in \eqref{eq:distance W1 and W2} and \eqref{eq:procrustes}, respectively.
\item For any $\mW\in\calR_2$, $G(\mW)$ has a directional negative curvature, i.e.,
    \begin{align*}
\lambda_{\min}\left(\nabla^2G(\mW)\right) \le {-\frac{1}{6}}\sigma_r(\mX^\star).
\end{align*}
\item For any $\mW\in\calR_3=\calR_3'\cup \calR_3''\cup \calR_3'''$, $G(\mW)$ has large gradient:
    \begin{align*}
    &\|\nabla G(\mW)\|_F \ge \revise{ \frac{1}{50}}\sigma_r^{3/2}(\mX^\star), \quad \forall~ \mW\in\calR_3'; \\
    &\|\nabla G(\mW)\|_F\ge{\frac{1}{50}} \|\mW\|^3, \quad \forall~ \mW\in\calR_3''; \\
    &\|\nabla G(\mW)\|_F \ge{\frac{1}{45}} \left\| \mW\mW^\T\right\|_F^{3/2}, \quad \forall~ \mW\in\calR_3'''.   \end{align*}
\end{enumerate}
\end{cor}

\begin{remark}
Similar to \eqref{eq:thm RIP condition}, the requirement for $\delta_{4r}$ in \eqref{eq:cor RIP condition} can be weakened to ensure the properties of $g(\mW)$ are preserved for $G(\mW)$ in some regions. For example, the local regularity condition \eqref{eq:cor sensing regularity condition} holds when
\[
\delta_{4r}\leq \frac{1}{50}
\]
which is independent of $\mX^\star$. Note that Tu et al.~\cite[Section 5.4, (5.15)]{tu2015low} provided a similar regularity condition. However, the result there requires $\delta_{6r}\leq \frac{1}{25}$ and $\dist(\mW,\mW^\star)\leq \frac{1}{2\sqrt{2}}\sigma_r(\mX^\star)$ which defines a smaller region than $\calR_1$. Based on this local regularity condition, Tu et al.~\cite{tu2015low} showed that gradient descent with a good initialization (which is close enough to $\mW^\star$) converges to the unknown matrix $\mW^\star$ (and hence $\mX^\star$). With the analysis of the global geometric structure in $G(\mW)$, \Cref{thm:robust strict saddle property sensing} ensures that many local search algorithms can find the unknown matrix $\mX^\star$ in polynomial time.
\end{remark}

\begin{remark} A Gaussian $\calA$ will  have the RIP with high probability when the number of measurements $p$ is comparable to the number of degrees of freedom in an $n\times m$ matrix with rank $r$. By Gaussian $\calA$ we mean the $\ell$-th element in $\vy = \calA(\mX)$, $y_\ell$, is given by
\[
y_\ell = \left\langle \mX, \mA_\ell \right\rangle = \sum_{i=1}^n\sum_{j=1}^m \mX[i,j]\mA_\ell[i,j],
\]
where the entries of each $n\times m$ matrix $\mA_\ell$ are independent and identically distributed normal random variables with zero mean and variance $\frac{1}{p}$. Specifically, a Gaussian $\calA$ satisfies~\eqref{eq:RIP} with high probability when~\cite{candes2011tight,davenport2016overview,recht2010guaranteed}
\begin{align*}
p \gtrsim r(n+m)\frac{1}{\delta_{r}^2}.
\end{align*}
Now utilizing the inequality $\|\mX^\star\|_F\leq \sqrt{r}\|\mX^\star\|$ for \eqref{eq:thm RIP condition}, we conclude that in the case of Gaussian measurements, the robust strict saddle property is preserved for the matrix sensing problem with high probability when the number of measurements exceeds a constant times $(n+m)r^2\kappa(\mX^\star)^3$
where $\kappa(\mX^\star) = \frac{\sigma_1(\mX^\star)}{\sigma_r(\mX^\star)}$. This further implies that, when applying the stochastic gradient descent algorithm to the matrix sensing problem~\eqref{eq:sensing fact} with Gaussian measurements, we are guaranteed to find the unknown matrix $\mX^\star$ in polynomial time with high probability when
\begin{align}
p \gtrsim (n+m)r^2\kappa(\mX^\star)^3.
\label{eq:number of measurement}\end{align}
When $\mX^\star$ is an $n\times n$ PSD matrix, Li et al.~\cite{li2019symmetry}  showed that the corresponding matrix sensing problem with Gaussian measurements has similar global geometry to the low-rank PSD matrix factorization problem when the number of measurements
\begin{align}
p \gtrsim nr^2\frac{\sigma_1^4(\mX^\star)}{\sigma_r^2(\mX^\star)}.
\label{eq:number of measurement PSD}\end{align}
Comparing \eqref{eq:number of measurement} with  \eqref{eq:number of measurement PSD}, we find both results for the number of measurements needed  depend similarly on the rank $r$, but slightly differently on the spectrum of $\mX^\star$. We finally remark that the sampling complexity in \eqref{eq:number of measurement} is $O((n+m)r^2)$, which is slightly larger than the information theoretically optimal bound $O((n+m)r)$ for matrix sensing. This is because \Cref{thm:robust strict saddle property sensing} is a direct consequence of \Cref{thm:robust strict saddle property general} in which we directly characterize the landscapes of the objective functions in the whole space by combining the results for matrix factorization in Appendix~\ref{sec:low rank factorization} and the restricted strong convexity and smoothness condition. We believe this mismatch is an artifact of our proof strategy and could be mitigated by a different approach, like utilizing the properties of quadratic loss functions~\cite{ge2017no}. If one desires only to characterize the geometry for critical points, then
$O((n+m)r)$ measurements are enough to ensure the strict saddle property and lack of spurious local minima for matrix sensing~\cite{park2016non,zhu2017GlobalOptimality}. {We finally note that for matrix completion where the RIP is not satisfied, \cite{ge2017no} proves the robust strict saddle property for the factorization approach by utilizing an additional regularizer which promotes incoherence of $\mW$.}
\end{remark}

\appendices

\section{The optimization geometry of low-rank matrix factorization}
\label{sec:low rank factorization}
In this appendix, we consider the low-rank matrix factorization problem
\begin{equation}
\minimize_{\mU\in\R^{n\times r},\mV\in\R^{m\times r}} g(\mW) := \frac{1}{2}\left\|\mU\mV^\T - \mX^\star\right\|_F^2 + \rho(\mW)
\label{eq:low rank approx fact}\end{equation}
where $\rho(\mW)$ is the regularizer used in \eqref{eq:general low rank} and repeated here:
\[
\rho(\mW) = \frac{\mu}{4} \left\|\mU^\T\mU - \mV^\T \mV\right\|_F^2.
\]
We provide a comprehensive geometric analysis for the matrix factorization problem \eqref{eq:low rank approx fact}. In particular, we show  that the objective function in~ \eqref{eq:low rank approx fact} obeys the strict saddle property and has no spurious local minima not only for exact-parameterization ($r = \rank(\mX^\star)$), but also for over-parameterization ($r > \rank(\mX^\star)$) and under-parameterization ($r < \rank(\mX^\star)$). For the exact-parameterization case, we further show that the objective function satisfies the robust strict saddle property, ensuring global convergence of many local search algorithms in polynomial time. As we believe these results  are also of independent interest and to make it easy to follow, we only present the main results in this appendix and defer the proofs to other appendices.

\subsection{Relationship to PSD low-rank matrix factorization}
Similar to \eqref{eq:define U Vstar}, let $\mX^\star = \mPhi\mSigma\mPsi^\T = \sum_{i=1}^{r}\sigma_i\vphi_i\vpsi_i^\T$ be a reduced SVD of $\mX^\star$, where $\mSigma$ is a diagonal matrix with $\sigma_1\geq \cdots \geq \sigma_r$ along its diagonal, and denote $\mU^\star = \mPhi\mSigma^{1/2}\mR, \mV^\star = \mPsi\mSigma^{1/2}\mR$ for any $\mR\in\calO_r$.
The following result to some degree characterizes the relationship between the nonsymmetric low-rank matrix factorization problem~\eqref{eq:low rank approx fact} and the following PSD low-rank matrix factorization problem~\cite{li2019symmetry}:
\begin{align}
\minimize_{\mU\in\R^{n\times r}} \left\|\mU\mU^\T - \mM\right\|_F^2,
\label{eq:symmetric low rank approx}\end{align}
where $\mM\in\R^{n\times n}$ is a rank-$r$ PSD matrix.

\begin{lem}\label{lem:nonsymmetric to symmetric}
Suppose $g(\mW)$ is defined as in \eqref{eq:low rank approx fact} with $\mu>0$. Then we have
\begin{align*}
g(\mW) \geq \min\{\frac{\mu}{4},\frac{1}{8}\}\left\|\mW\mW^\T - \mW^\star\mW^{\star\T}\right\|_F^2.
\end{align*}
In particular, if we choose $\mu = \frac{1}{2}$, then we have
\begin{align*}
g(\mW) = \frac{1}{8}\left\|\mW\mW^\T - \mW^\star\mW^{\star\T}\right\|_F^2 + \frac{1}{4} \left\|\mU^\T\mU^\star - \mV^\T\mV^\star\right\|_F^2.
\end{align*}
\end{lem}
The proof of Lemma~\ref{lem:nonsymmetric to symmetric} is given in Appendix~\ref{sec:proof lem nonsymmetric to symmetric}. Informally, Lemma~\ref{lem:nonsymmetric to symmetric} indicates that minimizing $g(\mW)$ also results in minimizing $\left\|\mW\mW^\T - \mW^\star\mW^{\star\T}\right\|_F^2$ (which is the same form as the objective function in \eqref{eq:symmetric low rank approx}) and hence the distance between $\mW$ and $\mW^\star$ (though $\mW^\star$ is unavailable at priori). The global geometry for the PSD low-rank matrix factorization problem~\eqref{eq:symmetric low rank approx} is recently analyzed by Li et al. in \cite{li2019symmetry}.

\subsection{Characterization of critical points}
We first provide the gradient and Hessian expression for $g(\mW)$.  The gradient of $g(\mW)$ is given by
\begin{align*}
&\nabla_{\mU}g(\mU,\mV) = (\mU\mV^\T - \mX^\star)\mV + \mu\mU(\mU^\T\mU - \mV^\T\mV),\\
&\nabla_{\mV}g(\mU,\mV) = (\mU\mV^\T - \mX^\star)^\T\mU - \mu\mV(\mU^\T\mU - \mV^\T\mV),
\end{align*}
which can be rewritten as
\begin{align*}
\nabla g(\mW) = \begin{bmatrix} (\mU\mV^\T - \mX^\star)\mV \\ (\mU\mV^\T - \mX^\star)^\T\mU \end{bmatrix} +  \mu \widehat\mW\widehat\mW^\T\mW.
\end{align*}
Standard computations give the Hessian quadratic form  $[\nabla^2 g(\mW)](\mDelta,\mDelta)$ for any $\mDelta = \begin{bmatrix} \mDelta_{\mU}\\ \mDelta_{\mV}\end{bmatrix}\in\R^{(n+m)\times r}$ (where $\mDelta_{\mU}\in\R^{n\times r}$ and $\mDelta_{\mV}\in\R^{m\times r}$) as
\begin{equation}\begin{split}
&[\nabla^2g(\mW)](\mDelta,\mDelta) \\&= \left\|\mDelta_{\mU}\mV^\T+ \mU\mDelta_{\mV}^\T\right\|_F^2  + 2\left\langle \mU\mV^\T - \mX^\star,\mDelta_{\mU}\mDelta_{\mV}^\T \right\rangle \\ &\quad + [\nabla^2\rho(\mW)](\mDelta,\mDelta),
\end{split}
\label{eq:hessian}\end{equation}
where
\begin{equation}\begin{split}
&[\nabla^2\rho(\mW)](\mDelta,\mDelta)\\
&= \mu\left\langle \widehat\mW^\T\mW,\widehat\mDelta^\T\mDelta \right\rangle + \mu\left\langle \widehat\mW\widehat\mDelta^\T,\mDelta \mW^\T  \right\rangle\\
&\quad + \mu \left\langle \widehat\mW\widehat\mW^\T,\mDelta \mDelta^\T\right\rangle.
\end{split}\label{eq:hessian for regularizer}\end{equation}

By Lemma~\ref{lem:critical point balanced}, we can simplify the equations for critical points as follows
\begin{align}
&\nabla_{\mU}\rho(\mU,\mV) = \mU\mU^\T\mU - \mX^\star\mV = \mzero,\label{eq:proof thm cirtical 3}\\
&\nabla_{\mV}\rho(\mU,\mV) =\mV\mV^\T\mV - \mX^{\star\T}\mU  = \mzero.\label{eq:proof thm cirtical 4}
\end{align}

Now suppose $\mW$ is a critical point of $g(\mW)$. We can apply the Gram-Schmidt process to orthonormalize the columns of $\mU$ such that $\widetilde \mU = \mU\mR$, {where $\mR\in\calO_r=\left\{\mR\in\R^{r\times r}, \mR^\T\mR = \mId\right\}$ and $\widetilde \mU$ is orthogonal.}\footnote{{As defined in Section~\ref{sec:notation}, by orthogonal we mean that $\langle\widetilde \mU[:,i], \widetilde \mU[:,j]\rangle = 0$ for all $i\neq j$. The columns of $\widetilde \mU$ are not required to be normalized, and could even be zero. Also, another way to find $\mR$ is via the SVD. Let $\mU = \mL\mSigma\mR^\T$ be a reduced SVD of $\mU$, where $\mL$ is an $n\times r$ orthonormal matrix, $\mSigma$ is an $r\times r$ diagonal matrix with non-negative diagonals, and $\mR\in\calO_r$. Then $\widetilde \mU = \mU \mR = \mL \mSigma$ is orthogonal, with possible zero columns.}} Also let $\widetilde \mV = \mV\mR$. Since $\mU^\T\mU = \mV^\T\mV$, we have $\widetilde \mU^\T\widetilde \mU = \widetilde \mV^\T\widetilde \mV$. Thus $\widetilde \mV$ is also orthogonal. Noting that $\mU\mV^\T = \widetilde \mU\widetilde \mV^\T$, we conclude that $g(\mW) = g(\widetilde \mW)$ and $\widetilde \mW$ is also a critical point of $g(\mW)$ since $\nabla_{\widetilde \mU} g(\widetilde\mW) = \nabla_{\mU} g(\mW)\mR = \mzero$ and $\nabla_{\widetilde \mV} g(\widetilde\mW) = \nabla_{\mV}g(\mW)\mR = \mzero$. Also for any $\mDelta\in\R^{(n+m)\times r}$, we have $[\nabla^2 g(\mW)](\mDelta,\mDelta) = [\nabla^2 g(\widetilde\mW)](\mDelta\mR,\mDelta\mR)$, indicating that $\mW$ and $\widetilde\mW$ have the same Hessian information. Thus, without loss of generality, we assume $\mU$ and $\mV$ are orthogonal {(including the possibility that they have zero columns)}. With this, we use $\vu_i$ and $\vv_i$ to denote the $i$-th columns of $\mU$ and $\mV$, respectively. It follows from $\nabla g(\mW) = \mzero$ that
\begin{align*}
\|\vu_i\|^2\vu_i &= \mX^\star \vv_i,\\
\|\vv_i\|^2\vv_i &= \mX^{\star\T} \vu_i.
\end{align*}
\revise{It is clear that $\vu_i = \vzero$ and $\vv_i = \vzero$ satisfy the above equations. For $\vu_i\neq \vzero$, the above equations are equivalent to (noting that $\|\vu_i\| = \|\vv_i\|$ since $\mU^\T\mU = \mV^\T\mV$)
\begin{align*}
\|\vu_i\|^2 \frac{\vu_i}{\|\vu_i\|} &= \mX^\star \frac{\vv_i}{\|\vv_i\|},\\
\|\vu_i\|^2\frac{\vv_i}{\|\vv_i\|} &= \mX^{\star\T} \frac{\vu_i}{\|\vu_i\|},
\end{align*}
which implies that $\frac{\vu_i}{\|\vu_i\|}$ and $\frac{\vv_i}{\|\vv_i\|}$ are respectively the left-singular and right-singular vectors and $\|\vu_i\|^2$ is the corresponding singular value. Thus, we conclude that 
\[
(\vu_i,\vv_i)\in\left\{(\sqrt{\sigma_1}\vphi_1,\sqrt{\sigma_1}\vpsi_1),\ldots,(\sqrt{\sigma_{r}}\vphi_{r},\sqrt{\sigma_{r}}\vpsi_{r}),(\vzero,\vzero)\right\}.
\]
Now we identify all the critical points of $g(\mW)$ in the following lemma.}
\begin{lem}
Let $\mX^\star = \mPhi\mSigma\mPsi^\T = \sum_{i=1}^{r}\sigma_i\vphi_i\vpsi_i^\T$ be a reduced SVD of $\mX^\star$ and $g(\mW)$ be defined as in \eqref{eq:low rank approx fact} with $\mu>0$.
Any $\mW = \begin{bmatrix} \mU \\ \mV \end{bmatrix}$ is a critical point of $g(\mW)$ if and only if $\mW\in\calC$ with
\begin{equation}\begin{split}
\calC: = \bigg\{\mW = \begin{bmatrix} \mU \\ \mV \end{bmatrix}: &\mU = \mPhi\mLambda^{1/2}\mR,\mV = \mPsi\mLambda^{1/2}\mR, \mR\in\calO_r,\\
  & \mLambda \textup{ is diagonal}, \mLambda\geq \mzero, (\mSigma - \mLambda)\mSigma = \mzero\bigg\}.
\end{split}\label{eq:set of critical points}\end{equation}
\label{lem:set of critical points}\end{lem}
Intuitively, \eqref{eq:set of critical points} means that a critical point $\mW$ of $g(\mW)$ is one such that $\mU\mV^\T$ is a rank-$\ell$ approximation to $\mX^\star$ with $\ell\leq r$ and $\mU$ and $\mV$ are equal factors of this rank-$\ell$ approximation. Let $\lambda_1,\lambda_2,\ldots,\lambda_r$ denote the diagonals of $\mLambda$. Unlike $\mSigma$, we note that these diagonals $\lambda_1,\lambda_2,\ldots,\lambda_r$ are not necessarily placed in decreasing or increasing order. Actually, this equation $(\mSigma - \mLambda)\mSigma = \mzero$ is equivalent to
\begin{align*}
\lambda_i \in\left\{\sigma_i,0\right\}
\end{align*}
for all $i\in\{1,2,\ldots,r\}$. Further, we introduce the set of optimal solutions:
\begin{align}
\calX: =\left\{\mW = \begin{bmatrix} \mU \\ \mV \end{bmatrix}: \mU = \mPhi\mSigma^{1/2}\mR,\mV = \mPsi\mSigma^{1/2}\mR, \mR\in\calO_r\right\}.
\label{eq:set of optimal points}\end{align}
It is clear that the set $\calX$ containing all the optimal solutions, the set $\calC$ containing all the critical points and the set $\calE$ containing all the points with balanced factors have the nesting relationship: $\calX\subset\calC\subset\calE$. Before moving to the next section, we provide one more result regarding $\mW\in\calE$. The proof of the following result is given in Appendix~\ref{sec:proof concequence of balanced factors}.
\begin{lem}\label{lem:concequence of balanced factors}
For any $\mDelta = \begin{bmatrix} \mDelta_{\mU}\\ \mDelta_{\mV}\end{bmatrix}\in\R^{(n+m)\times r}$ and $\mW\in\calE$ where $\calE$ is defined in~\eqref{eq:set of balanced factors}, we have
\begin{align}
\|\mDelta_{\mU}\mU^\T\|_F^2 + \|\mDelta_{\mV}\mV^\T\|_F^2 = \|\mDelta_{\mU}\mV^\T\|_F^2 + \|\mDelta_{\mV}\mU^\T\|_F^2,
\label{eq:on off block same energy}\end{align}
and
\begin{align}
\nabla^2 \rho(\mW)\succeq \mzero.
\label{eq:hessian regularizer psd}\end{align}
\end{lem}

\subsection{Strict saddle property}
Lemma~\ref{lem:concequence of balanced factors} implies that the Hessian of $\rho(\mW)$ evaluated at any critical point $\mW$ is PSD, i.e., $
\nabla^2 \rho(\mW)\succeq \mzero $
for all $\mW\in\calC$. Despite this fact, the following result establishes the strict saddle property for $g(\mW)$.

\begin{thm}\label{thm:strict saddle property}
Let $g(\mW)$ be defined as in \eqref{eq:low rank approx fact} with $\mu>0$ and $\rank(\mX^\star) = r$. Let $\mW = \begin{bmatrix} \mU \\ \mV \end{bmatrix}$ be any critical point satisfying $\nabla g(\mW) = \mzero$, i.e., $\mW\in\calC$. Any $\mW\in\calC\setminus\calX$ is a strict saddle of $g(\mW)$ satisfying
\begin{align}
\lambda_{\min}(\nabla^2g(\mW))\leq - \frac{1}{2}\left\|\mW\mW^\T - \mW^\star\mW^{\star\T}\right\|\leq  -\sigma_r(\mX^\star).
\label{eq:strict saddle property}\end{align}
Furthermore, $g(\mW)$ is not strongly convex at any global minimum point $\mW\in\calX$.
\end{thm}
The proof of Theorem~\ref{thm:strict saddle property} is given in Appendix~\ref{sec:proof strict saddle property}. We note that this strict saddle property is also covered in \cite[Theorem 3]{zhu2017GlobalOptimality}, but with much looser bounds (in particular, directly applying \cite[Theorem 3]{zhu2017GlobalOptimality} gives $\lambda_{\min}(\nabla^2g(\mW)) \leq -0.1 \sigma_r(\mX^\star)$ rather than $\lambda_{\min}(\nabla^2g(\mW)) \leq - \sigma_{r}(\mX^\star)$ in \eqref{eq:strict saddle property}). Theorem~\ref{thm:strict saddle property} actually implies that $g(\mW)$ has no spurious local minima (since all local minima belong to $\calX$) and obeys the strict saddle property. With the strict saddle property and lack of spurious local minima  for $g(\mW)$,  the recent results \cite{lee2016gradient,panageas2016gradient} ensure that gradient descent converges to a global minimizer almost surely with random initialization.
We also note that Theorem~\ref{thm:strict saddle property} states that $g(\mW)$ is not strongly convex at any global minimum point $\mW\in\calX$ because of the invariance property of $g(\mW)$. This is the reason we introduce the distance in \eqref{eq:distance W1 and W2} and also the robust strict saddle property in \Cref{def:revised robust strict saddle}.
%
%

\subsection{Extension to over-parameterized case: $\rank(\mX^\star)<r$}
In this section, we briefly discuss the over-parameterized scenario where the low-rank matrix $\mX^\star$ has rank smaller than $r$. Similar to Theorem~\ref{thm:strict saddle property}, the following result shows that the strict saddle property also holds in this case.
\begin{thm}\label{thm:strict saddle property over}
Let $\mX^\star = \mPhi\mSigma\mPsi^\T = \sum_{i=1}^{r'}\sigma_i\vphi_i\vpsi_i^\T$ be a reduced SVD of $\mX^\star$ with $r'\leq r$, and let $g(\mW)$ be defined as in \eqref{eq:low rank approx fact} with $\mu>0$.
Any $\mW = \begin{bmatrix} \mU \\ \mV \end{bmatrix}$ is a critical point of $g(\mW)$ if and only if $\mW\in\calC$ with
\begin{align*}
\calC: = \bigg\{\mW = \begin{bmatrix} \mU \\ \mV \end{bmatrix}:& \mU = \mPhi\mLambda^{1/2}\mR,\mV = \mPsi\mLambda^{1/2}\mR, \mR\mR^\T = \mId_{r'},\\
&\mLambda \textup{ is diagonal}, \mLambda\geq \mzero, (\mSigma - \mLambda)\mSigma = \mzero\bigg\}.
\end{align*}
Further, all the local minima (which are also global) belong to the following set
\begin{align*}
\calX =\bigg\{\mW = \begin{bmatrix} \mU \\ \mV \end{bmatrix}: &\mU = \mPhi\mSigma^{1/2}\mR,\mV = \mPsi\mSigma^{1/2}\mR,\\
& \mR\mR^\T = \mId_{r'}\bigg\}.
\end{align*}
Finally, any $\mW\in\calC\setminus\calX$ is a strict saddle of $g(\mW)$ satisfying
\begin{align*}
\lambda_{\min}(\nabla^2g(\mW))\leq -\frac{1}{2}\left\|\mW\mW^\T - \mW^\star\mW^{\star\T}\right\|\leq -\sigma_{r'}(\mX^\star).
\end{align*}
\end{thm}
The proof of Theorem~\ref{thm:strict saddle property over} is given in Appendix~\ref{sec:proof strict saddle property over}. We note that this strict saddle property is also covered in \cite[Theorem 3]{zhu2017GlobalOptimality}, but with much looser bounds (in particular, directly applying \cite[Theorem 3]{zhu2017GlobalOptimality} gives $\lambda_{\min}(\nabla^2g(\mW)) \leq -0.1 \sigma_{r'}(\mX^\star)$ rather than $\lambda_{\min}(\nabla^2g(\mW)) \leq - \sigma_{r'}(\mX^\star)$ in Theorem \ref{thm:strict saddle property over}).

\subsection{Extension to under-parameterized case: $\rank(\mX^\star)>r$}
We further discuss the under-parameterized case where $\rank(\mX^\star)>r$. In this case, \eqref{eq:low rank fact no regu} is also known as the low-rank approximation problem as the product $\mU\mV^\T$ forms a rank-$r$ approximation to $\mX^\star$. Similar to Theorem~\ref{thm:strict saddle property}, the following result shows that the strict saddle property also holds for $g(\mW)$ in this scenario.

\begin{thm}\label{thm:strict saddle property under}
Let $\mX^\star = \mPhi\mSigma\mPsi^\T = \sum_{i=1}^{r'}\sigma_i\vphi_i\vpsi_i^\T$ be a reduced SVD of $\mX^\star$ with $r'> r$ and $\sigma_r(\mX^\star)>\sigma_{r+1}(\mX^\star)$.\footnote{If $\sigma_{r_1} = \cdots = \sigma_{r} = \cdots = \sigma_{r_2}$ with $r_1\leq r\leq r_2$, then the optimal rank-$r$ approximation to $\mX^\star$ is not unique. For this case, the optimal solution set $\calX$ for the factorized problem needs to be changed correspondingly, but the main arguments still hold.} Also let $g(\mW)$ be defined as in \eqref{eq:low rank approx fact} with $\mu>0$.
Any $\mW = \begin{bmatrix} \mU \\ \mV \end{bmatrix}$ is a critical point of $g(\mW)$ if and only if $\mW\in\calC$ with
\begin{align*}
\calC: = &\bigg\{\mW = \begin{bmatrix} \mU \\ \mV \end{bmatrix}: \mU = \mPhi[:,\Omega]\mLambda^{1/2}\mR,\mV = \mPsi[:,\Omega]\mLambda^{1/2}\mR,\\
 & \mLambda =\mSigma[\Omega,\Omega], \mR\mR^\T = \mId_{\ell}, \Omega\subset\{1,2,\ldots,r'\}, |\Omega| = \ell\leq r \bigg\}
\end{align*}
where we recall that $\mPhi[:,\Omega]$ is a submatrix of $\mPhi$ obtained by keeping the columns indexed by $\Omega$ and $\mSigma[\Omega,\Omega]$ is an $\ell\times \ell$ matrix obtained by taking the elements  of $\mSigma$ in rows and columns indexed by $\Omega$.

Further, all local minima belong to the following set
\begin{align*}
\calX =&\bigg\{\mW = \begin{bmatrix} \mU \\ \mV \end{bmatrix}: \mLambda = \mSigma[1:r,1:r], \mR\in\calO_r,\\
& \mU = \mPhi[:,1:r]\mLambda^{1/2}\mR,\mV = \mPsi[:,1:r]\mLambda^{1/2}\mR \bigg\}.
\end{align*}
Finally, any $\mW\in\calC\setminus\calX$ is a strict saddle of $g(\mW)$ satisfying
\begin{align*}
\lambda_{\min}(\nabla^2g(\mW))\leq -(\sigma_{r}(\mX^\star) - \sigma_{r+1}(\mX^\star)).
\end{align*}
\end{thm}
The proof of Theorem~\ref{thm:strict saddle property under} is given in Appendix~\ref{sec:proof strict saddle property under}. It follows from Eckart-Young-Mirsky theorem~\cite{horn2012matrix} that for any $\mW\in\calX$, $\mU\mV^\T$ is the best rank-$r$ approximation to $\mX^\star$. Thus, this strict saddle property ensures that the local search algorithms applied to the factored problem~\eqref{eq:low rank approx fact} converge to global optimum which corresponds to the best rank-$r$ approximation to $\mX^\star$. {Note that Theorems~\ref{thm:strict saddle property}--\ref{thm:strict saddle property under} require $\mu>0$. Based on these results, it has been recently proved in \cite{zhu2018distributed} that the strict saddle property also holds for $g(\mW)$ even when $\mu=0$, but without an explicit bound on $\lambda_{\min}(\nabla^2 g(\mW))$ as in Theorems~\ref{thm:strict saddle property}--\ref{thm:strict saddle property under}.}

\subsection{Robust strict saddle property}
\label{sec:robust strict saddle property PCA}
We now consider the revised robust strict saddle property defined in Definition~\ref{def:revised robust strict saddle} for the low-rank matrix factorization problem~\eqref{eq:low rank approx fact}. As guaranteed by Theorem~\ref{thm:strict saddle property},  $g(\mW)$ satisfies the strict saddle property for any $\mu>0$. However, too small a $\mu$ would make analyzing the robust strict saddle property difficult. To see this, we denote
\[
f(\mW) =\frac{1}{2}\left\|\mU\mV^\T - \mX^\star\right\|_F^2
\]
for convenience. Thus we can rewrite $g(\mW)$ as the sum of $f(\mW)$ and $\rho(\mW)$.  Note that for any $\mW=\begin{bmatrix} \mU \\ \mV \end{bmatrix}\in\calC$ where $\calC$ is the set of critical points defined in \eqref{eq:set of critical points},  $\widetilde\mW=\begin{bmatrix} \mU\mM \\ \mV\mM^{-1} \end{bmatrix}$ is a critical point of $f(\mW)$ for any invertible $\mM\in\R^{r\times r}$. This further implies that the gradient at $\widetilde \mW$ reduces to
\[
\nabla g(\widetilde \mW) = \nabla \rho(\widetilde \mW),
\]
which could be very small if $\mu$ is very small since $\rho(\mW) = \frac{\mu}{4} \left\|\mU^\T\mU - \mV^\T \mV\right\|_F^2$. On the other hand, $\widetilde \mW$ could be far away from any point in $\calX$ for some $\mM$ that is not well-conditioned. Therefore, we choose a proper $\mu$ controlling the importance of the regularization term such that for any $\mW$ that is not close to the critical points $\calX$, $g(\mW)$ has large gradient. Motivated by Lemma~\ref{lem:nonsymmetric to symmetric}, we choose $\mu = \frac{1}{2}$.

The following result establishes the robust strict saddle property for $g(\mW)$.
\begin{thm}\label{thm:robust strict saddle property}
Let $\calR_1,\calR_2,\calR_3',\calR_3'',\calR_3'''$ be the regions as defined in Theorem~\ref{thm:robust strict saddle property general}.
Let $g(\mW)$ be defined as in \eqref{eq:low rank approx fact} with $\mu = \frac{1}{2}$. Then $g(\mW)$ has the following robust strict saddle property:
\begin{enumerate}
\item For any $\mW\in\calR_1$, $g(\mW)$ satisfies local regularity condition:
\begin{equation}\label{eq:thm regularity condition}\begin{split}
\left\langle \nabla g(\mW), \mW - \mW^\star\mR\right\rangle \geq & \frac{1}{32}\sigma_r(\mX^\star)\dist^2(\mW,\mW^\star)\\& + \frac{1}{48\|\mX^\star\|}\left\|\nabla g(\mW)\right\|_F^2,
\end{split}\end{equation}
where $\dist(\mW,\mW^\star)$ and $\mR$  are defined in \eqref{eq:distance W1 and W2} and \eqref{eq:procrustes}, respectively.
\item For any $\mW\in\calR_2$, $g(\mW)$ has a directional negative curvature:
    \begin{align}\label{eq:thm negative curvature}
\lambda_{\min}\left(\nabla^2g(\mW)\right) \leq  -\frac{1}{4}\sigma_r(\mX^\star).
\end{align}
\item For any $\mW\in\calR_3=\calR_3'\cup \calR_3''\cup \calR_3'''$, $g(\mW)$ has large gradient:
    \begin{align}
    &\|\nabla g(\mW)\|_F \geq \revise{\frac{1}{11}}\sigma_r^{3/2}(\mX^\star), \quad \forall~ \mW\in\calR_3';\label{eq:thm large gradient 1}\\
        &\|\nabla g(\mW)\|_F > \frac{39}{800}\|\mW\|^{3}, \quad \forall~ \mW\in\calR_3'';\label{eq:thm large gradient 2}\\
    &\left\langle \nabla g(\mW),\mW\right\rangle > \frac{1}{20} \left\| \mW\mW^\top\right\|_F^2, \quad \forall~ \mW\in\calR_3'''.\label{eq:thm large gradient 3}
    \end{align}
\end{enumerate}
\end{thm}
The proof is given in Appendix~\ref{sec:prf robust strict saddle property}.
\begin{remark}
Recall that all the strict saddles  of $g(\mW)$ are actually rank deficient (see Theorem~\ref{thm:strict saddle property}). Thus the region $\calR_2$ attempts to characterize all the neighbors of the saddle saddles by including all rank deficient points. Actually, \eqref{eq:thm negative curvature} holds not only for $\mW\in\calR_2$, but for all $\mW$ such that $\sigma_r(\mW)\leq\sqrt{\frac{1}{2}}\sigma_r^{1/2}(\mX^\star)$. The reason we add another constraint controlling the term $\|\mW^\star\mW^{\star\T}\|_F$ is to ensure this negative curvature property in the region $\calR_2$ also holds for the matrix sensing problem discussed in next section. This is the same reason we add two more constraints $\|\mW\|\leq \frac{20}{19} \|\mW^\star\|_F$ and $\|\mW\mW^\T\|_F\leq \frac{10}{9} \|\mW^\star\mW^{\star\T}\|_F$ for the region $\calR_3'$.
\end{remark}

\section{Proof of Lemma~\ref{lem:local descent}}\label{sec:prf local descent}
Denote $a_{\vx,\vx^\star}= \argmin_{a'\in\calG}\|\vx - a'(\vx^\star)\|$. Utilizing the definition of distance in \eqref{eq:dist with group}, the regularity condition~\eqref{eq:regularity condition} and the assumption that $\mu\leq 2\beta$, we have
\begin{align*}
&\dist^2(\vx_{t+1},\vx^\star)\\ &=\left\| \vx_{t+1} - a_{\vx_{t+1},\vx^\star}(\vx^\star)\right\|^2 \\
&\leq \left\| \vx_{t} - \nu \nabla h(\vx_t) - a_{\vx_{t},\vx^\star}(\vx^\star)\right\|^2\\
& = \left\| \vx_{t}  - a_{\vx_{t},\vx^\star}(\vx^\star)\right\|^2 + \nu^2\left\| \nabla h(\vx_t)\right\|^2  \\
&\quad -2\nu\left\langle \vx_{t}  - a_{\vx_{t},\vx^\star}(\vx^\star),\nabla h(\vx_t)\right\rangle\\
& \leq \left(1 - 2\nu \alpha\right)\dist^2(\vx_t,\vx^\star) - \nu(2\beta - \nu)\left\| \nabla h(\vx_t)\right\|^2\\
&\leq \left(1 - 2\nu \alpha\right)\dist^2(\vx_t,\vx^\star)
\end{align*}
where the fourth line uses the regularity condition~\eqref{eq:regularity condition} and the last line holds because $\nu \leq 2\beta$. Thus we conclude $\vx_t\in B(\delta)$ for all $t\in\N$ if $\vx_0\in B(\delta)$  by noting that $0\leq 1-2\nu\alpha<1$ since $\alpha\beta\leq \frac{1}{4}$ and $\nu\leq 2\beta$.

\section{Proof of \Cref{prop:RIP to unique}}
\label{sec:prf prop RIP to unique}
First note that if $\mX^\star$ is a critical point of $f$, then
\[
\nabla f(\mX^\star) = \mzero.
\]
Now for any $\mX\in\R^{n\times m}$ with $\rank(\mX)\leq r$, the second order Taylor expansion gives
\begin{align*}
f(\mX) =& f(\mX^\star) + \left\langle \nabla f(\mX^\star), \mX-\mX^\star \right\rangle \\
&+ \frac{1}{2}[\nabla^2 f(\widetilde \mX)](\mX-\mX^\star, \mX-\mX^\star)\\
& = f(\mX^\star) + \frac{1}{2}[\nabla^2 f(\widetilde \mX)](\mX-\mX^\star, \mX-\mX^\star)
\end{align*}
where $\widetilde \mX = t\mX^\star + (1-t)\mX$ for some $t\in[0,1]$. This Taylor expansion together with \eqref{eq:RIP like} (both $\widetilde\mX$ and  $\mX'-\mX^\star$ have rank at most $2r$) gives
\[
f(\mX) - f(\mX^\star) \geq \frac{a}{2} \|\mX - \mX^\star\|_F^2.
\]

\section{Proof of Lemma~\ref{lem:critical point balanced}}\label{sec:proof critical point balanced}
Any critical point (see Definition~\ref{def:critical point}) $\mW = \begin{bmatrix} \mU \\ \mV \end{bmatrix}$ satisfies $\nabla G(\mW) = \mzero$, i.e.,
\begin{align}
&\nabla f(\mU\mV^\T)\mV + \mu\mU\left(\mU^\T\mU - \mV^\T\mV\right) = \mzero,\label{eq:proof thm cirtical 1}\\
&(\nabla f(\mU\mV^\T))^\T\mU- \mu\mV\left(\mU^\T\mU - \mV^\T\mV\right) = \mzero.\label{eq:proof thm cirtical 2}
\end{align}
By~\eqref{eq:proof thm cirtical 2}, we obtain
\[
(\nabla f(\mU\mV^\T))^\T\mU = \mu\left(\mU^\T\mU - \mV^\T\mV\right)\mV^\T.
\]
Multiplying \eqref{eq:proof thm cirtical 1} by $\mU^\T$ and plugging it in the expression for $\mU^\T\nabla f(\mU\mV^\T)$ from the above equation gives
\begin{align*}
\left(\mU^\T\mU - \mV^\T\mV\right)\mV^\T\mV + \mU^\T\mU\left(\mU^\T\mU - \mV^\T\mV\right) = \mzero,
\end{align*}
which further implies
\begin{align*}
\mU^\T\mU\mU^\T\mU=\mV^\T\mV\mV^\T\mV.
\end{align*}
In order to show \eqref{eq:critical point balanced}, note that $\mU^\T\mU$ and $\mV^\T\mV$ are the principal square roots (i.e., PSD square roots) of $\mU^\T\mU\mU^\T\mU$ and $\mV^\T\mV\mV^\T\mV$, respectively. Utilizing the result that a PSD matrix $\mA$ has a unique PSD matrix $\mB$ such that $\mB^k = \mA$ for any $k\geq 1$~\cite[Theorem 7.2.6]{horn2012matrix}, we obtain
\begin{align*}
\mU^\T\mU = \mV^\T\mV
\end{align*}
for any critical point $\mW$.

\section{Proof of Lemma \ref{lem:nonsymmetric to symmetric}}\label{sec:proof lem nonsymmetric to symmetric}
We first rewrite the objective function $g(\mW)$:
\begin{align*}
&g(\mW) = \frac{1}{2}\left\|\mU\mV^\T - \mU^\star\mV^{\star\T}\right\|_F^2 +  \frac{\mu}{4} \left\|\mU^\T\mU - \mV^\T \mV\right\|_F^2\\
& \geq \min\{\mu,\frac{1}{2}\}\left(\left\|\mU\mV^\T - \mU^\star\mV^{\star\T}\right\|_F^2 + \frac{1}{4} \left\|\mU^\T\mU - \mV^\T \mV\right\|_F^2\right)\\
& = \min\{\mu,\frac{1}{2}\}\left(\frac{1}{4}\left\|\mW\mW^\T - \mW^\star\mW^{\star\T}\right\|_F^2 + g'(\mW)\right),
\end{align*}
where the second line attains the equality when $\mu = \frac{1}{2}$, and $g'(\mW)$ in the last line is defined as
\begin{align*}
g'(\mW) :=& \frac{1}{2}\left\|\mU\mV^\T - \mU^\star\mV^{\star\T}\right\|_F^2  - \frac{1}{4}\left\|\mU\mU^\T - \mU^\star\mU^{\star\T}\right\|_F^2 \\
 &- \frac{1}{4}\left\|\mV\mV^\T - \mV^\star\mV^{\star\T}\right\|_F^2 + \frac{1}{4} \left\|\mU^\T\mU - \mV^\T \mV\right\|_F^2.
\end{align*}
We further show $g'(\mW)$ is always nonnegative:
\begin{align*}
&g'(\mW) = \frac{1}{2}\left\|\mU\mV^\T - \mU^\star\mV^{\star\T}\right\|_F^2 - \frac{1}{4}\left\|\mU\mU^\T- \mU^\star\mU^{\star\T}\right\|_F^2 \\
 &\quad\quad\quad   - \frac{1}{4}\left\|\mV\mV^\T - \mV^\star\mV^{\star\T}\right\|_F^2 + \frac{1}{4} \left\|\mU^\T\mU - \mV^\T \mV\right\|_F^2.\\
& = \frac{1}{2}\left\|\mU\mV^\T - \mU^\star\mV^{\star\T}\right\|_F^2 + \frac{1}{2}\left\|\mU^\T\mU^\star \right\|_F^2  + \frac{1}{2}\left\|\mV^\T\mV^\star \right\|_F^2 \\
&\quad- \frac{1}{2}\trace\left(\mU^\T\mU\mV^\T\mV\right) -\frac{1}{4}\left\|\mU^\star\mU^{\star\T}\right\|_F^2 - \frac{1}{4}\left\|\mV^\star\mV^{\star\T}\right\|_F^2\\
& = \frac{1}{2}\left\|\mU^\T\mU^\star - \mV^\T\mV^\star\right\|_F^2 + \frac{1}{2}\left\|\mU^\star\mV^{\star\T}\right\|_F^2 \\ &\quad-\frac{1}{4}\left\|\mU^\star\mU^{\star\T}\right\|_F^2 - \frac{1}{4}\left\|\mV^\star\mV^{\star\T}\right\|_F^2\\
& = \frac{1}{2}\left\|\mU^\T\mU^\star - \mV^\T\mV^\star\right\|_F^2 \geq 0,
\end{align*}
where the last line follows because $\mU^{\star\T}\mU^{\star} = \mV^{\star\T}\mV^{\star}$.
Thus, we have
\begin{align*}
g(\mW) \geq \min\{\frac{\mu}{4},\frac{1}{8}\}\left\|\mW\mW^\T - \mW^\star\mW^{\star\T}\right\|_F^2,
\end{align*}
and
\begin{align*}
g(\mW) = \frac{1}{8}\left\|\mW\mW^\T - \mW^\star\mW^{\star\T}\right\|_F^2 + \frac{1}{4} \left\|\mU^\T\mU^\star - \mV^\T\mV^\star\right\|_F^2
\end{align*}
if $\mu = \frac{1}{2}$.

\section{Proof of Lemma~\ref{lem:concequence of balanced factors}}\label{sec:proof concequence of balanced factors}
Utilizing the result that any point $\mW\in\calE$ satisfies $\widehat\mW^\T\mW = \mU^\T\mU - \mV^\T\mV = \mzero$, we directly obtain
\begin{align*}
\|\mDelta_{\mU}\mU^\T\|_F^2 + \|\mDelta_{\mV}\mV^\T\|_F^2 = \|\mDelta_{\mU}\mV^\T\|_F^2 + \|\mDelta_{\mV}\mU^\T\|_F^2
\end{align*}
since $\|\mDelta_{\mU}\mU^\T\|_F^2 = \trace\left(\mDelta_{\mU}\mU^\T\mU\mDelta_{\mU}\right) = \trace\left(\mDelta_{\mU}\mV^\T\mV\mDelta_{\mU}\right) = \|\mDelta_{\mU}\mV^\T\|_F^2$ (and similarly for the other two terms).

We then rewrite the last two terms in~\eqref{eq:hessian for regularizer} as
\begin{align*}
&\left\langle \widehat\mW\widehat\mDelta^\T,\mDelta \mW^\T  \right\rangle + \left\langle \widehat\mW\widehat\mW^\T,\mDelta \mDelta^\T\right\rangle\\
& = \left\langle \widehat\mW^\T\mDelta,\mDelta^\T\widehat\mW\right\rangle + \left\langle \widehat\mW^\T\mDelta,\widehat\mW^\T\mDelta\right\rangle\\
& = \left\langle \widehat\mW^\T\mDelta,\widehat\mW^\T\mDelta + \mDelta^\T\widehat\mW\right\rangle\\
& = \frac{1}{2}\left\langle \widehat\mW^\T\mDelta + \mDelta^\T\widehat\mW,\widehat\mW^\T\mDelta + \mDelta^\T\widehat\mW\right\rangle \\
&\quad+ \frac{1}{2} \left\langle \widehat\mW^\T\mDelta - \mDelta^\T\widehat\mW,\widehat\mW^\T\mDelta + \mDelta^\T\widehat\mW\right\rangle\\
& = \frac{1}{2}\left\| \widehat\mW^\T\mDelta + \mDelta^\T\widehat\mW\right\|_F^2,
\end{align*}
where the last line holds because $\left\langle \mA - \mA^\T, \mA + \mA^\T\right\rangle = 0$.
Plugging these with the factor $\widehat\mW^\T\mW  = \mzero$ into the Hessian quadratic form  $[\nabla^2 \rho(\mW)](\mDelta,\mDelta)$ defined in~\eqref{eq:hessian for regularizer} gives
\begin{align*}
[\nabla^2 \rho(\mW)](\mDelta,\mDelta) \geq  \frac{\mu}{2}\left\| \widehat\mW^\T\mDelta + \mDelta^\T\widehat\mW\right\|_F^2\geq 0.
\end{align*}
This implies that the Hessian of $\rho$ evaluated at any $\mW\in\calE$ is PSD, i.e., $\nabla^2 \rho(\mW)\succeq \mzero$.\footnote{This can also be observed since any critical point $\mW$ is a global minimum of $\rho(\mW)$, which directly indicates that $\nabla^2\rho(\mW)\succeq \mzero$.}

\section{Proof of Theorem~\ref{thm:strict saddle property} (strict saddle property for \eqref{eq:low rank approx fact})}\label{sec:proof strict saddle property}
We begin the proof of Theorem~\ref{thm:strict saddle property} by characterizing  any $\mW\in\calC\setminus\calX$. For this purpose, let $\mW = \begin{bmatrix} \mU \\ \mV \end{bmatrix}$, where $\mU = \mPhi\mLambda^{1/2}\mR,\mV = \mPsi\mLambda^{1/2}\mR, \mR\in\calO_r,  \mLambda \textup{ is diagonal}, \mLambda\geq \mzero, (\mSigma - \mLambda)\mSigma = \mzero$, and $\rank(\mLambda)<r$. Denote the corresponding optimal solution $\mW^\star = \begin{bmatrix} \mU^\star \\ \mV^\star \end{bmatrix}$, where $\mU^\star = \mPhi\mSigma^{1/2}\mR,\mV^\star = \mPsi\mSigma^{1/2}\mR$. Let
\begin{align*}
k=\argmax_{i}\sigma_i - \lambda_i
\end{align*}
 denote the location of the first zero diagonal element in $\mLambda$. Noting that $\lambda_i\in\{\sigma_i,0\}$, we conclude that
 \begin{align}\label{eq:strict saddle k property}
 \lambda_k = 0, \quad \vphi_k^\T \mU = \vzero, \quad \vpsi_k^\T \mV = \vzero.
 \end{align}
In words, $\vphi_k$ and $\vpsi_k$ are orthogonal to $\mU$ and $\mV$, respectively. Let $\valpha \in\R^r$ be the eigenvector associated with the smallest eigenvalue of $\mW^\T\mW$. Such $\valpha$ simultaneously lives in the null spaces of $\mU$ and $\mV$ since $\mW$ is rank deficient indicating
 \begin{align*}
0 =\valpha^\T \mW^\T\mW \valpha =  \valpha^\T \mU^\T\mU \valpha +  \valpha^\T \mV^\T\mV \valpha, \end{align*}
 which further implies
  \begin{align}
 \left\{\begin{matrix}\valpha^\T \mU^\T\mU \valpha = 0,\\ \valpha^\T \mV^\T\mV \valpha = 0. \end{matrix}\right.
 \label{eq:strict saddle alpha}\end{align}
With this property, we construct $\mDelta$ by setting $\mDelta_{\mU} = \vphi_k\valpha^\T$ and $\mDelta_{\mV} = \vpsi_k\valpha^\T$. Now we show that $\mW$ is a strict saddle by arguing that $g(\mW)$ has a strictly negative curvature along the constructed direction $\mDelta$, i.e., $[\nabla^2g(\mW)](\mDelta,\mDelta)<0$. To that end, we compute the five terms in \eqref{eq:hessian} as follows
\begin{align*}
\left\|\mDelta_{\mU}\mV^\T+ \mU\mDelta_{\mV}^\T\right\|_F^2 & = 0 \quad (\textup{since} \quad  \eqref{eq:strict saddle alpha}),\\
\left\langle \mU\mV^\T - \mX^\star,\mDelta_{\mU}\mDelta_{\mV}^\T \right\rangle &= \lambda_k - \sigma_k = -\sigma_k \quad (\textup{since} \quad \eqref{eq:strict saddle k property}),\\
\left\langle \widehat\mW^\T\mW,\widehat\mDelta^\T\mDelta \right\rangle &= 0 \quad (\textup{since} \quad \widehat\mW^\T\mW = \mzero),\\
\left\langle \widehat\mW\widehat\mDelta^\T,\mDelta \mW^\T  \right\rangle &= \trace\left( \widehat\mDelta^\T \mW\mDelta^T \widehat\mW \right) = 0,\\
\left\langle \widehat\mW\widehat\mW^\T,\mDelta \mDelta^\T\right\rangle &= \trace\left( \widehat\mW^\T\mDelta\mDelta^\T\widehat\mW\right) = 0,
\end{align*}
where $\widehat\mW^\T\mW = \mzero$ since $\mU^\T\mU -\mV^\T\mV = \mzero$, the last two lines utilize $\widehat\mDelta^\T \mW = \mzero$ (or $\widehat\mW^\T\mDelta = \mzero$) because $\widehat\mDelta^\T \mW = \valpha\vphi_k^\T\mU -\valpha\vpsi_k^\T\mV = \mzero$ (see \eqref{eq:strict saddle k property}). Plugging these terms into~\eqref{eq:hessian} gives
\begin{align*}
&[\nabla^2g(\mW)](\mDelta,\mDelta) \\
&= \left\|\mDelta_{\mU}\mV^\T+ \mU\mDelta_{\mV}^\T\right\|_F^2 + 2\left\langle \mU\mV^\T - \mX^\star,\mDelta_{\mU}\mDelta_{\mV}^\T \right\rangle\\ &\quad+  \mu\left\langle \widehat\mW^\T\mW,\widehat\mDelta^\T\mDelta \right\rangle + \mu\left\langle \widehat\mW\widehat\mDelta^\T,\mDelta \mW^\T  \right\rangle\\
 &\quad + \mu\left\langle \widehat\mW\widehat\mW^\T,\mDelta \mDelta^\T\right\rangle
\\
& = -2\sigma_k.
\end{align*}
The proof of the strict saddle property is completed by noting that
\[
\left\|\mDelta\right\|_F^2 = \left\|\mDelta_{\mU}\right\|_F^2 +  \left\|\mDelta_{\mV}\right\|_F^2 =  \left\|\vphi_k\valpha^\T\right\|_F^2 + \left\|\vpsi_k\valpha^\T\right\|_F^2 =  2,
\]
which further implies
\begin{align*}
\lambda_{\min}\left(\nabla^2g(\mW)\right)&\leq \frac{[\nabla^2g(\mW)](\mDelta,\mDelta)}{\left\|\mDelta\right\|_F^2}\leq -\frac{2\sigma_k}{2}\\& =  -\|\mLambda - \mSigma\|= -\frac{1}{2}\left\|\mW\mW^\T - \mW^\star\mW^{\star\T}\right\|,
\end{align*}
where the first equality holds because
\[
\|\mLambda - \mSigma\| = \max_{i}\sigma_i - \lambda_i = \sigma_k,
\]
and the second equality follows since
\begin{align*}
&\mW\mW^\T - \mW^\star\mW^{\star\T} = \frac{1}{2} \mQ\left( \mLambda - \mSigma \right) \mQ^\T,\\
& \mQ = \begin{bmatrix}\mPhi/\sqrt{2}\\\mPsi/\sqrt{2} \end{bmatrix}, \quad \mQ^\T\mQ = \mId.
\end{align*}
We finish the proof of~\eqref{eq:strict saddle property} by noting that
\[
\sigma_k = \sigma_k(\mX^\star)\geq \sigma_r(\mX^\star).
\]

Now suppose $\mW^\star\in\calX$. Applying \eqref{eq:hessian regularizer psd}, which states that the Hessian of $\rho$ evaluated at any critical point $\mW$ is PSD, we have
\begin{align*}
&[\nabla^2g(\mW^\star)][\mDelta,\mDelta] \\
&= \left\|\mDelta_{\mU}\mV^{\star\T}+ \mU^\star\mDelta_{\mV}^\T\right\|_F^2 + 2\left\langle \mU^\star\mV^{\star\T} - \mX^\star,\mDelta_{\mU}\mDelta_{\mV}^\T \right\rangle \\&\quad+ [\nabla^2\rho(\mW^\star)][\mDelta,\mDelta]\\
& \geq \left\|\mDelta_{\mU}\mV^{\star\T}+ \mU^\star\mDelta_{\mV}^\T\right\|_F^2 + 2\left\langle \mU^\star\mV^{\star\T}- \mX^\star,\mDelta_{\mU}\mDelta_{\mV}^\T \right\rangle \\
&\geq 0
\end{align*}
since $\mU^\star\mV^{\star\T}- \mX^\star = \mzero$.
We show $g$ is not strongly convex at $\mW^\star$ by arguing that $\lambda_{\min}(\nabla^2 g(\mW^\star)) = 0$. For this purpose, we first recall that $\mU^\star = \mPhi\mSigma^{1/2},\mV^\star = \mPsi\mSigma^{1/2}$, where we assume $\mR = \mId$ without loss of generality. Let $\{\ve_1,\ve_2,\ldots,\ve_r\}$ be the standard orthobasis for $\R^r$, i.e., $\ve_\ell$ is the $\ell$-th column of the $r\times r$ identity matrix.  Construct $\mDelta_{(i,j)} = \begin{bmatrix}\mDelta_{\mU}^{(i,j)}\\ \mDelta_{\mV}^{(i,j)}\end{bmatrix}$, where
\[
\mDelta_{\mU}^{(i,j)} = \mU^\star\ve_j\ve_i^\T - \mU^\star\ve_i\ve_j^\T, \quad \mDelta_{\mV}^{(i,j)} = \mV^\star\ve_j\ve_i^\T - \mU^\star\ve_i\ve_j^\T,
\]
for any $1\leq i<j\leq r$. That is, the $\ell$-th columns of the matrices $\mDelta_{\mU}^{(i,j)}$ and $\mDelta_{\mV}^{(i,j)}$ are respectively given by
\begin{align*}
\mDelta_{\mU}^{(i,j)}[:,\ell]  = \begin{cases}\sigma^{1/2}_j\vphi_j, & \ell = i,\\ - \sigma^{1/2}_i\vphi_i, & \ell = j,\\ \vzero, &\textup{otherwise}, \end{cases}, \\ \mDelta_{\mV}^{(i,j)}[:,\ell]  = \begin{cases}\sigma^{1/2}_j\vpsi_j, & \ell = i,\\ - \sigma^{1/2}_i\vpsi_i, & \ell = j,\\ \vzero, &\textup{otherwise}, \end{cases}
\end{align*}
for any $1\leq i<j\leq r$.  We then compute the five terms in \eqref{eq:hessian} as follows
\begin{align*}
&\left\|\mDelta_{\mU}^{(i,j)}\mV^{\star\T}+ \mU^\star(\mDelta_{\mV}^{(i,j)})^\T\right\|_F^2\\ &= \left\|\sigma_i^{1/2}\sigma_j^{1/2}\left(\vphi_j\vpsi_i^\T - \vphi_i\vpsi_j^\T + \vphi_i\vpsi_j^\T - \vphi_j\vpsi_i^\T \right)\right\|_F^2 = 0,
\end{align*}
\begin{align*}
&\langle \mU^\star\mV^{\star\T} - \mX^\star,\mDelta_{\mU}^{(i,j)}(\mDelta_{\mV}^{(i,j)})^\T \rangle = 0 \ (\textup{as}\ \mU^\star\mV^{\star\T} - \mX^\star = \mzero),\\
&\langle \widehat\mW^{\star\T}\mW^\star,\widehat\mDelta_{(i,j)}^\T\mDelta_{(i,j)} \rangle = 0 \quad (\textup{as} \ \widehat\mW^{\star\T}{\mW^\star} = \mzero),\\
&\langle \widehat\mW^\star \widehat\mDelta_{(i,j)}^\T,\mDelta_{(i,j)} \mW^{\star\T} \rangle = \trace( \widehat\mW^{\star\T}\mDelta_{(i,j)}\mW^{\star\T}\widehat\mDelta_{(i,j)} ) = 0,\\
&\langle \widehat\mW^\star\widehat\mW^{\star\T},\mDelta_{(i,j)} \mDelta_{(i,j)}^\T\rangle = \trace( \widehat\mW^{\star\T}\mDelta_{(i,j)}\mDelta_{(i,j)}^\T\widehat\mW^\star) = 0 ,
\end{align*}
where the last two lines  hold because
\begin{align*}
\widehat\mW^{\star\T}\mDelta_{(i,j)} &= \mU^{\star\T}\mU^\star(\ve_j\ve_i^\T - \ve_i\ve_j^\T) - \mV^{\star\T}\mV^\star(\ve_j\ve_i^\T - \ve_i\ve_j^\T) \\ &= \mzero
 \end{align*}
since $\mU^{\star\T}\mU^\star = \mV^{\star\T}\mV^\star$.

Thus, we obtain the Hessian evaluated at the optimal solution point $\mW^\star$ along the direction $\mDelta^{(i,j)}$:
\begin{align*}
&\left[\nabla^2g(\mW^\star)\right]\left(\mDelta^{(i,j)},\mDelta^{(i,j)}\right)= 0
\end{align*}
for all $1\leq i<j\leq r$. This proves that $g(\mW)$ is not strongly convex at a global minimum point $\mW^\star\in\calX$.

\section{Proof of Theorem \ref{thm:strict saddle property over} (strict saddle property of $g(\mW)$ when over-parameterized)}\label{sec:proof strict saddle property over}
Let $\mX^\star = \mPhi\mSigma\mPsi^\T = \sum_{i=1}^{r'}\sigma_i\vphi_i\vpsi_i^\T$ be a reduced SVD of $\mX^\star$ with $r'\leq r$. Using an approach similar to that in Appendix~\ref{sec:proof lem set of critical points} for proving Lemma~\ref{lem:set of critical points}, we can show that
any $\mW = \begin{bmatrix} \mU \\ \mV \end{bmatrix}$ is a critical point of $g(\mW)$ if and only if $\mW\in\calC$ with
\begin{align*}
\calC = \bigg\{ &\mW = \begin{bmatrix} \mU \\ \mV \end{bmatrix}: \mU = \mPhi\mLambda^{1/2}\mR,\mV = \mPsi\mLambda^{1/2}\mR, \mR\mR^\T = \mId_{r'},  \\&\mLambda \textup{ is diagonal},\mLambda\geq \mzero, (\mSigma - \mLambda)\mSigma = \mzero \bigg\}.
\end{align*}
Recall that
\begin{align*}
\calX =\bigg\{\mW = \begin{bmatrix} \mU \\ \mV \end{bmatrix}: &\mU = \mPhi\mSigma^{1/2}\mR,\mV = \mPsi\mSigma^{1/2}\mR, \\&\mR\mR^\T = \mId_{r'}\bigg\}.
\end{align*}
It is clear that $\calX$ is the set of optimal solutions since for any $\mW\in\calX$, $g(\mW)$ achieves its global minimum, i.e., $g(\mW) = 0$.

Using an approach similar to that in Appendix~\ref{sec:proof strict saddle property} for proving Theorem~\ref{thm:strict saddle property}, we can show that any $\mW\in\calC\setminus\calX$  is a strict saddle satisfying
\begin{align*}
\lambda_{\min}\left(\nabla^2g(\mW)\right) \leq - \sigma_{r'}(\mX^\star).
\end{align*}

\section{Proof of Theorem \ref{thm:strict saddle property under} (strict saddle property of $g(\mW)$ when under-parameterized)}\label{sec:proof strict saddle property under}
Let $\mX^\star = \mPhi\mSigma\mPsi^\T = \sum_{i=1}^{r'}\sigma_i\vphi_i\vpsi_i^\T$ be a reduced SVD of $\mX^\star$ with $r'> r$ and $\sigma_r(\mX^\star)>\sigma_{r+1}(\mX^\star)$.
Using an approach similar to that in Appendix~\ref{sec:proof lem set of critical points} for proving Lemma~\ref{lem:set of critical points}, we can show that any $\mW = \begin{bmatrix} \mU \\ \mV \end{bmatrix}$ is a critical point of $g(\mW)$ if and only if $\mW\in\calC$ with
\begin{align*}
\calC = &\bigg\{\mW = \begin{bmatrix} \mU \\ \mV \end{bmatrix}: \mU = \mPhi[:,\Omega]\mLambda^{1/2}\mR,\mV = \mPsi[:,\Omega]\mLambda^{1/2}\mR,\\
 &\mLambda =\mSigma[\Omega,\Omega], \mR\mR^\T = \mId_{\ell}, \Omega\subset\{1,\ldots,r'\}, |\Omega| = \ell\leq r \bigg\}.
\end{align*}
Intuitively, a critical point is one such that $\mU\mV^\T$ is a rank-$\ell$ approximation to $\mX^\star$ with $\ell\leq r$ and $\mU$ and $\mV$ are equal factors of their product $\mU\mV^\T$.

It follows from the Eckart-Young-Mirsky theorem~\cite{horn2012matrix} that the set of optimal solutions is given by
\begin{align*}
\calX &=\bigg\{\mW = \begin{bmatrix} \mU \\ \mV \end{bmatrix}: \mU = \mPhi[:,1:r]\mLambda^{1/2}\mR,\\&\quad\quad\mV = \mPsi[:,1:r]\mLambda^{1/2}\mR, \mLambda = \mSigma[1:r,1:r], \mR\in\calO_r \bigg\}.
\end{align*}
Now we characterize  any $\mW\in\calC\setminus\calX$ by letting $\mW = \begin{bmatrix} \mU \\ \mV \end{bmatrix}$, where
\begin{align*}
& \mU = \mPhi[:,\Omega]\mLambda^{1/2}\mR,\mV = \mPsi[:,\Omega]\mLambda^{1/2}\mR,\\ &\mLambda =\mSigma[\Omega,\Omega],\mR\in\R^{\ell\times r},\mR\mR^\T = \mId_{\ell},\\ &\Omega\subset\{1,2,\ldots,r'\}, |\Omega| = \ell\leq r, \Omega \neq \{1,2,\ldots,r\}.
 \end{align*}

Let $\valpha\in\R^r$ be the eigenvector associated with the smallest eigenvalue of $\mU^\T\mU$ (or $\mV^\T\mV$). By the typical structures in $\mU$ and $\mV$ (see the above equation), we have
\begin{equation}\label{eq:prf strict saddle under}\begin{split}
\|\mV\valpha\|_F^2 &= \|\mU\valpha\|_F^2 = \sigma_r^2(\mU) \\&= \left\{\begin{matrix}\sigma_j(\mX^\star),& |\Omega| = r ~\text{and}~ j= \max \Omega\\
0,& |\Omega| < r,  \end{matrix}\right.
\end{split}\end{equation}
where $j>r$ because $\Omega \neq \{1,2,\ldots,r\}$.
Note that there always exists an index
\[
i\in\{1,2,\ldots,r\},i\neq \Omega
\]
since $\Omega \neq \{1,2,\ldots,r\}$ and $|\Omega|\leq r$.
We construct $\mDelta$ by setting
\[
\mDelta_{\mU} = \vphi_i\valpha^\T,\quad \mDelta_{\mV} = \vpsi_i\valpha^\T.
\]
Since $i\notin \Omega$, we have
\begin{equation}\begin{split}
\mU^\T\mDelta_{\mU}&= \mU^\T \vphi_i\valpha^\T = \mzero,\\
\mV^\T\mDelta_{\mV}&= \mV^\T \vpsi_i\valpha^\T = \mzero.
\end{split}
\label{eq:prf strict saddle 2 under}\end{equation}
We compute the five terms in \eqref{eq:hessian} as follows
\begin{align*}
&\left\|\mDelta_{\mU}\mV^\T+ \mU\mDelta_{\mV}^\T\right\|_F^2\\ &= \left\|\mDelta_{\mU}\mV^\T\right\|_F^2 + \left\| \mU\mDelta_{\mV}^\T\right\|_F^2 + 2\trace\left(\mU^\T\mDelta_{\mU}\mV^\T\mDelta_{\mV}\right)\\& = 2\sigma_r^2(\mU),
\end{align*}
\begin{align*}
\left\langle \mU\mV^\T - \mX^\star,\mDelta_{\mU}\mDelta_{\mV}^\T \right\rangle &= \left\langle \mU\mV^\T - \mX^\star,\vphi_i\vpsi_i^\T \right\rangle\\&= -\left\langle  \mX^\star,\vphi_i\vpsi_i^\T \right\rangle = -\sigma_i(\mX^\star),\\
\left\langle \widehat\mW^\T\mW,\widehat\mDelta^\T\mDelta \right\rangle &= 0 \quad (\textup{since} \quad \widehat\mW^\T\mW = \mzero),\\
\left\langle \widehat\mW\widehat\mDelta^\T,\mDelta \mW^\T  \right\rangle &= \trace\left( \widehat\mW^\T\mDelta\mW^\T\widehat\mDelta  \right) = 0,\\
\left\langle \widehat\mW\widehat\mW^\T,\mDelta \mDelta^\T\right\rangle &= \trace\left( \widehat\mW^\T\mDelta\mDelta^\T\widehat\mW\right) = 0,
\end{align*}
where the last equality in the first line holds because $\mU^\T\mDelta_{\mU} = \mzero$ (see \eqref{eq:prf strict saddle 2 under}) and $\left\|\mDelta_{\mU}\mV^\T\right\|_F^2 =\left\| \mU\mDelta_{\mV}^\T\right\|_F^2 = \sigma_r^2(\mU)$ (see \eqref{eq:prf strict saddle under}),
$\widehat\mW^\T\mW = \mzero$ in the third line holds since $\mU^\T\mU -\mV^\T\mV = \mzero$, and $\widehat\mW^\T\mDelta = \mzero$ in the fourth and last lines holds because
\[
\widehat\mW^\T\mDelta = \mU^\T\mDelta_{\mU} - \mV^\T\mDelta_{\mV} = \mzero.
\]
Now plugging these terms into \eqref{eq:hessian} yields
\begin{align*}
&[\nabla^2g(\mW)](\mDelta,\mDelta) \\
&= \|\mDelta_{\mU}\mV^\T+ \mU\mDelta_{\mV}^\T\|_F^2 + 2\langle \mU\mV^\T - \mX^\star,\mDelta_{\mU}\mDelta_{\mV}^\T \rangle\\&+  \mu(\langle \widehat\mW^\T\mW,\widehat\mDelta^\T\mDelta \rangle + \langle \widehat\mW\widehat\mDelta^\T,\mDelta \mW^\T  \rangle + \langle \widehat\mW\widehat\mW^\T,\mDelta \mDelta^\T\rangle
)\\
& = -2(\sigma_i(\mX^\star) - \sigma_r^2(\mU)).
\end{align*}
The proof of the strict saddle property is completed by noting that
\[
\|\mDelta\|_F^2 = \|\mDelta_{\mU}\|_F^2 + \|\mDelta_{\mV}\|_F^2 = 2,
\]
which further implies
\begin{align*}
\lambda_{\min}\left(\nabla^2g(\mW)\right)&\leq -2\frac{\sigma_i(\mX^\star) - \sigma_r^2(\mU)}{\|\mDelta\|_F^2} \\&\leq -\left(\sigma_r(\mX^\star) - \sigma_{r+1}(\mX^\star)\right),
\end{align*}
where the last inequality holds because of \eqref{eq:prf strict saddle under} and because $i\leq r$.

\section{Proof of Theorem \ref{thm:robust strict saddle property} (robust strict saddle for $g(\mW)$)}\label{sec:prf robust strict saddle property}

We first establish the following useful results.
\begin{lem}\label{lem:trace inequality for two PSD}
For any two PSD matrices $\mA,\mB\in\R^{n\times n}$, we have
\[
\sigma_n(\mA)\trace(\mB)\leq \trace\left(\mA\mB\right) \leq \|\mA\|\trace(\mB).
\]
\begin{proof}[Proof of Lemma~\ref{lem:trace inequality for two PSD}]
Let $\mA = \mPhi_1\mLambda_1\mPhi_1^\T$ and $\mB = \mPhi_2\mLambda_2\mPhi_2^\T$ be the eigendecompositions of $\mA$ and $\mB$, respectively. Here $\mLambda_1$ ($\mLambda_2$) is a diagonal matrix with the eigenvalues of $\mA$ ($\mB$) along its diagonal. We first rewrite $\trace\left(\mA\mB\right)$ as
\begin{align*}
\trace\left(\mA\mB\right) = \trace\left(\mLambda_1\mPhi_1^\T \mPhi_2\mLambda_2\mPhi_2^\T \mPhi_1\right).
\end{align*}
Noting that $\mLambda_1$ is a diagonal matrix, we have
\begin{align*}
&\trace\left(\mLambda_1\mPhi_1^\T \mPhi_2\mLambda_2\mPhi_2^\T \mPhi_1\right) \\ &\geq \min_{i}\mLambda_1[i,i]\cdot\trace\left(\mPhi_1^\T \mPhi_2\mLambda_2\mPhi_2^\T \mPhi_1\right) \\&= \sigma_n(\mA)\trace(\mB).
\end{align*}
The other direction follows similarly.
\end{proof}
\end{lem}
\begin{cor}\label{cor:inequality for Frobunius of products}
For any two matrices $\mA\in\R^{n\times r}$ and $\mB\in\R^{r\times r}$, we have
\[
\sigma_r(\mB)\|\mA\|_F\leq \left\|\mA\mB\right\|_F \leq \|\mB\|\|\mA\|_F.
\]
\end{cor}

We provide one more result before proceeding to prove the main theorem.
\begin{lem}\label{lem:regularity property for PSD}
Suppose $\mA,\mB\in\R^{n\times r}$ such that $\mA^\T\mB = \mB^\T\mA\succeq \mzero$ is PSD. If $\|\mA - \mB\|\leq \frac{\sqrt{2}}{2}\sigma_r(\mB)$, we have
\begin{equation}\label{eq:regularity property for PSD_1}\begin{split}
&\underbrace{\left\langle \left(\mA\mA^\T - \mB\mB^\T\right)\mA, \mA - \mB\right\rangle}_{(\aleph_1)}\\
 &\geq \frac{1}{16}(\underbrace{\trace((\mA - \mB)^\T(\mA - \mB)\mB^\T\mB)}_{(\aleph_2)} + \underbrace{\|\mA\mA^\T - \mB\mB^\T\|_F^2}_{(\aleph_3)}).
\end{split}\end{equation}

\end{lem}

\begin{proof}
Denote $\mE = \mA - \mB$. We first rewrite the terms $(\aleph_1)$, $(\aleph_2)$ and $(\aleph_3)$ as follows
\begin{align*}
(\aleph_1) & = \trace\bigg(\left(\mE^\T\mE\right)^2 + 3\mE^\T\mE\mE^\T\mB + \left(\mE^\T\mB\right)^2\\
 &\quad\quad\quad\quad + \mE^\T\mE\mB^\T\mB\bigg),\\
(\aleph_2) & = \trace\left(\mE^\T\mE\mB^\T\mB\right),\\
(\aleph_3) & = \trace\bigg(\left(\mE^\T\mE\right)^2 + 4\mE^\T\mE\mE^\T\mB + 2\left(\mE^\T\mB\right)^2 \\
&\quad\quad\quad\quad
+ 2\mE^\T\mE\mB^\T\mB\bigg),
\end{align*}
where $\mE^\T\mB = \mA^\T\mB - \mB^\T\mB = \mB^\T\mE$. Now we have
\begin{align*}
&(\aleph_1) - \frac{1}{16}(\aleph_2) - \frac{1}{16}(\aleph_3)\\
&= \trace\bigg( \frac{15}{16} \left(\mE^\T\mE\right)^2 + \frac{11}{4}\mE^\T\mE\mE^\T\mB + \frac{7}{8}\left(\mE^\T\mB\right)^2 \\
&\quad\quad\quad\quad+ \frac{13}{16}\mE^\T\mE\mB^\T\mB\bigg)\\
& = \left\|\sqrt{\frac{121}{56}}\mE^\T\mE + \sqrt{\frac{7}{8}}\mE^\T\mB\right\|_F^2\\
&\quad +\trace\left( \frac{13}{16}\mE^\T\mE\mB^\T\mB - \frac{137}{112}\mE^\T\mE\mE^\T\mE\right)\\
& \geq \trace\left( \frac{13}{16}\mE^\T\mE\sigma_r^2(\mB) - \frac{137}{112}\mE^\T\mE\|\mE\|^2\right)\\
& \geq \trace\left( \left(\frac{13}{16} - \frac{137}{112}\frac{1}{2}\right)\sigma_r^2(\mB)\mE^\T\mE \right)\\
& \geq 0,
\end{align*}
where the third line follows from Lemma~\ref{lem:trace inequality for two PSD} and the fourth line holds because by assumption $\|\mE\|\leq \frac{\sqrt{2}}{2}\sigma_r(\mB)$.
\end{proof}

Now we turn to prove the main results. Recall that $\mu = \frac{1}{2}$ throughout the proof.

\subsection{Regularity condition for the region $\calR_1$}\label{prf:locol sescent}
It follows from Lemma~\ref{lem:Procrustes problem} that $\mW^\T\mW^\star \mR = \mR^\T\mW^{\star\T}\mW$ is PSD, where $\mR = \argmin_{R'\in\O_r}\|\mW-\mW^\star \mR'\|_F^2$. We first perform the change of variable $\mW^\star\mR \rightarrow \mW^\star$ to avoid $\mR$ in the following equations. With this change of variable we have instead $\mW^\T\mW^\star = \mW^{\star\T}\mW$ is PSD. We now rewrite the gradient $\nabla g(\mW)$ as follows:
\begin{equation}\label{eq:local descent gradient g}\begin{split}
\nabla g(\mW)&= \begin{bmatrix}\mzero & \mU\mV^\T -\mU^\star\mV^{\star\T}\\ \mV\mU^\T -\mV^\star\mU^{\star\T} & \mzero \end{bmatrix}
\mW \\
&\quad+ \mu \widehat\mW(\widehat\mW^\T\mW)\\
& = \frac{1}{2}\left(\mW\mW^\T - \mW^\star\mW^{\star\T}\right)\mW + \frac{1}{2}\widehat\mW^\star\widehat\mW^{\star\T}\mW \\
&\quad+ (\mu - \frac{1}{2})\widehat\mW \widehat\mW^\T\mW\\
& = \frac{1}{2} \left(\mW\mW^\T - \mW^\star\mW^{\star\T}\right)\mW + \frac{1}{2} \widehat\mW^\star\widehat\mW^{\star\T}\mW.
\end{split}\end{equation}
Plugging this into the left hand side of \eqref{eq:thm regularity condition} gives
\begin{equation}\label{eq:proof descent main}
\begin{split}
&\left\langle \nabla g(\mW), \mW - \mW^\star\right\rangle \\
&= \frac{1}{2}  \left\langle \left(\mW\mW^\T - \mW^\star\mW^{\star\T}\right)\mW, \mW - \mW^\star\right\rangle \\
& \quad +  \frac{1}{2}  \left\langle \widehat\mW^\star\widehat\mW^{\star\T}\mW, \mW - \mW^\star\right\rangle\\
& = \frac{1}{2}  \left\langle \left(\mW\mW^\T - \mW^\star\mW^{\star\T}\right)\mW, \mW - \mW^\star\right\rangle \\
&\quad + \frac{1}{2}  \left\langle \widehat\mW^\star\widehat\mW^{\star\T}, \mW\mW^\T\right\rangle
\end{split}\end{equation}
where the last line follows from the fact that $\mW^{\star\T}\widehat\mW^\star = \mzero$. We first show the first term in the right hand side of the above equation is sufficiently large
\begin{equation}\begin{split}
&\left\langle \left(\mW\mW^\T - \mW^\star\mW^{\star\T}\right)\mW, \mW - \mW^\star\right\rangle \\
&\geq \frac{1}{16}\trace\left((\mW-\mW^\star)^\T(\mW-\mW^\star)\mW^{\star\T}\mW^\star\right) \quad\\
&+ \frac{1}{16}\left\|\mW\mW^\T - \mW^\star\mW^{\star\T}\right\|_F^2\\
& \geq \frac{1}{16}\sigma_r(\mW^{\star\T}\mW^\star)\left\|\mW-\mW^\star\right\|_F^2 \\&\quad+ \frac{1}{16}\left\|\mW\mW^\T - \mW^\star\mW^{\star\T}\right\|_F^2\\
& = \frac{1}{8}\sigma_r(\mX^\star)\left\|\mW-\mW^\star\right\|_F^2 + \frac{1}{16}\left\|\mW\mW^\T - \mW^\star\mW^{\star\T}\right\|_F^2,
\end{split}\label{eq:proof descent main second term 2}\end{equation}
where the first inequality follows from Lemma~\ref{lem:regularity property for PSD} since $\mW^\T\mW^{\star} = \mW^{\star\T}\mW$ is PSD and $\left\|\mW - \mW^\star\right\|\leq \sigma_r^{1/2}(\mX^\star) = \frac{\sqrt{2}}{2}\sigma_r(\mW^\star)$, the second inequality follows from Lemma~\ref{lem:trace inequality for two PSD}, and the last line holds because $\sigma_r\left(\widehat\mW^{\star\T}\widehat\mW^\star\right) = \sigma_r\left(\widehat\mU^{\star\T}\widehat\mU^\star + \widehat\mV^{\star\T}\widehat\mV^\star\right) = 2\sigma_r\left(\mSigma\right)=2\sigma_r\left(\mX^{\star}\right)$. We then show the second term in the right hand side of \eqref{eq:proof descent main} is lower bounded by
\begin{equation}\begin{split}\label{eq:proof descent main second term}
&\left\langle \widehat\mW^\star\widehat\mW^{\star\T}, \mW\mW^\T\right\rangle \\& =  \frac{1}{2\left\|\mX^{\star}\right\|}\left\|\widehat\mW^{\star\T}\widehat\mW^\star\right\|\trace\left(\widehat\mW^{\star\T}\mW\mW^\T\widehat\mW^\star\right)\\
& \geq \frac{1}{2\left\|\mX^{\star}\right\|} \trace\left( \widehat\mW^{\star\T}\widehat\mW^\star\widehat\mW^{\star\T}\mW\mW^\T\widehat\mW^\star\right) \\
& =  \frac{1}{2\left\|\mX^{\star}\right\|}\left\|\widehat\mW^\star\widehat\mW^{\star\T}\mW\right\|_F^2
\end{split}\end{equation}
where the first line holds because $\left\|\widehat\mW^{\star\T}\widehat\mW^\star\right\| = \left\|\widehat\mU^{\star\T}\widehat\mU^\star + \widehat\mV^{\star\T}\widehat\mV^\star\right\| = 2\left\|\mSigma\right\|=2\left\|\mX^{\star}\right\|$, and
the inequality follows from Lemma~\ref{lem:trace inequality for two PSD}.

On the other hand, we attempt to control the gradient of $g(\mW)$. To that end,
it follows from \eqref{eq:local descent gradient g} that
\begin{equation}\begin{split}\label{eq:proof descent bound gradient}
&\left\|\nabla g(\mW)\right\|_F^2\\ &= \frac{1}{4}\left\|   \left(\mW\mW^\T - \mW^\star\mW^{\star\T}\right)\mW + \widehat\mW^\star\widehat\mW^{\star\T}\mW \right\|_F^2\\
& \leq \frac{12}{47} \left\|   \left(\mW\mW^\T - \mW^\star\mW^{\star\T}\right)\mW \right\|_F^2 + 12  \left\| \widehat\mW^\star\widehat\mW^{\star\T}\mW \right\|_F^2\\
& \leq \frac{12}{47}\|\mW\|^2\left\|\mW\mW^\T - \mW^\star\mW^{\star\T} \right\|_F^2 + 12  \left\| \widehat\mW^\star\widehat\mW^{\star\T}\mW \right\|_F^2,
\end{split}\end{equation}
where the first inequality holds since $(a+b)^2\leq \frac{1+\epsilon}{\epsilon}a^2 + (1+\epsilon)b^2$ for any $\epsilon>0$.

Combining \eqref{eq:proof descent main}-\eqref{eq:proof descent bound gradient}, we can conclude the proof of \eqref{eq:thm regularity condition} as long as we can show the following inequality:
\begin{align*}
&\frac{1}{8}\left\|\mW\mW^\T - \mW^\star\mW^{\star\T}\right\|_F^2
 \\&\geq \frac{1}{47}\frac{\|\mW\|^2}{\|\mX^\star\|}\left\|\mW\mW^\T - \mW^\star\mW^{\star\T} \right\|_F^2 .
\end{align*}
To that end, we upper bound  $\|\mW\|$ as follows:
\begin{align*}
\left\|\mW\right\| &\leq \left\| \mW^\star\right\| + \left\|\mW - \mW^\star\right\|\\
& \leq  \sqrt{2}\sigma_1^{1/2}(\mX^\star) + \left\|\mW - \mW^\star\right\|_F\\
& \leq (\sqrt{2}+1)\sigma_1^{1/2}(\mX^\star)
\end{align*}
since $\|\mW^\star\| = \sqrt{2}\sigma_1^{(1/2)}(\mX^\star) $ and $\dist(\mW,\mW^\star)\leq \sigma_r^{(1/2)}(\mX^\star)$.
This completes the proof of \eqref{eq:thm regularity condition}.

\subsection{Negative curvature for the region $\calR_2$}
To show \eqref{eq:thm negative curvature}, we utilize a strategy similar to that used in Appendix~\ref{sec:proof strict saddle property} for proving the strict saddle property of $g(\mW)$ by constructing a direction $\mDelta$ such that the Hessian evaluated at $\mW$ along this direction is negative. For this purpose, denote
\begin{align}
\mQ = \begin{bmatrix} \mPhi/\sqrt{2} \\ \mPsi/\sqrt{2} \end{bmatrix},
\label{eq:define Q}\end{align}
where we recall that $\mPhi$ and $\mPsi$ consist of the left and right singular vectors of $\mX^\star$, respectively.
The optimal solution $\mW^\star$ has a compact SVD $\mW^\star = \mQ(\sqrt{2}\mSigma^{1/2})\mR $. For notational convenience, we denote $\overline\mSigma = 2\mSigma$, where $\overline\mSigma$ is a diagonal matrix whose diagonal entries in the upper left corner are $\overline\sigma_1, \ldots,\overline \sigma_r$.

In order to characterize the neighborhood near all strict saddles $\calC\setminus\calX$,
we consider  $\mW$ such that $\sigma_r(\mW)\leq \sqrt{\frac{3}{8}}\sigma_r^{1/2}(\mX^\star)$. 
\revise{Let $\Span(\mQ)$ be the column space of $\mQ$ and
\[
\vq:= \argmin_{\overline\vq\in \Span(\mQ),\|\overline\vq\|=1}\overline\vq^\top \mW\mW^\top \overline\vq.
\] 
Using the min-max principle for the singular value
\[
\sigma_r(\mW\mW^\top) = \max_{\dim(\calS)=r}\quad\min_{\overline\vq\in \calS,\|\overline\vq\|=1}\overline\vq^\top \mW\mW^\top \overline\vq
\]
where $\calS$ denotes a subspace in $R^{n+m}$ and $\dim(\calS)$ denotes its dimension, we have
\begin{align}
\vq^\top \mW\mW^\top \vq\leq \sigma_r^2(\mW)\leq \frac{3}{8}\sigma_r(\mX^\star).
\label{eq:near saddle lambda k}\end{align}
}
Let $\valpha\in\R^r$ be the eigenvector associated with the smallest eigenvalue of $\mW^\T\mW$.

Recall that $\mu = \frac{1}{2}$.
We show that the function $g(\mW)$ at $\mW$ has directional negative curvature along the direction
\begin{align}
\mDelta = \revise{\vq}\valpha^\T.
\label{eq:Delta for negative curvature}\end{align}
We repeat the Hessian evaluated at $\mW$ for $\mDelta$ as follows
\begin{align*}
&[\nabla^2g(\mW)](\mDelta,\mDelta) \\&= \underbrace{\left\|\mDelta_{\mU}\mV^\T+\mU\mDelta_{\mV}^\T\right\|_F^2}_{\Pi_1} + 2\underbrace{\left\langle \mU\mV^\T - \mX^\star,\mDelta_{\mU}\mDelta_{\mV}^\T \right\rangle}_{\Pi_2}\\
& \quad + \frac{1}{2}\underbrace{\left\langle \widehat\mDelta\widehat\mW^\T,\mDelta \mW^\T  \right\rangle}_{\Pi_3} + \frac{1}{2}\underbrace{\left\langle \widehat\mW\widehat\mDelta^\T,\mDelta \mW^\T  \right\rangle}_{\Pi_4} \\
&\quad+ \frac{1}{2}\underbrace{\left\langle \widehat\mW\widehat\mW^\T,\mDelta \mDelta^\T\right\rangle}_{\Pi_5}.
\end{align*}
The remaining part is to bound the five terms.

\noindent{\bf Bounding terms $\Pi_1$, $\Pi_3$ and $\Pi_4$:}
We first rewrite these three terms:
\begin{align*}
\Pi_1&= \|\mDelta_{\mU}\mV^\T\|_F^2+\|\mU\mDelta_{\mV}^\T\|_F^2  + 2 \left\langle\mU\mDelta_{\mV}^\T,\mDelta_{\mU}\mV^\T \right\rangle ,\\
\Pi_3 &= \left\langle \widehat\mDelta\widehat\mW^\T,\mDelta \mW^\T  \right\rangle = \|\mDelta_{\mU}\mU^\T\|_F^2 + \|\mDelta_{\mV}\mV^\T\|_F^2 \\&\quad\quad\quad\quad\quad\quad\quad\quad\quad\quad- \|\mDelta_{\mU}\mV^\T\|_F^2 - \|\mDelta_{\mV}\mU^\T\|_F^2,\\
\Pi_4 &= \left\langle \mU\mDelta_{\mU}^\T,\mDelta_{\mU}\mU^\T \right\rangle + \left\langle\mV\mDelta_{\mV}^\T,\mDelta_{\mV}\mV^\T \right\rangle \\
&\quad- 2 \left\langle\mU\mDelta_{\mV}^\T,\mDelta_{\mU}\mV^\T \right\rangle\\
& \leq  \|\mDelta_{\mU}\mU^\T\|_F^2 + \|\mDelta_{\mV}\mV^\T\|_F^2 -  2 \left\langle\mU\mDelta_{\mV}^\T,\mDelta_{\mU}\mV^\T \right\rangle,
\end{align*}
which implies
\begin{equation}\begin{split}
&\Pi_1 + \frac{1}{2} \Pi_3 +  \frac{1}{2}\Pi_4\\& \leq \|\mDelta_{\mU}\mV^\T\|_F^2+\|\mU\mDelta_{\mV}^\T\|_F^2 + \|\mDelta_{\mU}\mU^\T\|_F^2 + \|\mDelta_{\mV}\mV^\T\|_F^2  \\
&\quad - \frac{1}{2}\|\mDelta_{\mU}\mV^\T\|_F^2 - \frac{1}{2} \|\mDelta_{\mV}\mU^\T\|_F^2 + \left\langle\mU\mDelta_{\mV}^\T,\mDelta_{\mU}\mV^\T \right\rangle
\\& = \|\mW\mDelta^\T\|_F^2 -\frac{1}{2}\left\|\mDelta_{\mU}\mV^\T - \mU\mDelta_{\mV}^\T\right\|_F^2
\\& \leq \|\mW\mDelta^\T\|_F^2.
\end{split}\label{eq:near saddle 1_1}\end{equation}
Noting that $\mDelta^\T\mDelta = \valpha\revise{\vq^\T\vq}\valpha^\T = \valpha\valpha^\T$,
we now compute $\|\mW\mDelta^\T\|_F^2$ as
\begin{align*}
\|\mW\mDelta^\T\|_F^2 &= \trace\left(\mW^\T\mW\mDelta^\T\mDelta\right) = \trace\left(\mW^\T\mW\valpha\valpha^\T\right) \\&= \sigma_r^2(\mW).
\end{align*}
Plugging this into \eqref{eq:near saddle 1_1} gives
\begin{equation}\begin{split}
\Pi_1 + \frac{1}{2} \Pi_3 +  \frac{1}{2}\Pi_4& \leq \sigma_r^2(\mW).
\end{split}\label{eq:near saddle 1_2}\end{equation}

\noindent{\bf Bounding terms $\Pi_2$ and $\Pi_5$:}
To obtain an upper bound for the term $\Pi_2$, we first rewrite it as follows
\begin{align*}
\Pi_2& =\left\langle \mU\mV^\T - \mX^\star,\mDelta_{\mU}\mDelta_{\mV}^\T \right\rangle \\&= \frac{1}{2}\left\langle \begin{bmatrix}\mzero & \mU\mV^\T -\mU^\star\mV^{\star\T}\\ \mV\mU^\T -\mV^\star\mU^{\star\T} & \mzero \end{bmatrix}, \mDelta\mDelta^\T\right\rangle\\
& = \frac{1}{4}\left\langle \mW\mW^\T - \mW^\star\mW^{\star\T}, \mDelta\mDelta^\T\right\rangle - \frac{1}{4}\left\langle \widehat\mW\widehat\mW^\T \mDelta\mDelta^\T\right\rangle \\&\quad+
\frac{1}{4}\left\langle \widehat\mW^\star\widehat\mW^{\star\T}, \mDelta\mDelta^\T\right\rangle.
\end{align*}
 We then have
 \begin{equation}\begin{split}
 2\Pi_2 + \frac{1}{2}\Pi_5 &= \frac{1}{2}\left\langle \mW\mW^\T - \mW^\star\mW^{\star\T}, \mDelta\mDelta^\T\right\rangle  \\&\quad+
\frac{1}{2}\left\langle \widehat\mW^\star\widehat\mW^{\star\T}, \mDelta\mDelta^\T\right\rangle.
 \end{split}\label{eq:near saddle 2_1}\end{equation}
To bound these two terms in the above equation, we note that \revise{$\vq\in\Span(\mQ)$ and thus can be written as $\vq = \sum_{i=1}^r a_i \vq_i$ with $\sum_i a_i^2 = 1$. It follows that
\[
\mDelta\mDelta^\T = \sum_{i,j}a_ia_j\vq_i\vq_j^\T = \sum_{i,j}\frac{a_ia_j}{2}\begin{bmatrix}\vphi_i\vphi_j^\T & \vphi_i\vpsi_j^\T\\ \vpsi_i\vphi_j^\T & \vpsi_i\vpsi_j^\T\end{bmatrix}.
\]
Then we have
\begin{align*}
&\left\langle \widehat\mW^\star\widehat\mW^{\star\T}, \mDelta\mDelta^\T\right\rangle \\
&= \sum_{i,j}\frac{a_ia_j}{2}\left\langle \begin{bmatrix}\mPhi\mSigma\mPhi^\T & -\mPhi\mSigma\mPsi^\T \\ -\mPsi\mSigma\mPhi^\T & \mPsi\mSigma\mPsi^\T \end{bmatrix}, \begin{bmatrix}\vphi_i\vphi_j^\T & \vphi_i\vpsi_j^\T\\ \vpsi_i\vphi_j^\T & \vpsi_i\vpsi_j^\T\end{bmatrix}\right\rangle = 0,
\end{align*}
and
\begin{align*}
&\left\langle \mW\mW^\T - \mW^\star\mW^{\star\T}, \mDelta\mDelta^\T\right\rangle = \left\langle \mW\mW^\T - \mW^\star\mW^{\star\T}, \vq\vq^\T\right\rangle\\
&= \vq^\top \mW\mW^\top \vq - 
\left\langle \mQ\overline\mSigma\mQ^\T, \sum_{ij}a_ia_j\vq_i\vq_j^\T\right\rangle\\
&\le \sigma_r^2(\mW) - \sum_{i=1}^r a_i^2 \overline\sigma_i \le \sigma_r^2(\mW) - \overline\sigma_r 
\end{align*}
where the first inequality utilizes \eqref{eq:near saddle lambda k} and the last inequality holds because $\overline \sigma_1 \ge \cdots \ge \overline\sigma_r$ and $\sum_i a_i^2 = 1$.

Plugging these into \eqref{eq:near saddle 2_1} gives
\begin{align}
 2\Pi_2 + \frac{1}{2}\Pi_5 \le \frac{1}{2}(\sigma_r^2(\mW) - \overline \sigma_r).
\label{eq:near saddle 2_2}\end{align}
\noindent{\bf Merging together:} Putting \eqref{eq:near saddle 1_2} and \eqref{eq:near saddle 2_2} together yields
\begin{align*}
[\nabla^2g(\mW)](\mDelta,\mDelta) &= \Pi_1 + \frac{1}{2} \Pi_3 +\frac{1}{2} \Pi_4 + 2\Pi_2 + \frac{1}{2} \Pi_5\\
& \leq \sigma_r^2(\mW) + \frac{1}{2}(\sigma_r^2(\mW) - \overline \sigma_r)\\
& \leq \frac{1}{2}\sigma_r(\mX^\star) +\frac{1}{2}(\frac{1}{2}\sigma_r(\mX^\star)-2\sigma_r(\mX^\star))\\
& \leq -\frac{1}{4}\sigma_r(\mX^\star),
\end{align*}
where the third line follows because by assumption $\sigma_r(\mW)\leq \sqrt{\frac{1}{2}}\sigma_r^{1/2}(\mX^\star)$ and $\overline \sigma_r = 2\sigma_r(\mX^\star)$. This completes the proof of \eqref{eq:thm negative curvature}.
}
\subsection{Large gradient for the region $\calR_3'\cup\calR''_3\cup\calR'''_3$:}

In order to show that $g(\mW)$ has a large gradient in the three regions $\calR_3'\cup\calR''_3\cup\calR'''_3$, we first provide a lower bound for the gradient. By \eqref{eq:local descent gradient g}, we have
\begin{equation}\begin{split}
&\|\nabla g(\mW)\|_F^2 \\
& = \frac{1}{4}\left\| \left(\mW\mW^\T - \mW^\star\mW^{\star\T}\right)\mW + \widehat\mW^\star\widehat\mW^{\star\T}\mW\right\|_F^2\\
& = \frac{1}{4}\left(\left\| \left(\mW\mW^\T - \mW^\star\mW^{\star\T}\right)\mW\right\|_F^2 + \left\| \widehat\mW^\star\widehat\mW^{\star\T}\mW\right\|_F^2 \right)\\
& \quad\quad\quad+ \frac{1}{2}\left\langle \left(\mW\mW^\T - \mW^\star\mW^{\star\T}\right)\mW, \widehat\mW^\star\widehat\mW^{\star\T}\mW \right\rangle\\
& = \frac{1}{4}\left(\left\| \left(\mW\mW^\T - \mW^\star\mW^{\star\T}\right)\mW\right\|_F^2 + \left\| \widehat\mW^\star\widehat\mW^{\star\T}\mW\right\|_F^2\right) \\& \quad\quad\quad+ \frac{1}{2}\left\langle \mW\mW^\T \mW\mW^\T, \widehat\mW^\star\widehat\mW^{\star\T} \right\rangle\\
& \geq \frac{1}{4}\left\| \left(\mW\mW^\T - \mW^\star\mW^{\star\T}\right)\mW\right\|_F^2,
\end{split}\label{eq:large gradient 1}
\end{equation}
where the third equality follows because $\mW^{\star\T}\widehat\mW^\star = \mU^{\star\T}\mU^{\star} - \mV^{\star\T}\mV^{\star} = \mzero$ and the last line utilizes the fact that the inner product between two PSD matrices is nonnegative.

\subsubsection{Large gradient for the region $\calR_3'$}
\revise{For any $\mW$, we can always divide it into two parts, the projections onto the column spaces of $\mQ$ (which is defined in \eqref{eq:define Q}) and its orthogonal complement, respectively. 
Equivalently, we can write
\begin{align}
\mW = \mQ\mUpsilon + \mE,
\label{eq:rewrite W} \end{align}
where $\mUpsilon\in\R^{r\times r}$ is  the projection of $\mW$ onto the column space of $\mQ$, and $\mE^\T \mQ = \mzero$ (i.e., $\mE$ is orthogonal to $\mQ$).
Plugging this form of $\mW$ into the last term of \eqref{eq:large gradient 1} gives
\begin{equation}\begin{split}
& \left\| \left(\mW\mW^\T - \mW^\star\mW^{\star\T}\right)\mW\right\|_F^2=\\
 &  \left\| \mQ (\mUpsilon\mUpsilon^\T - \overline\mSigma)\mUpsilon + \mQ \mUpsilon\mE^\T\mE + \mE\mUpsilon^\T\mUpsilon + \mE\mE^\T\mE\right\|_F^2\\
  & =  \left\|(\mUpsilon\mUpsilon^\T - \overline\mSigma)\mUpsilon + \mUpsilon\mE^\T\mE\right\|_F^2 + \left\|  \mE\mUpsilon^\T\mUpsilon + \mE\mE^\T\mE\right\|_F^2
\end{split}\label{eq:large gradient 2}
\end{equation}
since $\mQ$ is orthogonal to $\mE$. The remaining part is to show at least one of the two terms is large for any $\mW\in\calR_3'$ by considering the following two cases.

Case I: $\left\|\mE\right\|_F^2\geq \frac{4}{25}\sigma_r(\mX^\star)$. As $\mE$ is large, we bound the second term in \eqref{eq:large gradient 2}:
\begin{equation}\begin{split}
\left\|  \mE\mUpsilon^\T\mUpsilon + \mE\mE^\T\mE\right\|_F^2 & \geq \sigma_r^2\left(\mUpsilon^\T\mUpsilon + \mE^\T\mE\right)\left\|  \mE\right\|_F^2\\& = \sigma_r^4\left(\mW\right)\left\|  \mE\right\|_F^2\\
&\geq (\frac{1}{2})^2\frac{4}{25}\sigma_r^3(\mX^\star) = \frac{1}{25} \sigma_r^3(\mX^\star),
\end{split}\label{eq:large gradient 3_1}
\end{equation}
where the first inequality follows from Corollary~\ref{cor:inequality for Frobunius of products}, the first equality follows from the fact $\mW^\T\mW = \mUpsilon^\T\mUpsilon + \mE^\T\mE$, and the last inequality holds because by assumption that $\sigma_r^2(\mW)\geq \frac{1}{2}\sigma_r(\mX^\star)$ and $\left\|\mE\right\|_F^2\geq \frac{4}{25}\sigma_r(\mX^\star)$.

Case II: $\left\|\mE\right\|_F^2\leq \frac{4}{25}\sigma_r(\mX^\star)$.
In this case, we start by bounding the smallest singular value of $\mUpsilon$. First, utilizing Weyl's inequality for perturbation of singular values~\cite[Theorem 3.3.16]{horn2012matrix} gives
\[
\left|\sigma_r(\mW) - \sigma_r(\mUpsilon)\right|\leq \|\mE\|_2,
\]
which implies
\begin{align}\begin{split}
\sigma_r(\mUpsilon) &\geq  \sigma_r(\mW) - \|\mE\|_2\\
&\geq \sqrt{\frac{1}{2}}\sigma_r^{1/2}(\mX^\star) - \frac{2}{5}\sigma_r^{1/2}(\mX^\star),
\end{split}\label{eq:large gradient bound min lambda}\end{align}
where we utilize $\|\mE\|_2 \leq \|\mE\|_F\leq \frac{2}{5}\sigma_r^{1/2}(\mX^\star)$. Thus, $\mUpsilon$ is invertible. On the other hand,
\begin{align*}
\dist(\mW,\mW^\star) &\leq \dist(\mQ\mUpsilon,\mW^\star)  + \left\| \mE\right\|_F
\\
&=  \dist(\mUpsilon, \overline\mSigma^{1/2})  + \left\| \mE\right\|_F,
\end{align*}
where the last line follows because $\mQ$ contain the left singular vectors of $\mW^\star$.
This together with the assumption that $\dist(\mW,\mW^\star)\geq \sigma_r^{1/2}(\mX^\star)$ gives
\begin{align*}
\dist(\mUpsilon, \overline\mSigma^{1/2}) \geq  \sigma_r^{1/2}(\mX^\star) - \frac{2}{5}\sigma_r^{1/2}(\mX^\star) = \frac{3}{5}\sigma_r^{1/2}(\mX^\star).
\end{align*}

We now bound the first term in \eqref{eq:large gradient 2}:
\begin{equation}\begin{split}
&\left\|(\mUpsilon\mUpsilon^\T - \overline\mSigma)\mUpsilon + \mUpsilon\mE^\T\mE\right\|_F\\ &
= \left\|\mUpsilon(\mUpsilon^\T \mUpsilon - \mUpsilon^{-1}\overline\mSigma\mUpsilon) + \mUpsilon\mE^\T\mE\right\|_F\\
& \ge \sigma_r(\mUpsilon) \left\|(\mUpsilon^\T \mUpsilon - \mUpsilon^{-1}\overline\mSigma\mUpsilon) + \mE^\T\mE\right\|_F\\
& \ge \sigma_r(\mUpsilon)\left( \left\|\mUpsilon^\T \mUpsilon - \mUpsilon^{-1}\overline\mSigma\mUpsilon\right\|_F - \left\|\mE^\T\mE\right\|_F\right)\\
& \ge \sigma_r(\mUpsilon) \left( \left\|\mUpsilon \mUpsilon^\T - \overline\mSigma\right\|_F - \left\|\mE^\T\mE\right\|_F\right)\\
& \ge \sigma_r(\mUpsilon) \left(\sqrt{2(\sqrt{2}-1)}\overline\sigma_r^{1/2} \dist(\mUpsilon, \overline\mSigma^{1/2}) - \left\|\mE^\T\mE\right\|_F\right)\\
& \ge \big(\sqrt{\frac{1}{2}} - \frac{2}{5}\big) \Big(\sqrt{2(\sqrt{2}-1)}\cdot\sqrt{2}\cdot\frac{3}{5} - \frac{4}{25}\Big)\sigma_r^{3/2}(\mX^\star)
\end{split}\label{eq:large gradient 3_2}
\end{equation}
where the third line uses \Cref{cor:inequality for Frobunius of products}, the fifth line will be proved soon, the sixth line utilizes \cite[Lemma 5.4]{tu2015low} that $\|\mA\mA^\T - \mB\mB^\T\|_F^2 \ge 2(\sqrt{2} -1)\sigma_r^2(\mB)\dist^2(\mA,\mB)$ for any $\mA,\mB\in\R^{n\times r}$, and last line holds because $\sigma_r(\mUpsilon) \geq \big(\sqrt{\frac{1}{2}} - \frac{2}{5}\big)\sigma_r^{1/2}(\mX^\star)$ by \eqref{eq:large gradient bound min lambda}, $\overline \sigma_r = 2\sigma_r(\mX^\star)$, $\dist(\mUpsilon, \overline\mSigma^{1/2}) \geq  \frac{3}{5}\sigma_r^{1/2}(\mX^\star)$, and
$\|\mE\mE^\T\|_F\leq \|\mE\|_F^2\leq \frac{4}{25}\sigma_r(\mX^\star)$. We now prove the fifth line in \eqref{eq:large gradient 3_2} as
\begin{align*}
	&\left\|\mUpsilon^\T \mUpsilon - \mUpsilon^{-1}\overline\mSigma\mUpsilon\right\|_F^2 - \left\|\mUpsilon \mUpsilon^\T - \overline\mSigma\right\|_F^2\\
	&=\left\|\mUpsilon^{-1}\overline\mSigma\mUpsilon\right\|_F^2 - \left\| \overline\mSigma\right\|_F^2\ge 0,
\end{align*}
where the last inequality holds because $(i)$ $\mUpsilon^{-1}\overline\mSigma\mUpsilon$ is similar to $\overline\mSigma$ and thus the diagonals of $\overline\mSigma$ are the eigenvalues of $\mUpsilon^{-1}\overline\mSigma\mUpsilon$, and $(ii)$ for any square matrix, the $\ell_2$ norm of its eigenvalues is no larger than its Frobenius norm \cite[Fact 9.11.3]{bernstein2009matrix}
.

Combining \eqref{eq:large gradient 1} with \eqref{eq:large gradient 2}, \eqref{eq:large gradient 3_1} and \eqref{eq:large gradient 3_2} gives
\begin{align*}
\|\nabla g(\mW)\|_F \geq \frac{1}{11}\sigma_r^{3/2}(\mX^\star).
\end{align*}
This completes the proof of \eqref{eq:thm large gradient 1}.}
\subsubsection{Large gradient for the region $\calR_3''$:}
By \eqref{eq:large gradient 1}, we have
\[
\|\nabla g(\mW)\|_F  \geq \frac{1}{2}\left\| \left(\mW\mW^\T - \mW^\star\mW^{\star\T}\right)\mW\right\|_F^2.
\]
Now~\eqref{eq:thm large gradient 2} follows directly from the fact $\|\mW\|> \frac{20}{19} \|\mW^\star\|$ and the following result.
\begin{lem}
For any $\mA,\mB\in\R^{n\times r}$ with $\|\mA\|\geq \alpha\|\mB\|$ and $\alpha>1$, we have
\[
\left\|\left(\mA\mA^\T - \mB\mB^\T\right)\mA\right\|_F \geq (1-\frac{1}{\alpha^2})\|\mA\|^3.\]
\end{lem}
\begin{proof}
Let $\mA = \mPhi_1\mLambda_1\mR_1^\T$ and $\mB = \mPhi_2\mLambda_2\mR_2^\T$ be the SVDs of $\mA$ and $\mB$, respectively. Then
\begin{align*}
\left\|\left(\mA\mA^\T - \mB\mB^\T\right)\mA\right\|_F& = \left\|\mPhi_1\mLambda_1^3 - \mPhi_2\mLambda_2^2\mPhi_2^\T\mPhi_1\mLambda_1\right\|_F\\
& \geq \left\|\mLambda_1^3 - \mPhi_1^\T\mPhi_2\mLambda_2^2\mPhi_2^\T\mPhi_1\mLambda_1\right\|_F\\
& \geq \left\|\mLambda_1^3 - \mLambda_2^2\mLambda_1\right\|_F\\
& \geq (1-\frac{1}{\alpha^2})\|\mA\|^3.
\end{align*}

\end{proof}

\subsubsection{Large gradient for the region $\calR_3'''$:}

By \eqref{eq:local descent gradient g}, we have
\begin{equation}\label{eq:large gradient III eq 1}\begin{split}
&\left\langle \nabla g(\mW),\mW\right\rangle \\&= \left\langle \frac{1}{2} \left(\mW\mW^\T - \mW^\star\mW^{\star\T}\right)\mW + \frac{1}{2} \widehat\mW^\star\widehat\mW^{\star\T}\mW,\mW\right\rangle\\
& \geq \frac{1}{2}\left\langle \left(\mW\mW^\T - \mW^\star\mW^{\star\T}\right)\mW,\mW\right\rangle\\
& \geq \frac{1}{2}\left(\left\| \mW\mW^\T\right\|_F^2 - \left\|  \mW\mW^\T\right\|_F\left\|  \mW^\star\mW^{\star\T}\right\|_F\right)\\
& > \frac{1}{20} \left\| \mW\mW^\T\right\|_F^2
\end{split}\end{equation}
where the last line holds because $\|\mW^\star\mW^{\star\T}\|_F<\frac{9}{10}\|\mW\mW^\T\|_F$.

\section{Proof of \Cref{thm:robust strict saddle property general} (robust strict saddle for $G(\mW)$)}\label{sec:prf robust strict saddle property sensing}
Throughout the proofs, we always utilize $\mX = \mU\mV^\T$ unless stated otherwise. To give a sense that the geometric result in Theorem~\ref{thm:robust strict saddle property} for $g(\mW)$ is also possibly preserved for $G(\mW)$, we first compute the derivative of $G(\mW)$ as
\begin{align}
\nabla G(\mW) = \begin{bmatrix} \nabla f(\mU\mV^\T)\mV \\ (\nabla f(\mU\mV^\T))^\T\mU \end{bmatrix} +  \mu \widehat\mW\widehat\mW^\T\mW.
\label{eq:gradient general prob}\end{align}
For any $\mDelta= \begin{bmatrix} \mDelta_{\mU}\\ \mDelta_{\mV}\end{bmatrix}\in\R^{(n+m)\times r}$, algebraic calculation gives the Hessian quadratic form $[\nabla^2 G(\mW)](\mDelta,\mDelta)$ as
\begin{equation}\begin{split}
&[\nabla^2G(\mW)](\mDelta,\mDelta) \\
&= [\nabla^2f(\mU\mV^\T)](\mDelta_{\mU}\mV^\T+ \mU\mDelta_{\mV}^\T,\mDelta_{\mU}\mV^\T + \mU\mDelta_{\mV}^\T)\\
& \quad + 2\langle \nabla f(\mU\mV^\T),\mDelta_{\mU}\mDelta_{\mV}^\T \rangle+ [\nabla^2\rho(\mW)](\mDelta,\mDelta)
\end{split}\label{eq:hessian general prob}\end{equation}
where $
[\nabla^2\rho(\mW)](\mDelta,\mDelta)$ is defined in \eqref{eq:hessian for regularizer}.
Thus, it is expected that $G(\mW)$, $\nabla G(\mW)$, and $\nabla^2 G(\mW)$ are close to their counterparts (i.e., $g(\mW)$, $\nabla g(\mW)$ and $\nabla^2 g(\mW)$) for the matrix factorization problem when $f(\mX)$ satisfies the $(2r,4r)$-restricted  strong convexity and smoothness condition \eqref{eq:RIP like}.

Before moving to the main proofs, we provide several useful results regarding the deviations of the gradient and Hessian. We start with a useful characterization of the restricted  strong convexity and smoothness condition.
\begin{lem}Suppose $f$ satisfies the $(2r,4r)$-restricted  strong convexity and smoothness condition \eqref{eq:RIP like} with positive constants $a = 1-c$ and $b = 1+c, c\in[0,1)$. Then  any $n\times m$ matrices $\mC,\mD,\mH$ with $\rank(\mC),\rank(\mD)\leq r$ and $\rank(\mH)\leq 2r$, we have
\[
\left|\left\langle \nabla f\left(\mC\right) - \nabla f\left(\mD\right) - (\mC- \mD), \mH \right\rangle\right| \leq c\left\|\mC - \mD\right\|_F\left\|\mH\right\|_F.
\]
\label{lem:RIP reformulation}\end{lem}
\begin{proof}[Proof of \Cref{lem:RIP reformulation}]
We first invoke \cite[Proposition 2]{zhu2017GlobalOptimality} which states that under \Cref{assump:2}
for any $n\times m$ matrices $\mZ,\mD,\mH$ of rank at most $2r$, we have
\begin{align}
&\left|[\nabla^2f(\mZ)](\mD,\mH) - \langle \mD,\mH \rangle\right|\leq c\left\|\mD\right\|_F \left\|\mH\right\|_F.
\label{eq:nabla 2 G H}\end{align}
Now using integral form of the mean value theorem for $\nabla f$, we have
\begin{align*}
&\left|\left\langle \nabla f\left(\mC\right) - \nabla f\left(\mD\right) - (\mC- \mD), \mH \right\rangle\right|\\
&=\left|\int_0^1\left[ \nabla^2 f(t\mC + (1-t)\mD)\right](\mC -\mD, \mH) - \langle \mC - \mD,\mH \rangle dt \right|\\
&\leq \int_0^1\left|\left[ \nabla^2 f(t\mC + (1-t)\mD)\right](\mC -\mD, \mH) - \langle \mC - \mD,\mH \rangle \right|dt\\
& \leq \int_0^1 c\left\|\mC - \mD\right\|_F\left\|\mH\right\|_F dt = c\left\|\mC - \mD\right\|_F\left\|\mH\right\|_F.
\end{align*}
where the second inequality follows from \eqref{eq:nabla 2 G H} since $t\mC + (1-t)\mD$, $\mC - \mD$, and $\mH$ all are rank at most $2r$.

\end{proof}

The following result controls the deviation of the gradient between the general low-rank optimization \eqref{eq:general low rank} and the matrix factorization problem by utilizing the $(2r,4r)$-restricted  strong convexity and smoothness condition \eqref{eq:RIP like}.
\begin{lem}\label{lem:gradient g - G norm} Suppose $f(\mX)$ has a critical point $\mX^\star\in\R^{n\times m}$ of rank $r$ and satisfies the $(2r,4r)$-restricted  strong convexity and smoothness condition \eqref{eq:RIP like} with positive constants $a = 1-c$ and $b = 1+c, c\in[0,1)$. Then, we have
\begin{align*}
\left\|\nabla G(\mW) -  \nabla g(\mW)\right\|_F \leq c\left\|\mW\mW^\T - \mW^\star\mW^{\star\T}\right\|_F \left\|\mW\right\|.
\end{align*}
\end{lem}
\begin{proof}[Proof of Lemma~\ref{lem:gradient g - G norm}]
We bound the deviation directly:
\begin{align*}
&\left\|\nabla G(\mW) -  \nabla g(\mW)\right\|_F = \max_{\|\mDelta\|_F=1} \left\langle\nabla G(\mW) -  \nabla g(\mW),\mDelta\right\rangle\\
& = \max_{\|\mDelta\|_F=1}\left\langle \nabla f(\mX) ,\mDelta_{\mU}\mV^\T\right\rangle  - \left\langle \mX - \mX^\star,\mDelta_{\mU}\mV^\T \right\rangle \\
& \quad + \left\langle \nabla f(\mX),\mU\mDelta_{\mV}^\T \right\rangle- \left\langle \mX - \mX^\star,\mU\mDelta_{\mV}^\T \right\rangle\\
& = \max_{\|\mDelta\|_F=1}\left\langle \nabla f(\mX) - \nabla f(\mX^\star) - (\mX - \mX^\star) ,\mDelta_{\mU}\mV^\T\right\rangle \\
& \quad + \left\langle \nabla f(\mX) - \nabla f(\mX^\star) - (\mX - \mX^\star),\mU\mDelta_{\mV}^\T \right\rangle\\
& \leq \max_{\|\mDelta\|_F=1}c\left\|\mX - \mX^\star\right\|_F\left( \left\|\mDelta_{\mU}\mV^\T\right\|_F +\left\|\mU\mDelta_{\mV}^\T\right\|_F\right)\\
& \leq c\|\mU\mV^\T - \mX^\star\|_F\left( \left\|\mV\right\| +\left\|\mU\right\|\right)\\
& \leq c\|\mW\mW^\T - \mW^\star\mW^{\star\T}\|_F \left\|\mW\right\|,
\end{align*}
where the last equality follows from \Cref{assump:1} that $\nabla f(\mX^\star) = \mzero$ and
and the first inequality utilizes Lemma~\ref{lem:RIP reformulation}.
\end{proof}
Similarly, the next result controls the deviation of the Hessian between the matrix sensing problem and the matrix factorization problem.
\begin{lem}\label{lem:Hessian g - G norm}  Suppose $f(\mX)$ has a critical point $\mX^\star\in\R^{n\times m}$ of rank $r$ and satisfies the $(2r,4r)$-restricted  strong convexity and smoothness condition \eqref{eq:RIP like} with positive constants $a = 1-c$ and $b = 1+c, c\in[0,1)$. Then, for any $\mDelta= \begin{bmatrix} \mDelta_{\mU}\\ \mDelta_{\mV}\end{bmatrix}\in\R^{(n+m)\times r}$ the following holds:
\begin{align*}
&\left|\nabla^2 G(\mW)[\mDelta,\mDelta] -  \nabla^2 g(\mW)[\mDelta,\mDelta]\right|\\ & \leq 2 c\left\|\mU\mV^\T - \mX^\star\right\|_F\left\|\mDelta_{\mU}\mDelta_{\mV}^\T\right\|_F + c\left\|\mDelta_{\mU}\mV^\T+ \mU\mDelta_{\mV}^\T\right\|_F^2.
\end{align*}
\end{lem}
\begin{proof}[Proof of Lemma~\ref{lem:Hessian g - G norm}]
First note that
\begin{align*}
&\nabla^2 G(\mW)[\mDelta,\mDelta] -  \nabla^2 g(\mW)[\mDelta,\mDelta]\\
&= 2\left\langle \nabla f(\mX),\mDelta_{\mU}\mDelta_{\mV}^\T \right\rangle- 2\left\langle \mX - \mX^\star,\mDelta_{\mU}\mDelta_{\mV}^\T \right\rangle \\&\quad+  [\nabla^2f(\mX)](\mDelta_{\mU}\mV^\T+ \mU\mDelta_{\mV}^\T) - \left\|\mDelta_{\mU}\mV^\T+ \mU\mDelta_{\mV}^\T\right\|_F^2.
\end{align*}
Now utilizing Lemma~\ref{lem:RIP reformulation} and \eqref{eq:RIP like}, we have
\begin{align*}
&\left|\nabla^2 G(\mW)[\mDelta,\mDelta] -  \nabla^2 g(\mW)[\mDelta,\mDelta]\right|\\
& \leq  2\left|\left\langle \nabla f(\mX) - \nabla f(\mX^\star),\mDelta_{\mU}\mDelta_{\mV}^\T \right\rangle- \langle \mX - \mX^\star,\mDelta_{\mU}\mDelta_{\mV}^\T \rangle\right| \\
&\quad+  \left| [\nabla^2f(\mX)](\mDelta_{\mU}\mV^\T+ \mU\mDelta_{\mV}^\T) - \left\|\mDelta_{\mU}\mV^\T+ \mU\mDelta_{\mV}^\T\right\|_F^2 \right|\\
& \leq 2 c\left\|\mU\mV^\T - \mX^\star\right\|_F\left\|\mDelta_{\mU}\mDelta_{\mV}^\T\right\|_F \\&\quad+ c\left\|\mDelta_{\mU}\mV^\T+ \mU\mDelta_{\mV}^\T\right\|_F^2.
\end{align*}
\end{proof}

We provide one more result before proceeding to prove the main theorem.
\begin{lem}\label{lem:bound (C-D)C^T}\cite[Lemma E.1]{bhojanapalli2016lowrankrecoveryl}
 Let $\mA$ and $\mB$ be two $n\times r$ matrices such that $\mA^\T\mB = \mB^\T\mA$ is PSD.
 Then
\[
\left\|\left(\mA -\mB \right)\mA^\T\right\|_F^2\leq
\frac{1}{2(\sqrt{2} -1)}\left\|\mA\mA^\T - \mB\mB^\T   \right\|_F^2.
\]
\end{lem}
\subsection{Local descent condition for the region $\calR_1$}
Similar to what used in Appendix~\ref{prf:locol sescent}, we perform the change of variable $\mW^\star\mR \rightarrow \mW^\star$ to avoid $\mR$ in the following equations. With this change of variable we have instead $\mW^\T\mW^\star = \mW^{\star\T}\mW$ is PSD.

We first control $\left|\left\langle\nabla G(\mW) -  \nabla g(\mW),\mW - \mW^\star\right\rangle\right|$ as follows:
\begin{align*}
&\left|\left\langle\nabla G(\mW) -  \nabla g(\mW),\mW - \mW^\star\right\rangle\right|\\
&\leq \left|\langle\nabla f(\mX),(\mU - \mU^\star)\mV^\T \rangle - \langle \mX - \mX^\star,(\mU - \mU^\star)\mV^\T\rangle\right|\\
 & \quad +\left|\langle\nabla f(\mX),\mU(\mV - \mV^\star)^\T \rangle - \langle \mX - \mX^\star,\mU(\mV - \mV^\star)^\T\rangle\right|\\
 & \leq c\left\|\mX - \mX^\star\right\|_F\left(\| (\mU - \mU^\star)\mV^\T\|_F +  \|\mU(\mV - \mV^\star)^\T\|_F \right) \\
& \leq c\|\mW\mW^\T - \mW^\star\mW^{\star\T}\|_F \left\|\mW(\mW-\mW^\star)^\T\right\|_F\\
& \leq  \frac{c}{2(\sqrt{2} -1)}\|\mW\mW^\T - \mW^\star\mW^{\star\T}\|_F^2
\end{align*}
where the second inequality utilizes $\nabla f(\mX^\star) = \vzero$ and Lemma~\ref{lem:RIP reformulation}, and the last inequality follows from Lemma~\ref{lem:bound (C-D)C^T}. The above result along with \eqref{eq:proof descent main}-\eqref{eq:proof descent main second term 2} gives
\begin{equation}
\begin{split}
&\left\langle \nabla G(\mW), \mW - \mW^\star\right\rangle\\
& \geq \left\langle \nabla g(\mW), \mW - \mW^\star\right\rangle - \left|\left\langle\nabla G(\mW) -  \nabla g(\mW),\mW - \mW^\star\right\rangle\right|\\
& \geq  \left\langle\nabla g(\mW), \mW - \mW^\star\right\rangle - \frac{c}{2(\sqrt{2} -1)}\left\|\mW\mW^\T - \mW^\star\mW^{\star\T}\right\|_F^2\\
& \geq \frac{1}{16}\sigma_r(\mX^\star)\dist^2(\mW,\mW^\star)  + \frac{1}{32}\left\|\mW\mW^\T - \mW^\star\mW^{\star\T}\right\|_F^2\\
 &\quad+ \frac{1}{4\|\mX^\star\|} \left\|\widehat\mW^\star\widehat\mW^{\star\T}\mW\right\|_F^2\\
 &\quad- \frac{c}{2(\sqrt{2} -1)}\left\|\mW\mW^\T - \mW^\star\mW^{\star\T}\right\|_F^2\\
& \geq \frac{1}{16}\sigma_r(\mX^\star)\dist^2(\mW,\mW^\star)  + \frac{1}{160}\left\|\mW\mW^\T - \mW^\star\mW^{\star\T}\right\|_F^2 \\&\quad+ \frac{1}{4\|\mX^\star\|} \left\|\widehat\mW^\star\widehat\mW^{\star\T}\mW\right\|_F^2
\end{split}
\label{eq:proof descent sensing main}\end{equation}
where we utilize $c\leq \frac{1}{50}$.

On the other hand, we control $\|\nabla G(\mW)\|_F$ with Lemma~\ref{lem:gradient g - G norm} controlling the deviation between $\nabla G(\mW)$ and $\nabla g(\mW)$ as follows:
\begin{equation}
\begin{split}
&\left\|\nabla G(\mW)\right\|_F^2 =  \left\|\nabla g(\mW) + \nabla G(\mW) - \nabla g(\mW)\right\|_F^2\\
&\leq \frac{20}{19}\left\|\nabla g(\mW)\right\|_F^2 + 20\left\|\nabla g(\mW) - \nabla G(\mW)\right\|_F^2\\
&\leq \frac{20}{19}\left\|\nabla g(\mW)\right\|_F^2 + 20c^2\|\mW\|^2\left\|\mW\mW^\T - \mW^\star\mW^{\star\T} \right\|_F^2\\
& =\frac{5}{19}\left\|   \left(\mW\mW^\T - \mW^\star\mW^{\star\T}\right)\mW + \widehat\mW^\star\widehat\mW^{\star\T}\mW \right\|_F^2\\
&\quad + 20c^2\|\mW\|^2\left\|\mW\mW^\T - \mW^\star\mW^{\star\T} \right\|_F^2\\
& \leq \left(\frac{5}{19}\frac{100}{99} + 20c^2\right) \left\|   \left(\mW\mW^\T - \mW^\star\mW^{\star\T}\right)\mW \right\|_F^2 \\
&\quad+ 25 \left\| \widehat\mW^\star\widehat\mW^{\star\T}\mW \right\|_F^2\\
& \leq(\frac{5}{19}\frac{100}{99} + 50c^2)(\sqrt{2}+1)^2\|\mX^\star\|\|\mW\mW^\T - \mW^\star\mW^{\star\T} \|_F^2 \\
&\quad+ 25   \| \widehat\mW^\star\widehat\mW^{\star\T}\mW \|_F^2,
\end{split}
\label{eq:proof descent sensing main 2}\end{equation}
where the first inequality  holds since $(a+b)^2\leq \frac{1+\epsilon}{\epsilon}a^2 + (1+\epsilon)b^2$ for any $\epsilon>0$, and
the fourth line follows from~\eqref{eq:local descent gradient g}.

Now combining \eqref{eq:proof descent sensing main}-\eqref{eq:proof descent sensing main 2} and assuming $c\leq \frac{1}{50}$ gives
\begin{align*}
&\left\langle \nabla G(\mW), \mW - \mW^\star\right\rangle\\
& \geq \frac{1}{16}\sigma_r(\mX^\star)\dist^2(\mW,\mW^\star) + \frac{1}{260\|\mX^\star\|}\|\nabla G(\mW)\|_F^2.
\end{align*}
This completes the proof of \eqref{eq:thm sensing regularity condition}.
\subsection{Negative curvature for the region $\calR_2$}
Let $\mDelta = \vq_k\valpha^\T$ be defined as in \eqref{eq:Delta for negative curvature}. First note that
\begin{align*}
\left\|\mDelta_{\mU}\mV^\T+ \mU\mDelta_{\mV}^\T\right\|_F^2 &\leq 2 \left\|\mDelta_{\mU}\mV^\T\right\|_F^2 + 2\left\| \mU\mDelta_{\mV}^\T\right\|_F^2\\
& \leq 2\left\|\mW\mDelta^\T\right\|_F^2 = 2\sigma_r^2(\mW)\leq \sigma_r(\mX^\star),
\end{align*}
where the last equality holds because $\sigma_r(\mW)\leq \sqrt{\frac{1}{2}}\sigma_r^{1/2}(\mX^\star)$.  Also utilizing the particular structure in $\mDelta$ yields
\[
\left\|\mDelta_{\mU}\mDelta_{\mV}^\T\right\|_F = \frac{1}{2}\left\|\vphi_k\vpsi_k^\T\right\|_F=\frac{1}{2}.
\]
Due to the assumption $\frac{20}{19}\|\mW^\star\mW^{\star\T}\|_F\geq \|\mW\mW^{\T}\|_F$, we have
\begin{align*}
&\|\mU\mV^\T - \mX^\star\|_F \leq \frac{\sqrt{2}}{2}\|\mW\mW^\T - \mW^\star\mW^{\star\T} \|_F\\
&\leq \frac{\sqrt{2}}{2} (\frac{20}{19}\|\mW^\star\mW^{\star\T} \|_F + \|\mW^\star\mW^{\star\T} \|_F) = \frac{39\sqrt{2}}{19}\|\mX^\star\|_F.
\end{align*}
Now combining the above results with Lemma~\ref{lem:Hessian g - G norm}, we have
\begin{align*}
&\nabla^2 G(\mW)[\mDelta,\mDelta]\\ &\leq \nabla^2 g(\mW)[\mDelta,\mDelta] +  \left|\nabla^2 G(\mW)[\mDelta,\mDelta] -  \nabla^2 g(\mW)[\mDelta,\mDelta]\right|\\
& \leq -\frac{1}{4}\sigma_r(\mX^\star) + 2 c\left\|\mU\mV^\T - \mX^\star\right\|_F\left\|\mDelta_{\mU}\mDelta_{\mV}^\T\right\|_F \\&\quad+ c\left\|\mDelta_{\mU}\mV^\T+ \mU\mDelta_{\mV}^\T\right\|_F^2\\
& \leq -\frac{1}{4}\sigma_r(\mX^\star) + \frac{39}{19}\sqrt{2}c\|\mX^\star\|_F + c\sigma_r(\mX^\star)\\
& \leq -\frac{1}{6}\sigma_r(\mX^\star),
\end{align*}
where the last line holds when $c\leq \frac{\sigma_r(\mX^\star)}{50\|\mX^\star\|_F}$. This completes the proof of \eqref{eq:thm sensing negative curvature}.

\subsection{Large gradient for the region $\calR_3'\cup\calR''_3\cup\calR'''_3$:}
To show that $G(\mW)$ has large gradient in these three regions, we mainly utilize Lemma~\ref{lem:gradient g - G norm} to guarantee that $\nabla G(\mW)$ is close to $\nabla g(\mW)$.
\subsubsection{Large gradient for the region $\calR_3'$}
Utilizing Lemma~\ref{lem:gradient g - G norm}, we have
\begin{align*}
&\left\|\nabla G(\mW) \right\|_F\\
&\geq \left\| \nabla g(\mW)\right\|_F - \left\|\nabla G(\mW) -  \nabla g(\mW)\right\|_F\\
&\geq \left\| \nabla g(\mW)\right\|_F -  c\left\|\mW\mW^\T - \mW^\star\mW^{\star\T}\right\|_F \left\|\mW\right\|\\
& \geq \left\| \nabla g(\mW)\right\|_F -  c(\frac{10}{9}\|\mW^\star\mW^{\star\T}\|_F + \| \mW^\star\mW^{\star\T}\|_F) \left\|\mW\right\|\\
& \geq \revise{\frac{1}{11}}\sigma_r^{3/2}(\mX^\star) - c\frac{19}{9}2\|\mX^\star\|_F\frac{20}{19}\sqrt{2}\|\mX^\star\|^{1/2}\\
& \geq \revise{\frac{1}{50}}\sigma_r^{3/2}(\mX^\star),
\end{align*}
where the fourth line follows because $\left\|\mW^\star\mW^{\star\T}\right\|_F = 2\|\mX^\star\|_F$ and $\left\|\mW\right\|\leq \frac{20}{19}\sqrt{2}\|\mX^\star\|^{1/2}$, and the last line holds if $c\leq \frac{1}{100}\frac{\sigma_r^{3/2}(\mX^\star)}{\|\mX^\star\|_F\|\mX^\star\|^{1/2}}$. This completes the proof of \eqref{eq:thm sensing large gradient 1}.

\subsubsection{Large gradient for the region $\calR_3''$}

Utilizing Lemma~\ref{lem:gradient g - G norm} again, we have
\begin{align*}
&\left\|\nabla G(\mW) \right\|_F\\
& \geq \left\| \nabla g(\mW)\right\|_F -  c\left(\left\|\mW\mW^\T\right\|_F + \left\| \mW^\star\mW^{\star\T}\right\|_F\right) \left\|\mW\right\|\\
& \geq \frac{39}{800}\|\mW\|^{3} - c\left(\frac{10}{9}\left\|\mW^\star\mW^{\star\T}\right\|_F + \left\| \mW^\star\mW^{\star\T}\right\|_F\right) \left\|\mW\right\|\\
& \geq \frac{39}{800}\|\mW\|^{3} - c\frac{19}{9}2\left\| \mX^\star\right\|_F \left\|\mW\right\|\\
& \geq  \frac{39}{800}\|\mW\|^{3} - \frac{19}{450}\left\| \mX^\star\right\| \left\|\mW\right\|\\
& \geq \frac{1}{50}\|\mW\|^3,
\end{align*}
where the fourth line holds if $c\leq \frac{1}{100}\frac{\sigma_r^{3/2}(\mX^\star)}{\|\mX^\star\|_F\|\mX^\star\|^{1/2}}$ and the last follows from the fact that
\[
\|\mW\|> \frac{20}{19}\|\mW^\star\| \geq \frac{20}{19}\sqrt{2}\|\mX^\star\|^{1/2}.
\]
This completes the proof of \eqref{eq:thm sensing large gradient 2}.
\subsubsection{Large gradient for the region $\calR_3'''$}
To show \eqref{eq:thm sensing large gradient 3},
we first control $\left|\left\langle\nabla G(\mW) -  \nabla g(\mW),\mW\right\rangle\right|$ as follows:
\begin{align*}
&\left|\left\langle\nabla G(\mW) -  \nabla g(\mW),\mW\right\rangle\right|\\
& = 2\left|\left\langle \nabla f(\mU\mV^\T ),\mU\mV^\T \right\rangle- \left\langle \mU\mV^\T - \mX^\star,{\mU}\mV^\T \right\rangle \right|\\
& \leq 2c\left\|\mU\mV^\T - \mX^\star\right\|_F \left\|{\mU}\mV^\T\right\|_F\\
& \leq 2c \frac{19}{20}\sqrt{2} \|\mW\mW^\T\|_F \frac{1}{2}\|\mW\mW^\T\|_F = \frac{19}{20}\sqrt{2}c \|\mW\mW^\T\|_F^2,
\end{align*}
where the first inequality utilizes the fact $\nabla f(\mX^\star) = \vzero$ and Lemma~\ref{lem:RIP reformulation}, and
the last inequality holds because
\begin{align*}
\left\|\mU\mV^\T - \mX^\star\right\|_F &\leq \frac{\sqrt{2}}{2}\left\|\mW\mW^\T - \mW^\star\mW^{\star\T} \right\|_F\\
&\leq \frac{\sqrt{2}}{2} \left(\frac{9}{10}\left\|\mW\mW^{\T} \right\|_F + \left\|\mW\mW^{\T} \right\|_F\right)\\
 &= \frac{19\sqrt{2}}{20}\left\|\mW\mW^{\T} \right\|_F
\end{align*}
 and
\begin{align*}
\|\mW\mW^\T\|_F^2& = \|\mU\mU^\T\|_F^2 \!+\! \|\mV\mV^\T\|_F^2 \!+\! 2\|\mU\mV^\T\|_F^2 \geq 4 \|\mU\mV^\T\|_F^2
\end{align*}
by noting that
\[
\|\mU\mU^\T\|_F^2 + \|\mV\mV^\T\|_F^2 -  2\|\mU\mV^\T\|_F^2 = \|\mU^\T\mU - \mV^\T\mV\|_F^2 \geq 0.
\]

Now utilizing \eqref{eq:large gradient III eq 1} to provide a lower bound for $\left\langle \nabla g(\mW),\mW\right\rangle$, we have
\begin{align*}
&\left|\left\langle\nabla G(\mW),\mW\right\rangle\right|\\ &\geq \left\langle \nabla g(\mW),\mW\right\rangle -  \left|\left\langle\nabla G(\mW) -  \nabla g(\mW),\mW\right\rangle\right| \\
& > \frac{1}{20}\|\mW\mW^\T\|_F^2 - \frac{19}{20}\sqrt{2}c\|\mW\mW^\T\|_F^2\\
& \geq \frac{1}{45}\|\mW\mW^\T\|_F^2,
\end{align*}
where the last line holds when $c\leq \frac{1}{50}$.
Thus,
\[
\|\nabla G(\mW)\|_F \geq \frac{1}{\|\mW\|}\left|\left\langle\nabla G(\mW),\mW\right\rangle\right|> \frac{1}{45}\|\mW\mW^\T\|_F^{3/2},
\]
where we utilize $\|\mW\|\leq \left(\|\mW\mW^\T\|_F\right)^{1/2}$. This completes the proof of \eqref{eq:thm sensing large gradient 3}.

\bibliographystyle{ieeetr}
\bibliography{nonconvex}

\end{document}

%% file: macro.tex
\usepackage{url}
\usepackage[dvipsnames]{xcolor}
\usepackage[colorlinks = true, citecolor=ForestGreen]{hyperref}
\usepackage{amsmath,amsthm,amssymb,amsbsy}
\usepackage{paralist}

\usepackage{color}
\usepackage{algorithm,algpseudocode}
\usepackage{comment}

\usepackage{fancyhdr}
\usepackage{cite}
\usepackage{cleveref}
\usepackage{graphicx}
\graphicspath{{./figs/}}

\newtheorem{thm}{Theorem}
\newtheorem{cor}{Corollary}
\newtheorem{prop}{Proposition}
\newtheorem{defi}{Definition}

\newtheorem{assump}{Assumption}

\newtheorem{lem}{Lemma}
\theoremstyle{remark}
\newtheorem{remark}{Remark}

\renewcommand{\O}{\mathbb{O}}
\newcommand{\R}{\mathbb{R}}

\newcommand{\C}{\mathbb{C}}

\newcommand{\N}{\mathbb{N}}

\newcommand{\vct}[1]{\boldsymbol{#1}}
\newcommand{\mtx}[1]{\boldsymbol{#1}}

\renewcommand{\H}{\mathrm{H}}
\newcommand{\T}{\mathrm{T}}

\newcommand{\Span}{\operatorname{Span}}
\newcommand{\trace}{\operatorname{trace}}

\newcommand{\rank}{\operatorname{rank}}

\newcommand{\dist}{\operatorname{dist}}

\newcommand{\set}[1]{\mathcal{#1}}

\DeclareMathOperator*{\minimize}{\text{minimize}}

\DeclareMathOperator*{\argmin}{\text{arg~min}}
\DeclareMathOperator*{\argmax}{\text{arg~max}}
\def \st {\operatorname*{subject\ to\ }}

\newcommand{\calS}{\mathcal{S}}

\newcommand{\calA}{\mathcal{A}}

\newcommand{\calC}{\mathcal{C}}
\newcommand{\calO}{\mathcal{O}}
\newcommand{\calE}{\mathcal{E}}
\newcommand{\calX}{\mathcal{X}}
\newcommand{\calR}{\mathcal{R}}
\newcommand{\calG}{\mathcal{G}}

\newcommand{\vb}{\vct{b}}

\newcommand{\ve}{\vct{e}}

\newcommand{\vq}{\vct{q}}

\newcommand{\vu}{\vct{u}}
\newcommand{\vv}{\vct{v}}

\newcommand{\vx}{\vct{x}}
\newcommand{\vy}{\vct{y}}
\newcommand{\vz}{\vct{z}}
\newcommand{\valpha}{\vct{\alpha}}

\newcommand{\vphi}{\vct{\phi}}
\newcommand{\vpsi}{\vct{\psi}}

\newcommand{\vzero}{\vct{0}}

\newcommand{\mA}{\mtx{A}}
\newcommand{\mB}{\mtx{B}}
\newcommand{\mC}{\mtx{C}}
\newcommand{\mD}{\mtx{D}}
\newcommand{\mE}{\mtx{E}}

\newcommand{\mH}{\mtx{H}}

\newcommand{\mL}{\mtx{L}}
\newcommand{\mM}{\mtx{M}}

\newcommand{\mP}{\mtx{P}}
\newcommand{\mQ}{\mtx{Q}}
\newcommand{\mR}{\mtx{R}}
\newcommand{\mS}{\mtx{S}}

\newcommand{\mU}{\mtx{U}}
\newcommand{\mV}{\mtx{V}}
\newcommand{\mW}{\mtx{W}}
\newcommand{\mX}{\mtx{X}}

\newcommand{\mZ}{\mtx{Z}}

\newcommand{\mDelta}{\mtx{\Delta}}
\newcommand{\mLambda}{\mtx{\Lambda}}
\newcommand{\mOmega}{\mtx{\Omega}}
\newcommand{\mPhi}{\mtx{\Phi}}
\newcommand{\mPsi}{\mtx{\Psi}}
\newcommand{\mSigma}{\mtx{\Sigma}}
\newcommand{\mUpsilon}{\mtx{\Upsilon}}
\newcommand{\mId}{{\bf I}}

\newcommand{\mzero}{{\bf 0}}

\newcommand{\setO}{\set{O}}

\graphicspath{{./figs/}}

\newlength{\imgwidth}
\newcommand{\revise}[1]{{\color{black}{#1}}}

\newboolean{twoColVersion}
\setboolean{twoColVersion}{false}
\newcommand{\twoCol}[2]{\ifthenelse{\boolean{twoColVersion}} {#1} {#2} }

%% file: LowrankNonsymmetric-Final-Correction.bbl
\begin{thebibliography}{10}

\bibitem{aaronson2007learnability}
S.~Aaronson, ``The learnability of quantum states,'' {\em Proceedings of the
  Royal Society of London A: Mathematical, Physical and Engineering Sciences},
  vol.~463, no.~2088, pp.~3089--3114, 2007.

\bibitem{liu2009interior}
Z.~Liu and L.~Vandenberghe, ``Interior-point method for nuclear norm
  approximation with application to system identification,'' {\em SIAM Journal
  on Matrix Analysis and Applications}, vol.~31, no.~3, pp.~1235--1256, 2009.

\bibitem{srebro2004maximum}
N.~Srebro, J.~Rennie, and T.~S. Jaakkola, ``Maximum-margin matrix
  factorization,'' in {\em Advances in Neural Information Processing Systems},
  pp.~1329--1336, 2004.

\bibitem{xu2016dynamic}
L.~Xu and M.~Davenport, ``Dynamic matrix recovery from incomplete observations
  under an exact low-rank constraint,'' in {\em Advances in Neural Information
  Processing Systems}, pp.~3585--3593, 2016.

\bibitem{davenport2016overview}
M.~A. Davenport and J.~Romberg, ``An overview of low-rank matrix recovery from
  incomplete observations,'' {\em IEEE Journal of Selected Topics in Signal
  Processing}, vol.~10, no.~4, pp.~608--622, 2016.

\bibitem{fazel2004rank}
M.~Fazel, H.~Hindi, and S.~Boyd, ``Rank minimization and applications in system
  theory,'' in {\em American Control Conference}, vol.~4, pp.~3273--3278, IEEE,
  2004.

\bibitem{recht2010guaranteed}
B.~Recht, M.~Fazel, and P.~A. Parrilo, ``Guaranteed minimum-rank solutions of
  linear matrix equations via nuclear norm minimization,'' {\em SIAM Review},
  vol.~52, no.~3, pp.~471--501, 2010.

\bibitem{harchaoui2012large}
Z.~Harchaoui, M.~Douze, M.~Paulin, M.~Dudik, and J.~Malick, ``Large-scale image
  classification with trace-norm regularization,'' in {\em IEEE Conference on
  Computer Vision and Pattern Recognition (CVPR)}, pp.~3386--3393, IEEE, 2012.

\bibitem{candes2009exact}
E.~J. Cand{\`e}s and B.~Recht, ``Exact matrix completion via convex
  optimization,'' {\em Foundations of Computational Mathematics}, vol.~9,
  no.~6, pp.~717--772, 2009.

\bibitem{cai2010singular}
J.-F. Cai, E.~J. Cand{\`e}s, and Z.~Shen, ``A singular value thresholding
  algorithm for matrix completion,'' {\em SIAM Journal on Optimization},
  vol.~20, no.~4, pp.~1956--1982, 2010.

\bibitem{burer2003nonlinear}
S.~Burer and R.~D. Monteiro, ``A nonlinear programming algorithm for solving
  semidefinite programs via low-rank factorization,'' {\em Mathematical
  Programming}, vol.~95, no.~2, pp.~329--357, 2003.

\bibitem{burer2005local}
S.~Burer and R.~D. Monteiro, ``Local minima and convergence in low-rank
  semidefinite programming,'' {\em Mathematical Programming}, vol.~103, no.~3,
  pp.~427--444, 2005.

\bibitem{bhojanapalli2016lowrankrecoveryl}
S.~Bhojanapalli, B.~Neyshabur, and N.~Srebro, ``Global optimality of local
  search for low rank matrix recovery,'' pp.~3873--3881, 2016.

\bibitem{ge2016matrix}
R.~Ge, J.~D. Lee, and T.~Ma, ``Matrix completion has no spurious local
  minimum,'' in {\em Advances in Neural Information Processing Systems},
  pp.~2973--2981, 2016.

\bibitem{park2016non}
D.~Park, A.~Kyrillidis, C.~Carmanis, and S.~Sanghavi, ``Non-square matrix
  sensing without spurious local minima via the {B}urer-{M}onteiro approach,''
  in {\em Artificial Intelligence and Statistics}, pp.~65--74, 2017.

\bibitem{ge2015escaping}
R.~Ge, F.~Huang, C.~Jin, and Y.~Yuan, ``Escaping from saddle points---online
  stochastic gradient for tensor decomposition,'' in {\em Proceedings of The
  28th Conference on Learning Theory}, pp.~797--842, 2015.

\bibitem{sun2015nonconvex}
J.~Sun, Q.~Qu, and J.~Wright, ``When are nonconvex problems not scary?,'' {\em
  arXiv preprint arXiv:1510.06096}, 2015.

\bibitem{lee2016gradient}
J.~D. Lee, M.~Simchowitz, M.~I. Jordan, and B.~Recht, ``Gradient descent
  converges to minimizers,'' {\em University of California, Berkeley},
  vol.~1050, p.~16, 2016.

\bibitem{panageas2016gradient}
I.~Panageas and G.~Piliouras, ``Gradient descent only converges to minimizers:
  Non-isolated critical points and invariant regions,'' in {\em 8th Innovations
  in Theoretical Computer Science Conference (ITCS 2017)}, Schloss
  Dagstuhl-Leibniz-Zentrum fuer Informatik, 2017.

\bibitem{jin2017escape}
C.~Jin, M.~Jordan, R.~Ge, P.~Netrapalli, and S.~Kakade, ``How to escape saddle
  points efficiently,'' in {\em International Conference on Machine Learning},
  pp.~2727--2752, 2017.

\bibitem{li2019alternating}
Q.~Li, Z.~Zhu, and G.~Tang, ``Alternating minimizations converge to
  second-order optimal solutions,'' in {\em International Conference on Machine
  Learning}, pp.~3935--3943, 2019.

\bibitem{sun2016geometric}
J.~Sun, Q.~Qu, and J.~Wright, ``A geometric analysis of phase retrieval,'' {\em
  Foundations of Computational Mathematics}, vol.~18, no.~5, pp.~1131--1198,
  2018.

\bibitem{candes2015Wirtinger}
E.~J. Cand\`{e}s, X.~Li, and M.~Soltanolkotabi, ``Phase retrieval via
  {W}irtinger flow: {T}heory and algorithms,'' {\em IEEE Transactions on
  Information Theory}, vol.~61, no.~4, pp.~1985--2007, 2015.

\bibitem{chen2015solving}
Y.~Chen and E.~Cand\`{e}s, ``Solving random quadratic systems of equations is
  nearly as easy as solving linear systems,'' in {\em Advances in Neural
  Information Processing Systems}, pp.~739--747, 2015.

\bibitem{lee2016fast}
K.~Lee, N.~Tian, and J.~Romberg, ``Fast and guaranteed blind multichannel
  deconvolution under a bilinear system model,'' {\em IEEE Transactions on
  Information Theory}, vol.~64, no.~7, pp.~4792--4818, 2018.

\bibitem{li2016rapid}
X.~Li, S.~Ling, T.~Strohmer, and K.~Wei, ``Rapid, robust, and reliable blind
  deconvolution via nonconvex optimization,'' {\em Applied and computational
  harmonic analysis}, vol.~47, no.~3, pp.~893--934, 2019.

\bibitem{agarwal2014learning}
A.~Agarwal, A.~Anandkumar, P.~Jain, P.~Netrapalli, and R.~Tandon, ``Learning
  sparsely used overcomplete dictionaries.,'' in {\em Conference on Learning
  Theory (COLT)}, pp.~123--137, 2014.

\bibitem{sun2016complete1}
J.~Sun, Q.~Qu, and J.~Wright, ``Complete dictionary recovery over the sphere
  {I}: {O}verview and the geometric picture,'' {\em IEEE Transactions on
  Information Theory}, vol.~63, no.~2, pp.~853--884, 2016.

\bibitem{sun2015complete}
J.~Sun, Q.~Qu, and J.~Wright, ``Complete dictionary recovery over the sphere
  {II}: {R}ecovery by {R}iemannian trust-region method,'' {\em IEEE
  Transactions on Information Theory}, vol.~63, no.~2, pp.~885--914, 2016.

\bibitem{liu2016estimation}
H.~Liu, M.-C. Yue, and A.~Man-Cho~So, ``On the estimation performance and
  convergence rate of the generalized power method for phase synchronization,''
  {\em SIAM Journal on Optimization}, vol.~27, no.~4, pp.~2426--2446, 2017.

\bibitem{tu2015low}
S.~Tu, R.~Boczar, M.~Simchowitz, M.~Soltanolkotabi, and B.~Recht, ``Low-rank
  solutions of linear matrix equations via procrustes flow,'' in {\em
  International Conference on Machine Learning}, pp.~964--973, 2016.

\bibitem{wang2016unified}
L.~Wang, X.~Zhang, and Q.~Gu, ``A unified computational and statistical
  framework for nonconvex low-rank matrix estimation,'' in {\em Artificial
  Intelligence and Statistics}, pp.~981--990, 2017.

\bibitem{sun2016guaranteed}
R.~Sun and Z.-Q. Luo, ``Guaranteed matrix completion via non-convex
  factorization,'' {\em IEEE Transactions on Information Theory}, vol.~62,
  no.~11, pp.~6535--6579, 2016.

\bibitem{jin2016provable}
C.~Jin, S.~M. Kakade, and P.~Netrapalli, ``Provable efficient online matrix
  completion via non-convex stochastic gradient descent,'' in {\em Advances in
  Neural Information Processing Systems}, pp.~4520--4528, 2016.

\bibitem{jain2013low}
P.~Jain, P.~Netrapalli, and S.~Sanghavi, ``Low-rank matrix completion using
  alternating minimization,'' in {\em Proceedings of the Forty-Fifth Annual ACM
  Symposium on Theory of Computing}, pp.~665--674, ACM, 2013.

\bibitem{ge2017no}
R.~Ge, C.~Jin, and Y.~Zheng, ``No spurious local minima in nonconvex low rank
  problems: a unified geometric analysis,'' in {\em Proceedings of the 34th
  International Conference on Machine Learning-Volume 70}, pp.~1233--1242,
  2017.

\bibitem{bhojanapalli2015dropping}
S.~Bhojanapalli, A.~Kyrillidis, and S.~Sanghavi, ``Dropping convexity for
  faster semi-definite optimization,'' in {\em Conference on Learning Theory},
  pp.~530--582, 2016.

\bibitem{zhao2015nonconvex}
T.~Zhao, Z.~Wang, and H.~Liu, ``A nonconvex optimization framework for low rank
  matrix estimation,'' in {\em Advances in Neural Information Processing
  Systems}, pp.~559--567, 2015.

\bibitem{li2016}
Q.~Li, Z.~Zhu, and G.~Tang, ``The non-convex geometry of low-rank matrix
  optimization,'' {\em Information and Inference: A Journal of the IMA},
  vol.~8, no.~1, pp.~51--96, 2018.

\bibitem{zhu2017GlobalOptimality}
Z.~Zhu, Q.~Li, G.~Tang, and M.~B. Wakin, ``Global optimality in low-rank matrix
  optimization,'' {\em IEEE Transactions on Signal Processing}, vol.~66,
  no.~13, pp.~3614--3628, 2018.

\bibitem{conn2000trust}
A.~R. Conn, N.~I. Gould, and P.~L. Toint, {\em Trust region methods}.
\newblock SIAM, 2000.

\bibitem{li2019symmetry}
X.~Li, J.~Lu, R.~Arora, J.~Haupt, H.~Liu, Z.~Wang, and T.~Zhao, ``Symmetry,
  saddle points, and global optimization landscape of nonconvex matrix
  factorization,'' {\em IEEE Transactions on Information Theory}, vol.~65,
  no.~6, pp.~3489--3514, 2019.

\bibitem{chirikjian2016harmonic}
G.~S. Chirikjian and A.~B. Kyatkin, {\em Harmonic Analysis for Engineers and
  Applied Scientists: Updated and Expanded Edition}.
\newblock Courier Dover Publications, 2016.

\bibitem{davenport20141}
M.~A. Davenport, Y.~Plan, E.~van~den Berg, and M.~Wootters, ``1-bit matrix
  completion,'' {\em Information and Inference}, vol.~3, no.~3, pp.~189--223,
  2014.

\bibitem{higham1995matrix}
N.~Higham and P.~Papadimitriou, ``Matrix procrustes problems,'' {\em Rapport
  technique, University of Manchester}, 1995.

\bibitem{candes2011tight}
E.~J. Cand\`{e}s and Y.~Plan, ``Tight oracle inequalities for low-rank matrix
  recovery from a minimal number of noisy random measurements,'' {\em IEEE
  Transactions on Information Theory}, vol.~57, no.~4, pp.~2342--2359, 2011.

\bibitem{udell2014generalized}
M.~Udell, C.~Horn, R.~Zadeh, and S.~Boyd, ``Generalized low rank models,'' {\em
  Foundations and Trends{\textregistered} in Machine Learning}, vol.~9, no.~1,
  pp.~1--118, 2016.

\bibitem{cai2013max}
T.~Cai and W.-X. Zhou, ``A max-norm constrained minimization approach to 1-bit
  matrix completion.,'' {\em Journal of Machine Learning Research}, vol.~14,
  no.~1, pp.~3619--3647, 2013.

\bibitem{salmon2014poisson}
J.~Salmon, Z.~Harmany, C.-A. Deledalle, and R.~Willett, ``Poisson noise
  reduction with non-local {PCA},'' {\em Journal of Mathematical Imaging and
  Vision}, vol.~48, no.~2, pp.~279--294, 2014.

\bibitem{candes2005decoding}
E.~J. Cand\`{e}s and T.~Tao, ``Decoding by linear programming,'' {\em IEEE
  Transactions on Information Theory}, vol.~51, no.~12, pp.~4203--4215, 2005.

\bibitem{horn2012matrix}
R.~A. Horn and C.~R. Johnson, {\em Matrix analysis}.
\newblock Cambridge University Press, 2012.

\bibitem{zhu2018distributed}
Z.~Zhu, Q.~Li, X.~Yang, G.~Tang, and M.~B. Wakin, ``Distributed low-rank matrix
  factorization with exact consensus,'' in {\em Advances in Neural Information
  Processing Systems}, pp.~8422--8432, 2019.

\bibitem{bernstein2009matrix}
D.~S. Bernstein, {\em Matrix mathematics}.
\newblock Princeton university press, 2009.

\end{thebibliography}
